\documentclass[11pt,reqno]{amsart}
\usepackage{amssymb}
\usepackage{amsmath}
\usepackage{amsthm}
\usepackage{amsfonts}
\usepackage{comment}
\usepackage{graphicx}
\usepackage{enumerate}
\usepackage[bottom,flushmargin,hang]{footmisc}
\usepackage{parskip} 
\usepackage{bm}
\usepackage{bbm} 
\usepackage{dsfont}
\usepackage{booktabs}
\usepackage[authoryear,longnamesfirst]{natbib}
\usepackage{multirow}

\usepackage{multicol}
\usepackage{tikz}
\usepackage{xcolor} 
\usepackage{algorithm} 
\usepackage{algpseudocode} 

\numberwithin{equation}{section}
\newtheorem{theorem}{Theorem}

\newtheorem{corollary}{Corollary}

\newtheorem{lemma}{Lemma}
\theoremstyle{definition}
\newtheorem{assumption}{Assumption}
\newtheorem{example}{Example}
\newtheorem{remark}{Remark}

\newcommand{\supp}{\mathrm{supp}}
\renewcommand{\tilde}{\widetilde}

\newcommand{\Gn}{\mathbb{G}_{n}}

\newcommand{\calA}[0]{\mathcal{A}}

\newcommand{\calF}[0]{\mathcal{F}}
\newcommand{\calG}[0]{\mathcal{G}}

\newcommand{\calS}[0]{\mathcal{S}}

\newcommand{\calY}[0]{\mathcal{Y}}

\newcommand{\R}[0]{\mathbb{R}}

\newcommand{\N}{\mathbb{N}}

\newcommand{\1}{\mathbbm 1}

\newcommand{\Z}[0]{\mathbb{Z}}

\renewcommand{\bar}{\overline}

\renewcommand{\i}{{\bm{i}}}
\newcommand{\bN}{{\bm{N}}}
\newcommand{\e}{{\bm{e}}}
\newcommand{\calE}{\mathcal{E}}

\renewcommand{\hat}{\widehat}

\begin{document}
\title{Cross-Fitting-Free Debiased Machine Learning with Multiway Dependence}
\author{Kaicheng Chen}
\address{
School of Economics, Shanghai University of Finance and Economics, 
111 Wuchuan Road, Yangpu District, Shanghai 200434, China}
\email{chenkaicheng@sufe.edu.cn}

\author{Harold D. Chiang}
\address{Department of Economics, University of Wisconsin-Madison, 1180 Observatory
Drive, Madison, WI 53706-1393, USA}
\email{hdchiang@wisc.edu}
\thanks{First arXiv date: 11 February 2026.
We thank Jianfei Cao, Bryan Graham, Ulrich Hounyo, Michael Jansson, Michael Leung, David Ritzwoller, and Tim Vogelsang for their invaluable comments. We also thank seminar participants at Hanyang, LSE, Tilburg, UCL, and conference participants at New York Camp Econometrics XX and the UvA-KU Leuven workshop on panel data econometrics, for valuable discussions. This paper supersedes the earlier manuscript ``Maximal inequalities for separately exchangeable empirical processes'' (arXiv:2502.11432). }

\begin{abstract}
This paper develops an asymptotic theory for two-step debiased machine learning (DML) estimators in generalised method of moments (GMM) models with general multiway clustered dependence, without relying on cross-fitting. Although commonly employed, cross-fitting can be statistically inefficient and computationally burdensome when first-stage learners are complex and the effective sample size is governed by the number of independent clusters. We show that valid inference without sample splitting can be achieved by combining orthogonal moments with a localisation-based empirical process approach, allowing for an arbitrary number of clustering dimensions. The resulting debiased GMM estimators are asymptotically linear and asymptotically normal. A central technical contribution is the derivation of novel global and local maximal inequalities for separately exchangeable empirical processes indexed by general function classes, which underpin our theoretical arguments and are of independent interest. Python and R packages implementing the proposed procedures are available.

\end{abstract}

\maketitle

\section{Introduction}

Debiased machine learning (DML), also known as double/debiased machine learning, has become a leading approach to two-step estimation with high-dimensional or nonparametric nuisance components; see, among many others, \cite{chernozhukov2018double}, \cite{chernozhukov2022}, and \cite{escanciano2022debiased}. A central feature of this literature is cross-fitting, which uses sample splitting to separate nuisance estimation from target parameter evaluation. It is motivated by two considerations: in the conventional view, it mitigates overfitting bias arising from highly flexible first-stage learners; and it relaxes empirical process conditions by reducing the dependence between first-stage estimation errors and second-stage score evaluation. Cross-fitting has thus become close to a default recommendation in both theoretical analyses and empirical implementations of DML.

At the same time, many empirical applications in economics and finance involve multiway clustered dependence, where observations may be correlated along several dimensions, such as firm and time or region and industry; see, for example, \cite{petersen2008estimating,cameron2015practitioner}. \cite{chiang2022multiway} extend DML to such multiway clustered environments. 

However, extensive sample splitting is not without cost. First, with a finite number of folds, cross-fitting introduces additional randomness through random partitions, which may hinder reproducibility; more folds can partly mitigate this concern, but at a substantially greater computational cost, particularly when the first stage is complex or requires extensive tuning.
Second, cross-fitting reduces the sample size available to each first-stage problem. This matters especially for high-dimensional or nonparametric first stages, whose performance is inherently variance-sensitive: the resulting loss of precision in nuisance estimation may translate into non-negligible efficiency losses for the parameter of interest in finite samples.

These concerns are amplified under multiway clustering, where the effective sample size is determined by the number of independent cluster units rather than the total number of observations: partitioning the data into folds may leave each subsample with only a few independent clusters, making cross-fitting particularly costly, besides increasing the computational burden.
Moreover, unlike in classical one-step settings, first-stage overfitting bias in two-step procedures such as DML need not be intrinsically detrimental to inference on the target parameter, so the conventional rationale for cross-fitting may be less central than is sometimes presumed. Indeed, even in i.i.d.\ settings, the literature offers little general theoretical support for the advantages of cross-fitting, apart from some special cases considered, for example, by \citet{belloni2012sparse,newey2018cross}.

A growing literature therefore seeks to weaken or dispense with cross-fitting in general nonlinear estimation problems. One route to theoretical guarantees for DML without cross-fitting is the ``localisation'' method, employed, for example, in \cite*{belloni2015uniform,belloni2018uniformly}. It analyses the supremum of an empirical process indexed by estimating equations evaluated over deterministic sets that localise the nuisance parameter around its true value and contain the nuisance estimator with probability approaching one; this replaces a stochastic index with a deterministic one, disentangling the dependence between the estimating equations and the first-stage estimates. If the localisation sets shrink at appropriate rates---typically verified through maximal inequalities---the associated empirical process remainder is of smaller order than the leading asymptotically linear term and does not affect the limiting distribution. An alternative strategy rests on a ``stability'' condition; see \cite{chen2022debiased} in an i.i.d.\ setting and a generalisation under spatial/network $\beta$-mixing in \cite{cao2025neighborhood}. Verifying such conditions, however, is often substantially more involved beyond certain well-understood cases, so we do not pursue this route.

This paper develops a general asymptotic theory for two-step debiased GMM estimators under multiway clustered dependence, without cross-fitting. We consider a broad class of locally robust two-step GMM problems, similar to those in \cite*{chernozhukov2022}, in which low-dimensional target parameters are identified by orthogonal moment conditions that depend on high-dimensional or nonparametric nuisance components. The data may be dependent along an arbitrary number of clustering dimensions, accommodating correlation that arises simultaneously across, for example, firms, time periods, locations, or networks.
Combining orthogonality with a localisation-based empirical process argument tailored to multiway clustered arrays, we show that the impact of first-stage estimation can be controlled without cross-fitting and establish asymptotic linearity and normality of the resulting estimators under conditions that permit highly flexible first-stage learners. The analysis explicitly accounts for the reduced effective sample size under multiway clustering and provides inference procedures that remain valid when the number of independent cluster units, rather than the total number of observations, governs the stochastic order. This yields a unified framework for debiased GMM inference, well suited to empirically relevant clustered environments where conventional cross-fitting can be statistically and computationally costly.

Maximal inequalities are central to localisation-based arguments in DML, yet for multiway clustered---more precisely, separately exchangeable (SE)---arrays, the available theory remains limited. No general \emph{global} maximal inequality accommodates arbitrary numbers of clustering dimensions \(K\), arbitrary moment orders \(r\), and infinite pointwise measurable function classes: Theorem~B.2 of \cite{chiang2023inference} allows general \(K\) and \(r\in[1,\infty)\) but only finite classes, while Lemma~C.3 of \cite{liu2024estimation} covers potentially uncountable classes only for \(r=1\) and \(K=2\). For \emph{local} maximal inequalities, which are essential for sharper convergence rates, the situation is even more restrictive: none appears to be available for SE arrays beyond the \(K=1\) (i.i.d.) case. These gaps reflect the intrinsic difficulty posed by multiway dependence, which generates complex interactions across observations and undermines classical tools such as Hoeffding averaging. We overcome this with a new proof strategy based on a transversal partition of the index set that effectively decouples the dependence and permits the use of the Hoffmann--J{\o}rgensen inequality, yielding sharp higher-moment bounds. This approach delivers both global and local maximal inequalities for potentially uncountable, pointwise measurable function classes under SE sampling.

The paper is organised as follows. Section~\ref{sec:setup} introduces the multiway clustered sampling setup and notation. Section~\ref{sec:dml_gmm_main} develops the cross-fitting-free debiased GMM estimator, presents the main asymptotic results under general high-level conditions, verifies the rate and complexity conditions in three examples, and establishes the validity of variance estimation. Section~\ref{sec:maximal_inequalities_SE} provides the core technical tools, new maximal inequalities for separately exchangeable empirical processes. Section~\ref{sec:simulation} demonstrates the finite-sample performance and computational advantages of the proposed method relative to existing methods. The appendices collect proofs, supplementary arguments, an extension to cluster-size heterogeneity, and implementation guidance. Python and R packages implementing the proposed procedures are available.

\vskip 0.15in

\section{Setup and Notation}
\label{sec:setup}
Let $K$ be a fixed positive integer and denote a $K$-tuple index by
\(
\bm{i} = (i_1, i_2, \dots, i_K) \in \mathbb{N}^K.
\)
Let $\{X_{\bm i}:\bm i\in\mathbb{N}^K\}$ be an infinite $K$-way array of random elements, of which we observe the finite sub-array
\(
  \{ X_{\bm{i}} : \bm{i} \in [N_1] \times \cdots \times [N_K] \},
\)
where $[N_k] := \{1,\dots,N_k\}$ is the index set of dimension $k$ and $N_k$ the corresponding sample size. Write $\bm{N} = (N_1, \dots, N_K)$, let $[\bm{N}] = \prod_{k=1}^{K} [N_k]$ be the Cartesian product of the $K$ index sets, and
\[
N = \prod_{k=1}^K N_k,\quad n = \min\{N_1, \dots, N_K\},\quad \overline{N} = \max\{N_1,\dots,N_K\}.
\]

Firm--market data illustrate the two-way case $K=2$: $i_1\in[N_1]$ indexes firms, $i_2\in[N_2]$ indexes markets, and $X_{i_1i_2}$ collects the outcome and covariates of firm $i_1$ in market $i_2$. Here $N=N_1N_2$ is the number of cells, whereas $n=\min\{N_1,N_2\}$ counts the independent units along the shorter dimension and hence governs the effective sample size. Dependence runs along both margins: unobserved firm attributes, such as managerial ability, correlate the outcomes of firm $i_1$ across the markets it serves, while market-level shocks, such as local demand conditions, correlate the outcomes of all firms active in market $i_2$, so neither one-way clustering scheme suffices. Other familiar two-way examples are firm and time period in finance panels \citep{petersen2008estimating}, region and industry, worker and firm, and student and school; see \cite{cameron2015practitioner}. Appending a time dimension, for instance, gives $K=3$ with $\bm{i}=(i_1,i_2,i_3)$.

Suppose the infinite array $\{X_{\bm{i}}\}_{\bm i\in\N^K}$ consists of random elements of a measurable space $\calS$, defined on a common probability space $(\Omega,\calA, P)$, and satisfies the following separate exchangeability (SE) and dissociation (D) conditions.
\begin{enumerate}
\item[(SE)]  For any $\pi =(\pi_1,\dots,\pi_K)$, a $K$-tuple of permutations of $\N$, $\{X_{\bm{i}}\}_{\bm i \in \N^K}$ and $\{X_{\pi(\bm{i})}\}_{\bm i \in \N^K}$ are identically distributed.
\item[(D)]   Let $I,I'\subset \mathbb{N}^K$. If $i_k\ne i_k'$ for all $k=1,\dots,K$ and all $\bm i\in I$, $\bm i'\in I'$
, then  $\{X_{\bm{i}}\}_{\i\in I}$ and $\{X_{\bm{i}}\}_{\i\in I'}$ are independent.
\end{enumerate}

These two conditions are standard (see \citealt{menzel2021bootstrap,davezies2021empirical,mackinnon2021wild}) and play complementary roles: (SE) restricts the type of dependence the array may exhibit, and (D) where it may occur. Condition (SE) requires the joint law of $\{X_{\bm i}\}$ to be invariant to relabelling indices separately along each dimension, so labels are anonymous and carry no information beyond cluster membership. In regression settings, the requirement is imposed only conditionally on the observables: systematic differences across units may be absorbed by the covariates, and only the remaining unobserved heterogeneity need be label-anonymous. It is the array analogue of de Finetti exchangeability (see Lemma 2, Section 7 of \citealt{andrews2005cross}) and follows from symmetry restrictions common in economic models, such as exchangeable profit functions in dynamic oligopoly \citep{athey2001investment}---closely related to anonymity in cooperative game theory and social choice \citep{moulin1991axioms}---and primitives exchangeable in rivals' characteristics in differentiated-product markets \citep{berry1995automobile}. Condition (D) permits $X_{\bm i}$ and $X_{\bm i'}$ to be dependent only when $\bm i$ and $\bm i'$ agree in at least one coordinate, so dependence propagates only through shared cluster membership; this is exactly what makes the dependence multiway clustered rather than unrestricted. In the firm--market example, (SE) says that relabelling firms and markets leaves the joint law unchanged, while (D) makes $X_{12}$ and $X_{21}$ independent but allows $X_{12}$ and $X_{13}$ to be dependent through their shared firm.

Under Conditions (SE) and (D), the Aldous--Hoover--Kallenberg (AHK) representation (see Corollary 7.35 in \citealt{kallenberg2005probabilistic}) guarantees the existence of the representation
\begin{align}
	X_{\bm{i}} = \tau\Bigl( \{ U_{\bm{i} \odot \bm{e}} \}_{\bm{e} \in \{0,1\}^K \setminus \{\bm{0}\}} \Bigr), \label{eq:AHK_representation}
\end{align}
where $\odot$ denotes the Hadamard (element-wise) product, the random variables
\(
 \{ U_{\bm{j}} : \bm{j} \in (\N\cup\{0\})^K \setminus \{\bm{0}\} \}
\)---indexed by their distinct labels, each of the form $\bm j=\bm i\odot\bm e$---are i.i.d., and $\tau$ is a Borel measurable map taking values in $\calS$. For $K=2$, \eqref{eq:AHK_representation} reads $$X_{i_1i_2}=\tau(U_{(i_1,0)}, U_{(0,i_2)}, U_{(i_1,i_2)}),$$ an unknown, possibly nonlinear and non-separable, function of mutually independent firm-, market- and match-specific shocks, which is considerably more general than an additive two-way error component model.

\subsection{Cluster-size heterogeneity}
\label{sec:heterogeneous_clusters}
The baseline results treat \(X_{\bm i}\) as a cell-level observation, so that each cell, defined by an intersection of all cluster dimensions, contains a single unit; this keeps the asymptotic arguments at the cell level and matches many panel and network applications. In other applications a cell may contain multiple or zero units, a feature we term cluster-size heterogeneity, also considered in \cite{Carter2017} and \cite{DDG2018}. Let \(M_{\bm i}\) denote the number of units in cell \(\bm i\) and \(\{X_{\bm i m}:1\leq m \leq M_{\bm i}\}\) the corresponding unit-level variables, where \(M_{\bm i}=0\) is permitted and represents unbalanced data. The extension assumes that \(\{(M_{\bm i},\{X_{\bm i m}\}_{1\le m\le M_{\bm i}}):\bm i\in\mathbb N^K\}\) itself satisfies (SE) and (D), so \(M_{\bm i}\) may be random and correlated with the within-cell observations. In the firm--market example, \(M_{\bm i}\) counts the products, establishments or transactions of firm $i_1$ in market $i_2$ and is zero when the firm does not serve that market; matched employer--employee, patient--hospital and student--school data share this structure.

One may either aggregate within cells and then estimate and conduct inference at the cell level, which is valid but discards within-cell information, or estimate from the raw unit-level observations and aggregate the scores within cells for inference. We adopt the second strategy, which uses the data more efficiently and effectively weights the estimates by cell sizes. The main text develops the theory under the baseline setup; Appendix~\ref{app:cluster_heterogeneity} extends the results to the heterogeneous-cluster case with random and possibly zero or unbounded cell sizes.

\subsection*{Notation.}
Let $\N$ be the set of positive integers and $\R$ the real line. For $a,b\in \R$, let $a\vee b=\max\{a,b\}$ and $a\wedge b = \min\{a,b\}$, and for $m\in \mathbb N$ let
$
[m] = \{1,2,\ldots,m\}.
$
For real vectors $\bm{a}= (a_{1},\dots,a_{K})$ and $\bm{b} = (b_{1},\dots,b_{K})$, write $\bm{a} \le \bm{b}$ if $a_{j} \le b_{j}$ for all $1 \le j \le K$, and let $\supp(\bm{a}) = \{ j : a_j \ne 0\}$. For $\i = (i_1,\dots,i_K)$ and $\bm{j} = (j_1,\dots,j_K)$, the Hadamard product is $\i \odot \bm{j}= (i_1 j_1,\dots,i_K j_K)$. For $k=0,1,\dots,K$, let $ \calE_k=\{\e\in \{0,1\}^K: \Vert\e \Vert_0 =k\}$, so that $\{0,1\}^K=\cup_{k=0}^K \calE_k$. For a measure $Q$ on $\calS$ and a measurable $f:\calS\to \R$, write $Qf=\int f dQ$ and, for $r\in [1,\infty]$, $\|f\|_{Q,r}=( Q|f|^r)^{1/r}$. For a non-empty set $T$ and $f:T\to \R$, let $\|f\|_T=\sup_{t\in T}|f(t)|$. For a pseudometric space $(T,d)$, let $N(T,d,\varepsilon)$ denote the $\varepsilon$-covering number of $(T,d)$. We say $F:\calS\to \R_+$ is an envelope for a class of functions $\calF\ni f:\calS\to \R$ if $\sup_{f\in\calF}|f(x)|\le F(x)$ for all $x\in\calS$.

    \vskip 0.15in

\section{Debiased Machine Learning for GMM}
\label{sec:dml_gmm_main}

This section develops a full-sample DML procedure for GMM models with multiway clustered data. For each $n \in \mathbb{N}$, let $P_n$ denote the probability measure of the data. Denote by $X$ an independent copy of $X_{\bm{1}}$, whose dimension may depend on $n$, where $\bm{1} = (1, \ldots, 1)$. By separate exchangeability, $\mathbb{E}[f(X_{\bm{i}})]=\mathbb{E}[f(X)]=P_nf$ for every $\bm{i}\in[\bm{N}]$ and every $f$ with $P_n|f|<\infty$.

We consider a general two-step GMM setup, similar to that of \cite{chernozhukov2022}, in which a finite-dimensional target parameter is identified by moment restrictions involving high-dimensional or nonparametric nuisance components. The procedure has three components:
\begin{enumerate}
    \item The target parameter $\theta_0 \in \Theta \subset  \mathbb{R}^d$ is fixed-dimensional and uniquely identified through moment restrictions: \\
$g: \calS\times\Theta\times {\Gamma}\rightarrow\mathbb{R}^{{q}}$
such that
\begin{equation}
  \label{gmm_moment}
  \mathbb{E}\left[ g(X,\theta_0,\gamma_0) \right] = 0,
\end{equation}
\item The nuisance parameters $\gamma_0\in \Gamma$, where $\Gamma$ may expand with the sample size,\footnote{Our results allow for an expanding nuisance space and class of data generating processes; we suppress their dependence on the sample size whenever this causes no confusion.} are exactly identified and estimated by machine learners using the full sample.
\item The target parameter is then estimated by solving a full-sample GMM problem based on an orthogonalised moment function and a feasible weighting matrix, with the estimated nuisance parameters plugged in.
\end{enumerate}

The key step in our analysis is to show that the empirical process term induced by full-sample nuisance estimation is asymptotically negligible under multiway clustered dependence.

\subsection{Main results for the debiased GMM estimator}

To allow for an expanding nuisance space and class of data generating processes, let $\mathcal{P}_n$ denote, for each $n$, a collection of data generating processes for the $K$-array $\{X_{n,\bm i}:\bm i\in[\bm N_n]\}$; the score map $\psi=\psi_n$, the nuisance space $\Gamma=\Gamma_n$ and the nuisance parameter $\eta_0=\eta_{0,n}$ may then all depend on $n$, and we suppress these subscripts unless needed for clarity. Asymptotic statements in this section are understood uniformly over $\{\mathcal{P}_n\}$ in the sequential sense of \cite{belloni2015uniform}, that is, as holding along every sequence $\{P_n\}$ with $P_n\in\mathcal{P}_n$. Thus $P_n$ always denotes the data generating measure along an arbitrary such sequence, and every expectation, variance, covariance, probability, stochastic order ($O_{P_n}$, $o_{P_n}$) and convergence statement ($\overset{p}{\to}$, $\overset{d}{\to}$) in this section and its proofs is taken under $P_n$, with $\mathbb{E}=\mathbb{E}_{P_n}$ where no subscript is displayed. We regard $\Gamma$ as a convex subset of $L_2(P_n)$ equipped with $\Vert\cdot\Vert_{P_n,2}$. All proofs are stated for arbitrary such sequences, with bounds depending on $P_n$ only through the constants of Assumption~\ref{assu:gmm_reg} and the further uniform constants introduced below, all independent of $n$ and $\{P_n\}$. The maximal inequalities of Section~\ref{sec:maximal_inequalities_SE}, by contrast, are nonasymptotic and stated for a generic fixed probability measure $P$; here they are applied with $P=P_n$ for each $n$.

The first component of DML is the Neyman orthogonality of the moment condition. Let $\psi(X,\theta,\eta)$ be obtained by adding to $g$ an influence function for the effect of the first-step disturbance. Here $\eta_0=(\gamma_0,\alpha_0)$ may stack additional nuisance components $\alpha_0$ beyond $\gamma_0$; with a slight abuse of notation, $\Gamma$ also denotes the (possibly enlarged) nuisance space containing $\eta_0$. Orthogonality restricts how the population moment responds to perturbations of $\eta$ alone, so we define the corresponding map at the true parameter $\theta_0$:
\begin{align}
    m_{P_n}:\Gamma\to\R^{q},\qquad m_{P_n}(\eta):=\mathbb{E}\left[\psi(X,\theta_0,\eta)\right].\label{eq:m_map}
\end{align}
Write $\Gamma-\eta:=\{\tilde\eta-\eta:\tilde\eta\in\Gamma\}$ for the admissible directions at $\eta\in\Gamma$, and let $\mathcal{T}_{n}$ denote the closed linear span of $\Gamma-\eta_0$ in $L_2(P_n)$. Convexity of $\Gamma$ gives $\eta+\tau h\in\Gamma$ for every $h\in\Gamma-\eta$ and $\tau\in[0,1]$, so the perturbed moment $m_{P_n}(\eta+\tau h)$ is well defined for small $\tau>0$. We call $m_{P_n}$ Gateaux differentiable on $\Gamma$ if, at every $\eta\in\Gamma$,
\begin{align}
    D_\eta m_{P_n}(\eta)[h]:=\lim_{\tau\downarrow0}\frac{m_{P_n}(\eta+\tau h)-m_{P_n}(\eta)}{\tau},\qquad h\in\Gamma-\eta,\label{eq:gateaux_def}
\end{align}
exists and is the restriction of a bounded linear map $D_\eta m_{P_n}(\eta):\mathcal{T}_{n}\to\R^{q}$; the common domain $\mathcal{T}_{n}$ makes derivatives at different $\eta$ comparable, and we write $\Vert D_\eta m_{P_n}(\eta)\Vert_{\rm op}:=\sup\{\Vert D_\eta m_{P_n}(\eta)[h]\Vert:h\in\mathcal{T}_{n},\ \Vert h\Vert_{P_n,2}\le1\}$. In this notation, the score has mean zero at the truth and its population moment is flat there in every admissible direction:
\begin{align}
    m_{P_n}(\eta_0)=0,\qquad D_\eta m_{P_n}(\eta_0)[h]=0\ \ \text{for all }h\in\Gamma-\eta_0. \label{eq:orthogonal}
\end{align}
This adjustment offsets the effect of local perturbations of $\gamma$ on the identifying moment condition. 
Being a population property, Neyman orthogonality depends neither on whether the data are i.i.d.\ or multiway clustered nor on whether the nuisance functions are estimated with cross-fitting. Since this property is well established in the literature for certain moment conditions, we take \eqref{eq:orthogonal} as given.

For the observed array $\{X_{\bm{i}}:\i\in[\bm{N}]\}$ and a measurable $f:\calS\to\R$, the sample mean process is
\(\mathbb{E}_N f = N^{-1}\sum_{\bm{i} \in [\bm{N}]} f(X_{\bm{i}}),\)
applied componentwise when $f$ is vector-valued.
Let $\widehat\eta$ be obtained from machine learners appropriate for multiway clustered data\footnote{For example, cluster-LASSO in \cite{belloni2016inference} for one-way clustered data, two-way cluster-LASSO in \cite{chen2025inference} for two-way clustered panels, and multiplier-bootstrap LASSO in \cite{chiang2023inference} for arbitrary $K$-way cluster dependence.}, and let $\hat\psi_N$ denote the empirical average with $\widehat\eta$ plugged in:
\begin{align*}
   \hat\psi_N(\theta) = \mathbb{E}_N[\psi(X,\theta,\widehat\eta)].
\end{align*}
Given a positive semi-definite weighting matrix $\hat\Upsilon$ (e.g., the inverse of a multiway cluster-robust variance-covariance estimator of $\psi(X,\hat\theta^{(0)},\widehat\eta)$ for an initial estimate $\hat\theta^{(0)}$), the debiased GMM estimator of $\theta_0$ is
\begin{align}
    \hat\theta = \arg \min_{\theta\in \Theta} \hat\psi_N'(\theta)  \hat\Upsilon \hat\psi_N(\theta). \label{gmm_estimator}
\end{align}

\begin{assumption}
\label{assu:gmm_reg} For each $n$ large enough, the following conditions hold uniformly over $P_n\in \mathcal P_n$. Let $C_1,C_2,C_3,C_4$ be some finite constants independent of $n$ and $P_n$.
Let $\mathcal{H}_n\subset\Gamma$ be a sequence of nuisance realisation sets (e.g., the range of the first-stage learner $\widehat\eta$), and define the localisation set $\Gamma_n(\eta_0)= \{\eta\in\mathcal{H}_n: \Vert \eta-\eta_0\Vert_{P_n,2} \le C_1 n^{-1/4}\} $.  Let $\mathcal{N}(\theta_0)$ be a fixed convex neighbourhood of $\theta_0$.
\begin{enumerate}[(i)]
    \item $(X_\i)_{\i\in [\bm{N}]}$ satisfy Conditions (SE) and (D).

    \item $\psi(X,\theta_0,\eta)$ is continuous in $\eta$ a.s., and the nuisance moment map $m_{P_n}$ of \eqref{eq:m_map} is Gateaux differentiable on $\Gamma$ in the sense of \eqref{eq:gateaux_def}, with
    $\|D_\eta m_{P_n}(\eta)-D_\eta m_{P_n}(\eta_0)\|_{\rm op}
    \leq C_4\|\eta-\eta_0\|_{P_n,2}$ for every $\eta\in\Gamma$. Moreover, \eqref{eq:orthogonal} is satisfied and, for every $\eta\in\Gamma$, $$\mathbb{E} \|\psi(X,\theta_0,{\eta}) - \psi(X,\theta_0,{\eta}_0)\|^2 \leq C_2\Vert \eta-\eta_0\Vert_{P_n,2}^2.$$

    \item  $\psi(X,\theta,\eta)$ is continuously differentiable in $\theta$, $\partial_\theta\psi(X,\theta_0,\eta)$ is continuous at $\eta_0$ a.s., and $ {\mathbb{E} \left[\sup_{\eta\in\Gamma_n(\eta_0)}\|\partial_\theta\psi(X,\theta_0,\eta)\|^{1+\delta}\right]<C_3}$ for some $\delta>0$.

    \item There exists a non-negative $B(X,\eta)$ such that $\Vert\partial_\theta\psi(X,\theta,\eta) - \partial_\theta\psi(X,\theta_0,\eta) \Vert \leq B(X,\eta) \Vert \theta-\theta_0\Vert^\alpha$, $\forall \theta\in\mathcal{N}(\theta_0)$ with some $\alpha>0$; and $\eta\mapsto B(X,\eta)\ge 0$ is continuous at $\eta_0$ a.s. and $ \mathbb{E}\left[\sup_{\eta\in\Gamma_n(\eta_0)}B(X,\eta)^{1+\delta}\right]<C_3$.

\item  $ \|\widehat\eta - \eta_0\|_{P_n,2}= o_{P_n}\left(n^{-1/4}\right)$ and $P_n(\widehat\eta\in\mathcal{H}_n)\to1$.

\item $\mathbb{E}\Vert  \psi(X,\theta_0,\eta_0)\Vert^{2+\delta}<C_3$ for some $\delta>0$.

\item The smallest singular value of the population Jacobian $J_0 := -\partial_\theta \mathbb{E}[\psi(X,\theta,\eta_0)]\big|_{\theta=\theta_0}$ is bounded below by a positive constant.

\item For every sequence of constants $\delta_n\downarrow 0$,
\[
\sup_{P_n\in\mathcal P_n}\mathbb{E}_{P_n}\left[\omega_n(X;\delta_n)\right]\to 0,\qquad \omega_n(x;\delta):=\sup_{\eta\in\Gamma:\,\|\eta-\eta_0\|_{P_n,2}\le\delta}\|a(x,\eta)-a(x,\eta_0)\|,
\]
with $a(x,\eta)$ equal to each of $\partial_\theta\psi(x,\theta_0,\eta)$ and $B(x,\eta)$.

\end{enumerate}
\end{assumption}

Condition~(i) characterises multiway clustered data through the exchangeability and dissociation conditions standard in the cluster-robust inference literature. Conditions~(ii)--(iv) are score regularity conditions: (ii) holds when the scores come from a likelihood function or moment conditions that are smooth in the nuisance parameters, and (iii) and (iv) are nonlinear counterparts of the linear-in-$\theta$ conditions common in the DML literature. For a score twice differentiable in $\theta$, the existence of the integrable envelope $B(X,\eta)$ reduces to local integrability of the Hessian (see Remark~\ref{remark:B} on the choice of $B(X,\eta)$). Localisation to $\Gamma_n(\eta_0)$ ensures that $\eta$ does not take values that blow up the envelope, e.g., when $\eta$ enters as inverse probability weights. Intersecting with the nuisance realisation set $\mathcal{H}_n$ restricts the localisation ball to functions the first-stage learner can produce; this standard device (cf.\ \citealt{chernozhukov2018double,belloni2018uniformly}) keeps the localised class $\mathcal{F}_n$ of Theorem~\ref{thm:gmm} manageable, since an unrestricted $L_2$-ball around $\eta_0$ would in general admit no finite envelope and no useful entropy bound. Condition~(v) is a high-level condition on the quality of nuisance estimation, together with the requirement that the learner takes values in $\mathcal{H}_n$ with probability approaching one. This requirement can be verified using existing results for a range of machine-learning estimators under multiway clustering; for LASSO, for instance, it follows from Proposition 2 of \citet{chiang2023inference}. Conditions~(vi) and (vii) are standard, mild finite-moment and full-rank conditions. Condition~(viii) is an equicontinuity requirement on the score for genuinely drifting sequences; when the score, the nuisance space and the class of data generating processes do not vary with $n$, it follows from the almost-sure continuity in (iii) and (iv) and the moment bounds therein, by sequential continuity in the metric space $(\Gamma,\|\cdot\|_{P_n,2})$ and dominated convergence (see Lemma~\ref{lemma:4_3} and the remark following it).

\begin{remark}[On $B$ in Assumption~\ref{assu:gmm_reg}(iv)]
\label{remark:B}
    $B(X,\eta)$ in Assumption~\ref{assu:gmm_reg}(iv) is a local Hölder modulus that controls how nonlinear the score is in $\theta$, uniformly over $\eta\in\Gamma_n(\eta_0)$. For the fixed convex neighbourhood $\mathcal{N}(\theta_0)$, it may be defined by
    \begin{align*}
        B(X,\eta) :=\sup_{\substack{\theta\in\mathcal{N}(\theta_0)\\ \theta\ne\theta_0}}\frac{\Vert \partial_\theta\psi(X,\theta,\eta) - \partial_\theta\psi(X,\theta_0,\eta)\Vert}{\Vert \theta-\theta_0\Vert^\alpha}, \quad \alpha>0.
    \end{align*}
    If $\psi(X,\theta,\eta)$ is twice continuously differentiable in $\theta$, the integral identity
    \[
    \partial_\theta\psi(X,\theta,\eta)-\partial_\theta\psi(X,\theta_0,\eta)
    =\left\{\int_0^1\partial_{\theta\theta}\psi(X,\theta_0+t(\theta-\theta_0),\eta)\,dt\right\}(\theta-\theta_0)
    \]
    shows that, with $\alpha=1$, one may take
    $B(X,\eta)=\sup_{\theta\in\mathcal N(\theta_0)}\|\partial_{\theta\theta}\psi(X,\theta,\eta)\|$. If $\psi$ is linear in $\theta$, the natural choice is $B(X,\eta)=0$.

    To verify the integrability of $\sup_{\eta\in\Gamma_n(\eta_0)}B(X,\eta)$, one can often find integrable dominating functions that do not depend on $\eta$. If $B(X,\eta)$ admits a Hölder modulus at $\eta_0$, the desired integrability reduces to that of $B(X,\eta_0)$ and of the Hölder modulus at $\eta_0$.
\hfill $\lozenge$\end{remark}

The results below are stated uniformly over classes of functions, on which we impose two standard conditions. First, every class $\calF$ considered is assumed to be pointwise measurable, i.e.\ to admit a countable subclass $\calF'\subset \calF$ such that each $f\in \calF$ is the pointwise limit of a sequence $(f_j)_j\subset \calF'$. Second, following Section 3.6 in \cite{gine2016mathematical}, a function class $\calF$ on $\calS$ with a measurable envelope $F$ is called Vapnik–Chervonenkis-type (VC-type) with characteristics $(A,v)$ if
\begin{align}
	\sup_Q N(\calF,\|\cdot\|_{Q,2},\varepsilon\|F\|_{Q,2})\le \left(\frac{A}{\varepsilon}\right)^v\: \text{ for all } 0<\varepsilon\le 1,\label{eq:VC-type-def}
\end{align}
where the supremum is taken over all finite discrete distributions. A wide range of commonly used models and estimators in econometrics, machine learning and statistics gives rise to sequences of function classes satisfying this condition; see Section~\ref{sec:examples} for examples.

Our first main result establishes asymptotic linearity and normality of a generic debiased GMM estimator without cross-fitting.
\begin{theorem}[Asymptotic linearity and normality]
\label{thm:gmm}
Suppose Assumption~\ref{assu:gmm_reg} holds, $\hat{\theta}\overset{p}{\to} \theta_0 \in {\rm int}(\Theta)${\footnote{Primitive conditions for the consistency of $\hat\theta$ follow from standard extremum-estimator arguments (e.g., Theorem 2.1 of \citealt{NeweyMc1994}) combined with a uniform-in-$\theta$ law of large numbers, which in the present setting can be obtained from Theorem~\ref{theorem:global_maximal_ineq} applied to $\{\psi(\cdot,\theta,\eta):\theta\in\Theta,\,\eta\in\Gamma_n(\eta_0)\}$; we omit the details.}}, and $\widehat\Upsilon\overset{p}{\to} \Upsilon$ as $n\to\infty$ for some positive-definite limit $\Upsilon$ whose eigenvalues are bounded above and away from zero. For $j\in[q]$, write $f_j(\eta)=\psi_j(X,\theta_0,\eta)-\psi_j(X,\theta_0,\eta_0)$ and define the scalar class
\[
\mathcal F_n:=\{f_j(\eta)-\mathbb E_{P_n}[f_j(\eta)]:\eta\in\Gamma_n(\eta_0),\ j\in[q]\}.
\]
Let $F_n\geq1$ be a measurable envelope for $\mathcal F_n$. Suppose that, for some $r> 2$ and for each $n$, $\|F_n\|_{P_n,r}<\infty$\footnote{This condition need not hold uniformly in $n$; in particular, $\|F_n\|_{P_n,r}\to\infty$ is allowed.}, that $n^{1/4}\|F_n\|_{P_n,2}\le A_n\vee\overline N$,
 and that
\begin{align}
    \text{$\mathcal{F}_{n}$ is a VC-type class with characteristics $A_n\ge {(e^{2(K-1)}/16)\vee e}$ and $v_n\ge1$}.\label{eq:VC-type-rate}
\end{align}
Then (i) the following asymptotic linear representation holds uniformly over $P_n\in\mathcal P_n$ as $n\to\infty$:
\begin{align*}
        \sqrt{n}(\hat{\theta}-\theta_0) = &\left( J_0' \Upsilon  J_0\right)^{-1} J_0' \Upsilon \sqrt{n} \mathbb{E}_N [\psi(X,\theta_0,{\eta}_0)] + \sum_{k=1}^K\sum_{\bm{e}\in \mathcal{E}_k} O_{P_n}(\rho_{n,k}) +o_{P_n}(1) \\
        \text{where }\ \rho_{n,k}=&\left(\frac{v_n\log(A_n\vee \overline{N} )}{ {n^{1-1/(2k)}}}\right)^{k/2} \vee \left(\frac{\left\|  F_n  \right\|_{P_n,r} \{v_n\log(A_n\vee \overline{N} )\}}{n^{1/2-1/r}}\right)^{k}.
\end{align*}
(ii) If, additionally, (a) for each $k=1,\ldots, K$,
\begin{align}
 \frac{v_n\log(A_n\vee \overline{N} )}{ {n^{1-1/(2k)}}}\vee \frac{\left\|  F_n  \right\|_{P_n,r} \{v_n\log(A_n\vee \overline{N} )\}}{n^{1/2-1/r}}=o(1),\label{eq:rate_condition}
\end{align}
and (b)  the smallest eigenvalue of $\Psi_{0,n}$, defined below, is bounded from below by some positive constant, then it holds uniformly over $P_n\in\mathcal P_n$ as $n\to\infty$ that
    \begin{align*}
        V_{0,n}^{-1/2} \sqrt{n}(\hat{\theta}-\theta_0)\overset{d}{\to} N(0,I_d),
    \end{align*}
    where $V_{0,n} :=\left(J_0' \Upsilon  J_0\right)^{-1} J_0' \Upsilon \Psi_{0,n} \Upsilon J_0 \left(J_0' \Upsilon  J_0\right)^{-1}$ and
    \begin{align*}
        \Psi_{0,n} :
        =\mu_1 {\rm Var}( \mathbb{E}[\psi(X,\theta_0,\eta_0)|U_{1,0,...,0}]) +...+ \mu_K {\rm Var}( \mathbb{E}[\psi(X,\theta_0,\eta_0)|U_{0,0,...,1}]),
    \end{align*}
    where $\mu_k =\lim_{n\to\infty} \frac{n}{N_k}$ for $k=1,...,K$ (the limits are assumed to exist) and $\{U_{\bm i}\}_{\bm i\ge \bm 0,\,\bm i\ne \bm 0}$ are as defined in \eqref{eq:AHK_representation}.

\end{theorem}
The proof is given in Section~\ref{sec:Proof of Theorem gmm} of the appendix.
Condition (b) in part (ii) is a non-degeneracy requirement: the first-order projections associated with the clustering dimensions must jointly have nonzero variance in every direction $b\in\R^q$ of the score. More precisely, at least one clustering dimension $k$ with $\mu_k>0$ must have $b'{\rm Var}(\mathbb{E}[\psi(X,\theta_0,\eta_0)|U_{\bm e_k}])b$ bounded away from zero; dimensions with $\mu_k=0$ do not contribute to $\Psi_{0,n}$. The condition is mild, as non-degeneracy of a single such dimension (in every direction) suffices. Related conditions are also imposed in, e.g., \cite{davezies2021empirical}, \cite{chiang2022multiway} and \cite{chiang2023inference}.
 When $K=1$, the setup is i.i.d.\ and $\Psi_{0,n}={\rm Var}_{P_n}(\psi)$, so non-degeneracy may hold in the usual way. By contrast, if cell observations in a $K\geq2$ array are i.i.d., all first-order cluster projections vanish under the $\sqrt n$ normalisation, and inference belongs to the conventional $\sqrt N$ regime (\citealp{belloni2015uniform}, among others) rather than the non-degenerate multiway-cluster regime studied here.

 The expansion below, and much of the subsequent analysis, involves the empirical process associated with $\mathbb{E}_N$, defined for a measurable $f:\calS\to\R$ with $P_n|f|<\infty$ by
\begin{align}
	\Gn(f) = \frac{\sqrt{n}}{N} \sum_{\bm{i} \in [\bm{N}]} \Bigl\{ f(X_{\bm{i}}) - P_nf \Bigr\}, \label{eq:empirical_process}
\end{align}
which centres the sample mean process at $P_nf$ and scales it by the effective sample size $\sqrt{n}$ rather than $\sqrt{N}$. 

 For intuition, the first-order condition implies the expansion
\begin{align}
\sqrt{n}(\hat \theta - \theta_0)
\;\approx\;M\left\{\sqrt{n}\,
\mathbb{E}_N[\psi(X,\theta_0,\eta_0)]
\,+\, \sqrt{n}\,\mathbb{E}[f(\eta)]_{\eta=\hat \eta}\,
+\, \Gn\bigl(f(\hat \eta)\bigr)\right\}, \label{eq:expansion_intuition}
\end{align}
where \(M:=\left( J_0' \Upsilon  J_0\right)^{-1} J_0' \Upsilon\) and
\(
f(\eta)=\psi(X,\theta_0,\eta)-\psi(X,\theta_0,\eta_0).
\)
Such decompositions are standard in semiparametric theory; see, e.g., \citet{andrews1994asymptotics}.

The first term is asymptotically normal, and the second is negligible by the orthogonality condition \eqref{eq:orthogonal}. The main technical challenge is therefore to control the localised empirical process
\(\Gn\bigl(f(\hat \eta)\bigr),\)
which is non-standard because \(\hat \eta\) depends on the full sample.
A standard approach is cross-fitting, which removes this dependence: conditional on \(\hat \eta\), a Hoeffding-type decomposition and Markov's inequality give, under suitable smoothness conditions, with probability $1-o(1)$
\[
\Gn\bigl(f(\hat \eta)\bigr)
\;\lesssim\;
\|\hat \eta - \eta_0\|^{v}_{P_n,2}
\quad \text{for some } v>0.
\]

An alternative is a localisation argument. Suppose there exists a sequence of shrinking nuisance neighbourhoods \(\{\Gamma_n(\eta_0)\}\) such that $P_n(\hat \eta \in \Gamma_n(\eta_0)) \to 1$. Then, with probability \(1-o(1)\),
\[
\bigl|\Gn(f(\hat \eta))\bigr|
\;\le\;
\sup_{\eta\in\Gamma_n(\eta_0)}
\bigl|\Gn(f(\eta))\bigr|,
\]
which removes the stochastic dependence on \(\hat \eta\). The neighbourhoods induce deterministic function classes \(\{f(\eta):\eta\in\Gamma_n(\eta_0)\}\) and can often be constructed tightly from the convergence rate of $\|\widehat\eta - \eta_0\|_{P_n,2}$, which is typically available because the same rate is also needed to verify the orthogonality condition, with or without cross-fitting.

In the classical semiparametric literature, the function class \(\mathcal{F}\) is typically fixed and stochastic equicontinuity follows from standard functional CLT arguments. With machine-learning first stages, by contrast, a fixed \(\mathcal{F}\) is generally too large to control, necessitating shrinking (localised) classes.
This localisation strategy underlies the i.i.d.\ analyses of \citet{belloni2015uniform}, further generalised in \citet{belloni2018uniformly}, which rely on maximal inequalities of \citet{CCK2014AoS} to control the supremum. As comparable results are unavailable under multiway clustering, we develop the required global and local maximal inequalities in Section~\ref{sec:maximal_inequalities_SE}.

\begin{remark}[Choosing the envelopes $F_n$]
   For the coordinatewise class in Theorem~\ref{thm:gmm}, one may take $
    F_n(X)= {1\vee}\sup_{\substack{\eta\in\Gamma_n(\eta_0)\\j\in[q]}}
    \left|f_j(\eta)-\mathbb E_{P_n}[f_j(\eta)]\right| $.
    Its $L_r(P_n)$ integrability can be ensured by, for example, a local Hölder modulus $L(X)$ such that, for some $\alpha>0$,
    \[
    \|\psi(X,\theta_0,\eta)-\psi(X,\theta_0,\eta_0)\|
    \leq L(X)\|\eta-\eta_0\|_{P_n,2}^{\alpha}
    \quad\text{for all }\eta\in\Gamma_n(\eta_0),
    \]
    with $\|L\|_{P_n,r}<\infty$ and $\|\psi(X,\theta_0,\eta_0)\|_{P_n,r}<\infty$.
\hfill $\lozenge$\end{remark}

\begin{remark}[Cluster-size heterogeneity]
    \label{rem:cluster_heterogeneity_normal}
    Theorem~\ref{thm:gmm} is stated for the baseline setup with one observation per cell. For a modified estimator, the results extend to cluster-size heterogeneity with zero or multiple observations per cell: the stochastic order remains governed by the smallest cluster dimension $n$, while heterogeneous cell sizes affect the moment condition, envelope and asymptotic variance through the cell-aggregated score. The appendix gives the modified estimator and the formal results.
\end{remark}

\begin{remark}[Alternative asymptotics for semiparametric estimation]\label{rem:intuition}
Making the empirical process component of the asymptotic expansion negligible is not the only route to asymptotic linearity and normality in DML and related two-step estimation problems. An alternative arises when the nuisance parameter $\eta$ itself admits an asymptotically linear representation and satisfies a central limit theorem, a canonical example being the density-weighted average derivative estimator of \citet{powell1989semiparametric}. In such settings, valid inference can instead be conducted under the small-bandwidth asymptotics developed by \citet{cattaneo2014small}, under which the empirical process term $\Gn(f(\widehat\eta))$ need not vanish asymptotically: its limiting distribution can be characterised explicitly using a CLT for quadratic forms\footnote{See, e.g., \cite{de1987central}.} and may be of the same order as, or even dominate, the usual asymptotically linear component. This framework is particularly relevant when the nuisance parameter is estimated by kernel-based methods. Notably, (leave-one-out) cross-fitting continues to play a role in this literature.
\hfill $\lozenge$\end{remark}

\begin{remark}[Multiway-clustering stability]
    Alternatively, one may pursue asymptotic normality of the debiased GMM estimator without a maximal inequality, by directly controlling the empirical process component $\Gn(f(\widehat\eta))$ appearing in Remark~\ref{rem:intuition} through a multiway-clustering stability condition analogous to that of \cite*{chen2022debiased}.
Establishing the asymptotic negligibility of this term under such conditions, and verifying them for concrete first-stage learners under multiway dependence, require arguments beyond the scope of this paper; we leave this route for future research.
\hfill $\lozenge$\end{remark}

 \subsection{Verification of complexity and rate conditions: three examples}\label{sec:examples}
Under Assumption~\ref{assu:gmm_reg}, to apply Theorem~\ref{thm:gmm} it suffices to verify the high-level VC-type condition \eqref{eq:VC-type-rate} and the rate condition \eqref{eq:rate_condition}. These abstract conditions are not entirely straightforward to verify in general, so we give three examples of machine learning estimators of the first-stage nuisance function $\eta$ that satisfy both. To isolate the role of the first-stage learner, we treat $\psi(\cdot,\eta)$ as a map in $\eta$ that preserves the VC-type properties of the function class $\mathcal{G}_n$ to which $\eta$ belongs. For instance, $\psi(\cdot,\eta)$ may be a monotone or Lipschitz transformation of $\eta$, or a finite combination such as sums, products, minima or maxima; see Section 3.6 of \citet{gine2016mathematical}.

\begin{theorem}[Complexity/rate for machine learners]
\label{eg:vctype}
    Suppose the nuisance realisation set of Assumption~\ref{assu:gmm_reg} is taken as $\mathcal{H}_n=\mathcal{G}_n$ (so that $\Gamma_n(\eta_0)\subset\mathcal{G}_n$), $\psi(\cdot,\eta)$ preserves the VC-type characteristics of $\mathcal{G}_n $, and $\Vert F_n\Vert_{P_n,r} = {O(1)}$ for some $r>4$.\footnote{Here, ``preserves the VC-type characteristics'' means that the localised, centred class $\mathcal F_n$ of Theorem~\ref{thm:gmm} is VC-type with the same characteristics as $\mathcal G_n$ up to universal constants. This holds for the transformations listed above because localisation to $\Gamma_n(\eta_0)\subset\mathcal{H}_n=\mathcal{G}_n$ passes to a subclass, and subtracting the fixed function $\psi(\cdot,\theta_0,\eta_0)$ and the constants $\mathbb E_{P_n}[f_j(\eta)]$ alters $(A,v)$ only by universal constants, with the envelope adjusted accordingly (see Section 3.6 of \citealt{gine2016mathematical}).} Since $\|F_n\|_{P_n,2}=O(1)$ and $n\le\overline N$, the envelope-growth condition $n^{1/4}\|F_n\|_{P_n,2}\le A_n\vee\overline N$ of Theorem~\ref{thm:gmm} holds for all $n$ sufficiently large. Then, in each of the following three cases, Condition~\eqref{eq:VC-type-rate} holds with the stated characteristics, and Condition~\eqref{eq:rate_condition} holds under the additional rate restriction stated in that case:
    \begin{enumerate}
        \item (Generalised linear models with $\ell^1$-regularisation) If, for a monotone link $\Lambda$, $$\mathcal{G}_n = \{x\mapsto \Lambda(x^\top\beta): \beta\in \mathbb{R}^p,\Vert\beta \Vert_0<s, \Vert\beta \Vert_2{\le B_n \text{ for some finite } B_n} \},$$ then $\mathcal{G}_n$ is VC-type with $v_n= {2(s+1)}$ and $A_n = \frac{C \ ep}{s}\vee {((e^{2(K-1)}/16)\vee e)}$ for some universal constant $C<\infty$. Assume further that ${s\log(p/s)} \vee {s\log(\overline{N})} = o(n^{1/4})$.

        \item (Regression trees) Let $\{R_l\}_{l=1}^L$ denote random partitions of a regression tree with axis-aligned threshold splits on features in $\mathbb{R}^p$ and a constant output on each leaf. The corresponding function class is $$\mathcal{G}_n = \left\{x\mapsto g(x): g(x) = \sum_{l=1}^L \mu_{l} 1\{x\in R_l\},\ \bigcup_{l=1}^LR_l = \mathbb{R}^p,  \ |\mu_l|{\le M_n \text{ for some finite } M_n}\right\}.$$ Then $\mathcal{G}_n$ is a VC-subgraph class with pseudo VC-dimension of order $O(L\log(Lp))$. For some $C<\infty$, $v_n=2CL\log(2Lp)$ and $A_n= C\vee{((e^{2(K-1)}/16)\vee e)}$. Assume further that $ {L\log(2Lp)\log(A_n \vee \overline{N})} = o(n^{1/4})$.

        \item (Deep neural networks) Let $\mathcal{G}_n$ be the class of functions generated by a feed-forward neural network of fixed structure with the ReLU activation function, $p$ features, $L-2$ hidden layers, $U$ total hidden units and $W-1$ total parameters. Then $\mathcal{G}_n$ is a VC-subgraph class with pseudo VC-dimension of order $O(LW\log(pU))$. For some $C<\infty$, $v_n=2CLW\log(pU)$ and $A_n= C\vee{((e^{2(K-1)}/16)\vee e)}$. Assume further that  ${LW\log(pU)\log(A_n \vee \overline{N})} = o(n^{1/4})$.

    \end{enumerate}    
\end{theorem}

The proof is given in Section~\ref{sec:Proof of Theorem vctype} of the appendix. Case (1) covers sparse generalised linear models, including sparse linear, logit and exponential models, and corresponds to LASSO-type ($\ell^1$-penalty) machine learners. For basic definitions and textbook treatments of all three methods, see, e.g., \cite{chernozhukov2024applied}.

\subsection{Variance estimation}
\label{sec:variance}
For hypothesis testing based on Theorem~\ref{thm:gmm}, a missing piece is the unknown asymptotic variance $V_{0,n}$. We propose a full-sample variance estimator that accounts for (1) multiway clustering dependence and (2) estimation errors from both the high-dimensional nuisance estimation and the GMM estimation. To account for (1), we follow the multiway cluster-robust variance estimator of \cite{DDG2018}, but with the empirical scores replaced by the Neyman orthogonalised scores evaluated at the estimated nuisance parameters.

For any $\bm i, \bm j \in [\bm N]$ with $k$-th elements $i_k$ and $j_k$, let $\1_{k}\{\bm i, \bm j\}$ indicate whether the two observations $\bm i, \bm j$ share the same cluster in the $k$-th dimension, i.e.,
$\1_{k}\{\bm i, \bm j\} = \1\{ i_k = j_k \}.$
We estimate the middle term $\Psi_{0,n}$ by
\begin{align*}
    \hat{\Psi}_N(\hat\theta) = & \sum_{ k=1 }^K  \hat\Psi_{N,k}(\hat\theta), \\
    \hat\Psi_{N,k}(\hat\theta)= & \frac{n}{N^2} \sum_{\bm i,\bm j\in[\bm N]} \psi(X_{\bm i},\hat\theta,\widehat\eta)  \psi(X_{\bm j},\hat\theta,\widehat\eta)' \1_{k}\{\bm i, \bm j\}.  
\end{align*}
With the empirical Jacobian $\hat J_N(\hat\theta) := -\partial_\theta\hat\psi_N(\theta)\big|_{\theta=\hat\theta}$, the estimator of $V_{0,n}$ is
\begin{align}
\label{eq:var_est}
    \hat{V}=\left(\hat{J}_N(\hat\theta)' \hat{\Upsilon} \hat{J}_N(\hat\theta)\right)^{-1} \hat{J}_N(\hat\theta)' \hat{\Upsilon} \hat{\Psi}_N(\hat\theta) \hat{\Upsilon} \hat{J}_N(\hat\theta) \left(\hat{J}_N(\hat\theta)' \hat{\Upsilon}  \hat{J}_N(\hat\theta)\right)^{-1}.
\end{align}

\begin{remark}[Variance estimation under cluster-size heterogeneity]
    \label{rem:cluster_heterogeneity_var}
    The estimator \eqref{eq:var_est} is stated for the baseline case with one observation per cell. Under cluster-size heterogeneity, the same construction applies with the scores replaced by the cell-aggregated scores and \(N\) by the total number of unit observations \(\sum_{\bm i\in[\bm N]}M_{\bm i}\); Section~\ref{app:cluster_heterogeneity} of the appendix gives the formal expression and consistency result.
\end{remark}

\begin{remark}[Full-sample implementation]
    Full-sample DML estimation and multiway cluster-robust inference are straightforward to implement using the formulas above. 
    Appendix~\ref{app:implementation} presents a detailed algorithm, with companion packages and examples illustrating the implementation in software.
\end{remark}


\begin{theorem}[Consistent variance estimation] 
\label{thm:var}
    Suppose that Assumption~\ref{assu:gmm_reg} holds, $\hat\theta\overset{p}{\to}\theta_0\in{\rm int}(\Theta)$, $\widehat\Upsilon\overset{p}{\to}\Upsilon$ for some positive-definite limit $\Upsilon$ whose eigenvalues are bounded above and away from zero, and, for some finite constant $C_5$ independent of $n$ and $P_n$,
    \begin{align}
        & \mathbb{E}\left[ \sup_{\eta\in \Gamma_n(\eta_0)} \left\Vert \psi(X,\theta_0,\eta)\right\Vert^{2+\delta}\right]< {C_5}, \label{eq:var_moment1} \\
        &\mathbb{E}\left[\sup_{\eta\in \Gamma_n(\eta_0)} \left\Vert \partial_\theta\psi(X,\theta_0,\eta) \right\Vert^{2+\delta}\right]< {C_5} \label{eq:var_moment2},  \\
    &\mathbb{E}\left[\sup_{\eta\in\Gamma_n(\eta_0)}[B(X,\eta)]^{2+\delta}\right]<{C_5} \label{eq:var_moment3}   , 
    \end{align}
   where $B(X,\eta)$ is defined in Assumption~\ref{assu:gmm_reg}(iv). Suppose that Assumption~\ref{assu:gmm_reg}(viii) also holds with $a(x,\eta)$ equal to each of $\|\psi(x,\theta_0,\eta)-\psi(x,\theta_0,\eta_0)\|^2$, $\|\partial_\theta\psi(x,\theta_0,\eta)\|^2$, and $B(x,\eta)^2$. Then $\|\hat{V}-V_{0,n}\|\overset{p}{\to}0$ as $n\to\infty$, where $V_{0,n}$ is as defined in Theorem~\ref{thm:gmm}.
\end{theorem}

The proof is given in Section~\ref{sec:Proof of Theorem var} of the appendix. Theorem~\ref{thm:var} requires none of the VC-type condition \eqref{eq:VC-type-rate}, the rate condition \eqref{eq:rate_condition} or the non-degeneracy condition (b) of Theorem~\ref{thm:gmm}; the last is needed only when $\hat V$ is used for Studentised inference based on the normal approximation of Theorem~\ref{thm:gmm}(ii). The extra moment conditions mildly strengthen those in Theorem~\ref{thm:gmm}. As in Theorem~\ref{thm:gmm}, these local integrability conditions can be delivered by integrability of the score, Jacobian and Hessian, given enough smoothness in $\eta$.

$\hat{V}$ is positive semi-definite by construction because each $\hat\Psi_{N,k}(\hat\theta)$ is; this is easily seen for $K=2$, where each $\hat\Psi_{N,k}(\hat\theta)$ reduces to a one-way cluster variance estimator. A caveat is that, for a $K\geq2$ array whose cells are i.i.d., in which case the first-order cluster projections vanish and the non-degenerate $\sqrt n$ theory does not apply, this variance estimator would overestimate the asymptotic variance $V_{0,n}$ and result in a conservative test, a well-known issue in the cluster-robust inference literature; see, e.g., \cite{mackinnon2021wild}. A potential fix is to remove the double-counting terms in $\hat\Psi_{N}(\hat\theta)$ by defining
\begin{align*}
\tilde{\Psi}_N(\hat\theta) = & \sum_{k=1}^K\sum_{\bm e \in  \calE_k }  (-1)^{k+1}\tilde\Psi_{N,\bm{e}}(\hat\theta), \\
    \tilde\Psi_{N,\bm{e}}(\hat\theta) =& \frac{n}{N^2} \sum_{\bm{i,j\in [N]}} \psi(X_{\bm i},\hat\theta,\widehat\eta)  \psi(X_{\bm j},\hat\theta,\widehat\eta)' I_{\bm e}\{\bm i, \bm j\},  
\end{align*}
where $I_{\bm e}\{\bm i, \bm j\}= \1\{\bm{e}\odot \bm{i} = \bm{e}\odot \bm{j}\} $ indicates whether the two observations share the same clusters over the whole support of $\bm e$. This is a DML version of the $K$-way generalisation of the variance estimator proposed by \cite{Cameron2011}, referred to as the CGM estimator.\footnote{Other candidate inference procedures include the modified multiway empirical likelihood and the modified multiway jackknife variance estimator of \cite{chiang2024multiway}; though computationally more demanding, these may deliver improved higher-order asymptotic properties (see Theorem 3 therein).} In a parametric setting, \cite*{DDG2018} show that both variance estimators are consistent for the asymptotic variance under non-degeneracy. However, the CGM estimator involves more terms to compute and is not guaranteed to be positive semi-definite. In practice, multi-dimensional data are rarely i.i.d., owing to the common presence of unobserved heterogeneous effects.

In some degenerate yet cluster-dependent scenarios, such as those studied by \cite{menzel2021bootstrap}, the estimator itself may fail to be asymptotically normal, in which case standard inference procedures, including CGM-type variance estimators and the approach proposed here, become invalid. For such cases, \citet{menzel2021bootstrap} proposes uniformly valid bootstrap-based inference methods, which, however, require choosing tuning parameters and are typically conservative.
In the two-way clustering setting, \cite{davezies2025analytic} develop a simple analytical inference method that remains valid under non-Gaussian degeneracy and requires no tuning parameters; complementarily, \cite{hounyo2025projection} propose bootstrap procedures that adapt across a range of non-degenerate and (Gaussian) degenerate cases.
Extending these approaches to our setting is substantially more involved because of the two-step estimators and machine learning-based first stages. In particular, degeneracy changes the effective stochastic order of the leading term, tightening the rate requirements on the first-stage estimators needed for their estimation error to be asymptotically negligible. Moreover, incorporating the strategy of \citet{davezies2025analytic} is highly non-trivial in our framework, as its conditioning arguments further complicate the first-stage convergence requirements in the presence of multiway dependence and generated regressors. Addressing these challenges would require new techniques, which we leave for future research.

\vskip 0.15in

\section{Maximal Inequalities under Separate Exchangeability}
\label{sec:maximal_inequalities_SE}
Let $(X_\i)_{\i\in [\bm{N}]}$ satisfy Conditions (SE) and (D). This section bounds the $r$-th moment of the supremum, over a pointwise measurable class $\calF$, of the empirical process $\Gn$ in \eqref{eq:empirical_process},
	\(
 \mathbb{E}\bigl[\|\Gn\|_\calF^r \bigr] ,
	\)
	for some $r \in [1,\infty)$.
Throughout this section, assume without loss of generality that $Pf = 0$ for all $f \in \calF$.
 We first introduce the Hoeffding-type decomposition of \cite{chiang2023inference}. For \(\bm{i}\in [\bm{N}]\) and $\bm{e}\in \{0,1\}^{K} \setminus \{\bm{0}\}$, define
	\begin{align*}
		(P_{\bm{e}}f)\Bigl(\{U_{\bm{i}\odot \bm{e}'}\}_{\bm{e}'\le \bm{e}}\Bigr)
		&=\mathbb{E}\Bigl[f(X_{\bm{i}})\,\Big|\,\{U_{\bm{i}\odot \bm{e}'}\}_{\bm{e}'\le \bm{e}}\Bigr]. 
	\end{align*}
	Let $\bm{e}_k\in\mathcal{E}_1$ be the $K$-dimensional vector whose entries are all zero except the $k$-th. We then define recursively, for \(k=1,2,\dots, K\),
	\begin{align*}
		(\pi_{\bm{e}_k}f)(U_{\bm{i}\odot \bm{e}_k})
		&=(P_{\bm{e}_k}f)(U_{\bm{i}\odot \bm{e}_k}),
	\end{align*}
	and, for \(\bm{e}\in \bigcup_{k=2}^K\calE_k\),
	\begin{align*}
		(\pi_{\bm{e}}f)\Bigl(\{U_{\bm{i}\odot \bm{e}'}\}_{\bm{e}'\le \bm{e}}\Bigr)
		=\;& (P_{\bm{e}}f)\Bigl(\{U_{\bm{i}\odot \bm{e}'}\}_{\bm{e}'\le \bm{e}}\Bigr) \nonumber\\
		&\quad-\sum_{\substack{\bm{e}'\le \bm{e}\\ \bm{e}'\ne \bm{e}}}
		(\pi_{\bm{e}'}f)\Bigl(\{U_{\bm{i}\odot \bm{e}''}\}_{\bm{e}''\le \bm{e}'}\Bigr). 
	\end{align*}
	By the AHK representation \eqref{eq:AHK_representation}, for fixed \(\bm{e}\), the distributions of
	\[
	(P_{\bm{e}}f)\Bigl(\{U_{\bm{i}\odot \bm{e}'}\}_{\bm{e}'\le \bm{e}}\Bigr)
	\quad\text{and}\quad
	(\pi_{\bm{e}}f)\Bigl(\{U_{\bm{i}\odot \bm{e}'}\}_{\bm{e}'\le \bm{e}}\Bigr)
	\]
	do not depend on \(\bm{i}\), so we write \(P_{\bm{e}}f\) and \(\pi_{\bm{e}}f\) for a generic \(\bm{i}\).

To illustrate, consider the projection operators under two-way and three-way clustering.
\begin{example}[Two-way and three-way clustering]
\label{eg:two_way_three_way}
For two-way clustering, let $K=2$ and $\bm{i}=(i_1,i_2)$. Then
\begin{alignat*}{3}
    P_{10}f
    &= \mathbb{E}\!\left[f(X_{i_1,i_2})\mid U_{10}\right], \quad &
    P_{01}f
    &= \mathbb{E}\!\left[f(X_{i_1,i_2})\mid U_{01}\right], \quad &
    P_{11}f
    &= f(X_{i_1,i_2}), \\
    \pi_{10}f &= P_{10}f, &
    \pi_{01}f &= P_{01}f, &
    \pi_{11}f &= P_{11}f-P_{10}f-P_{01}f.
\end{alignat*}
For three-way clustering, let $K=3$ and $\bm{i}=(i_1,i_2,i_3)$. The one-way components are
\begin{align*}
    \pi_{100}f
    &=P_{100}f
      =\mathbb{E}\!\left[f(X_{i_1,i_2,i_3})\mid U_{100}\right],\\
    \pi_{010}f
    &=P_{010}f
      =\mathbb{E}\!\left[f(X_{i_1,i_2,i_3})\mid U_{010}\right],\\
    \pi_{001}f
    &=P_{001}f
      =\mathbb{E}\!\left[f(X_{i_1,i_2,i_3})\mid U_{001}\right].
\end{align*}
The two-way components are
\begin{align*}
    \pi_{110}f
    &=P_{110}f-P_{100}f-P_{010}f, &
    P_{110}f
    &=\mathbb{E}\!\left[
        f(X_{i_1,i_2,i_3})
        \mid U_{100},U_{010},U_{110}
      \right],\\
    \pi_{101}f
    &=P_{101}f-P_{100}f-P_{001}f, &
    P_{101}f
    &=\mathbb{E}\!\left[
        f(X_{i_1,i_2,i_3})
        \mid U_{100},U_{001},U_{101}
      \right],\\
    \pi_{011}f
    &=P_{011}f-P_{010}f-P_{001}f, &
    P_{011}f
    &=\mathbb{E}\!\left[
        f(X_{i_1,i_2,i_3})
        \mid U_{010},U_{001},U_{011}
      \right].
\end{align*}
Finally, the three-way component is
\begin{align*}
    \pi_{111}f
    &=
    P_{111}f
    -\pi_{100}f-\pi_{010}f-\pi_{001}f
    -\pi_{110}f-\pi_{101}f-\pi_{011}f,\\
    P_{111}f
    &=
f(X_{i_1,i_2,i_3}).
\end{align*}
$\lozenge$
\end{example}

	Fix \(1\le k\le K\) and \(\bm{e}\in \calE_k\). By Lemma~1 in \cite{chiang2023inference}, for any \(\ell\in \supp(\bm{e})\), the random variable
	\(
	(\pi_{\bm{e}}f)\Bigl(\{U_{\bm{i}\odot \bm{e}'}\}_{\bm{e}'\le \bm{e}}\Bigr)
	\)
	has mean zero conditionally on \(\{U_{\bm{i}\odot \bm{e}'}\}_{\bm{e}'\le \bm{e}-\bm{e}_{\ell}}\).
	Let
	\(
	I_{\bm{N},\bm{e}} = \{\bm{i}\odot \bm{e} : \bm{i}\in [\bm{N}]\}.
	\)
	Then
	\(
	\bigl|I_{\bm{N},\bm{e}}\bigr| = \prod_{k'\in\supp(\bm{e})} N_{k'}.
	\)
	Define
	\begin{align*}
		H_{\bm{N}}^{\bm{e}}(f)
		=\frac{1}{\bigl|I_{\bm{N},\bm{e}}\bigr|} \sum_{\bm{i}\in I_{\bm{N},\bm{e}}}
		(\pi_{\bm{e}}f)\Bigl(\{U_{\bm{i}\odot \bm{e}'}\}_{\bm{e}'\le \bm{e}}\Bigr).
	\end{align*}
	 We now obtain the Hoeffding-type decomposition
	\begin{align}
		\mathbb{E}_N f = \sum_{k=1}^K \sum_{\bm{e}\in\calE_k} H_{\bm{N}}^{\bm{e}}(f). \label{hoeffding}
	\end{align}
	To bound \(\mathbb{E}\bigl[ \|\Gn(f)\|_{\calF} \bigr] \), it thus suffices to control each term \(\mathbb{E}\bigl[\|H_{\bm{N}}^{\bm{e}}(f)\|_{\calF}\bigr]\) separately.

	Finally, for \(1\le k\le K\) and \(\bm{e}\in \calE_k\), define the \emph{uniform entropy integral}
	\begin{align*}
		J_{\bm{e}}(\delta) = J_{\bm{e}}(\delta,\calF,F)
		:= \int_{0}^{\delta} \sup_{Q} \Biggl\{ 1+ \log N\Bigl(P_{\bm{e}}\calF,\|\cdot\|_{Q,2},\tau\|P_{\bm{e}}F\|_{Q,2}\Bigr) \Biggr\}^{k/2} d\tau, 
	\end{align*}
	where
	\(
	P_{\bm{e}}\calF := \{P_{\bm{e}}f : f\in \calF\},
	\)
	and the supremum is taken over all finite discrete distributions \(Q\).
	
The following global maximal inequality for SE empirical processes allows for an arbitrary index order $K$ and moment order $r\in[1,\infty)$. Its proof adapts that of Corollary B.1 in \cite{chiang2023inference} to a more general class of functions.
\begin{theorem}[Global maximal inequality for SE processes]\label{theorem:global_maximal_ineq}
	Suppose $\calF$ is a pointwise measurable class of functions $f:\calS\to \R$.
	Let $(X_\i)_{\i\in [\bm{N}]}$ be a sample from $\calS$-valued random vectors $(X_\i)_{\i\in \N^K}$ satisfying Conditions (SE) and (D). Fix any $1 \le k \le K$ and $\bm{e} \in \mathcal{E}_{k}$. Then, for any $r \in [1,\infty)$, we have
	\[
	|I_{\bN,\e}|^{1/2}\left (\mathbb{E} \left [ \left \| H_{\bN}^\e(f)  \right \|_{\calF}^{{r}} \right ] \right)^{1/{r}} \lesssim
	J_\e(1)  \| F \|_{P,r\vee 2}.
	\]
\end{theorem}
The proof is given in Section~\ref{sec:Proof of Theorem global_maximal_ineq} of the appendix.

Although the global maximal inequality holds for a general moment order $r$, local maximal inequalities usually provide sharper bounds on the first absolute moment of the supremum.
The following is a novel local maximal inequality for SE empirical processes.
\begin{theorem}[Local maximal inequality for SE processes]\label{theorem:local_max_ineq}
		Suppose $\calF$ is a pointwise measurable class of functions $f:\calS\to \R$.
	Let $(X_\i)_{\i\in [\bm{N}]}$ be a sample from $\calS$-valued random vectors $(X_\i)_{\i\in \N^K}$ satisfying Conditions (SE) and (D). Fix $\e\in \{0,1\}^K\setminus\{\bm{0}\}$ and let $\sigma_\e$ be a constant such that
	$\sup_{f\in\calF}\|P_\e f \|_{P,2}\le \sigma_\e \le \|P_\e F\|_{P,2}$. Then, with
	\[ \delta_\e=\sigma_\e/\|P_\e F\|_{P,2}\quad \text{ and } \quad M_\e=\max_{t\in [n]}(P_\e F)\left(\{U_{(t,...,t)\odot \e'}\}_{\e'\le \e}\right),\] we have
	\begin{align*}
		|I_{\bN,\e}|^{1/2}\mathbb{E} \left [ \left \| H_{\bN}^\e(f)  \right \|_{\calF} \right ]  \lesssim
		J_\e(\delta_\e)
			\|P_\e F\|_{P,2}
	 +\frac{J_\e^2(\delta_\e)\|M_\e\|_{P,2}}{\sqrt{n}\delta_\e^2}.
	\end{align*}
\end{theorem}

The proof is given in Section~\ref{sec:Proof of theorem:local_max_ineq} of the appendix.

\begin{remark}\label{rem:local_max_ineq}
Our proof strategy broadly follows that of Theorem 5.1 in \cite{ChenKato2019b} for U-processes, adapted to the SE structure, but with a key difference. In their Theorem~5.1, the Hoffmann--J\o rgensen inequality (which requires independence) is applied via the classical \(U\)-statistic technique of Hoeffding averaging (see, for example, Section~5.1.6 in \citealt{serfling1980approximation}). However, the more intricate dependence structure of separately exchangeable arrays renders Hoeffding averaging inapplicable here. We instead establish Lemma~\ref{lemma:partition}, which partitions the index set \(I_{\bN,\e}\) into transversal groups of size \(m_\e=\min_{k'\in\supp(\e)}N_{k'}\ge n\) (see Lemma~\ref{lemma:partition} for the definition). Together with the AHK representation \eqref{eq:AHK_representation}, this yields i.i.d.\ elements within each group, which permits the use of the Hoffmann--J\o rgensen inequality; an elementary block-maximum comparison then converts the resulting bound into the form stated in Theorem~\ref{theorem:local_max_ineq}, with \(M_\e\) and \(\sqrt n\).
\hfill $\lozenge$\end{remark}



In practice, bounding the uniform entropy integrals can be involved. However, many function classes arising in econometrics, machine learning, and statistics are of VC-type in the sense of \eqref{eq:VC-type-def}, for which the entropy terms admit substantially simpler bounds. We therefore derive a local maximal inequality under the VC-type condition; this is the version used in the proof of Theorem~\ref{thm:gmm}.
\begin{corollary}\label{cor:local_maximal_ineq_VC}
	In the setting of Theorem~\ref{theorem:local_max_ineq}, suppose, in addition, that $\calF$ is of VC-type with characteristics $A\ge (e^{2(K-1)}/16)\vee e$ and $v\ge 1$. Then, for each $\e\in \calE_k$ with $k=\Vert\e\Vert_0$, one has
	\begin{align*}
		|I_{\bN,\e}|^{1/2}\mathbb{E} \left [ \left \| H_{\bN}^\e(f)  \right \|_{\calF} \right ]  \lesssim&
        \sigma_\e\{v\log (A\|P_\e F\|_{P,2}/\sigma_\e)\}^{k/2} +\frac{\|M_\e\|_{P,2}}{\sqrt{n}}\{v\log (A \|P_\e F\|_{P,2}/\sigma_\e)\}^k.
	\end{align*}
    If, in addition, $\|P_\e F\|_{P,2}\le \sigma_\e\,(A\vee\overline N)$, then also
    \begin{align*}
		|I_{\bN,\e}|^{1/2}\mathbb{E} \left [ \left \| H_{\bN}^\e(f)  \right \|_{\calF} \right ]  \lesssim
		\sigma_\e\{v\log (A\vee \overline N)\}^{k/2} +\frac{\|M_\e\|_{P,2}}{\sqrt{n}}\{v\log (A\vee \overline N)\}^k,
	\end{align*}
    since then $\log(A\|P_\e F\|_{P,2}/\sigma_\e)\le \log(A(A\vee\overline N))\le 2\log(A\vee\overline N)$.
\end{corollary}
The proof is given in Section~\ref{sec:proof of cor:local_maximal_ineq_VC} of the appendix. Unlike existing maximal inequalities for i.i.d.\ data (e.g., Corollary 5.1, \citealp{CCK2014AoS}), Corollary~\ref{cor:local_maximal_ineq_VC} applies to the Hoeffding components of separately exchangeable arrays. In particular, when $K=k=1$, the Hoeffding component $|I_{\bN,\e}|^{1/2}\mathbb{E} \left [ \left \| H_{\bN}^\e(f)  \right \|_{\calF} \right ]$ coincides with the ordinary empirical process, and our bound reduces to the i.i.d.\ maximal inequality up to universal constants. For $k\ge 2$, symmetrisation produces a homogeneous Rademacher chaos process of order \(k\), which raises the entropy factors to the $k$-th power. Although this entails a larger entropy penalty, the normalisation term $|I_{\bN,\e}|^{-1/2}$, which equals \(n^{-k/2}\) in a balanced array, decays faster than this penalty grows, provided that the complexity and envelope terms are properly controlled; this yields a tighter bound for higher-order Hoeffding terms.

\section{Simulation}
\label{sec:simulation}
We compare cross-fitting and cross-fitting-free DML with various first-step machine learners in a Monte Carlo simulation of a simple partially linear IV model, to assess whether the cross-fitting-free approach performs comparably to cross-fitting DML in terms of finite-sample bias, standard deviation (SD), root mean squared error (RMSE), and coverage, while offering substantial computational advantages in many settings.

The model is
\begin{align}
\label{eq:PLM}
    Y = D'\theta_0 + g(X) + U, \quad \mathbb{E}[U|X,Z] = 0, 
\end{align}
where $D$ are endogenous variables, $\theta_0$ are the parameters of interest, $X$ is a $d_X$-dimensional vector of exogenous variables, $Z$ is a vector of exogenous and excludable variables, and $g(.)$ is a nuisance function, specified below but treated as unknown in the simulation. The Neyman-orthogonal moment function for \eqref{eq:PLM} is
\begin{align*}
    \psi(W,\theta_0, \eta_0 ) =  \left(Z - \mathbb{E}[Z|X] \right) \left(Y - \mathbb{E}[Y|X]  -  \left({D}-\mathbb{E}[D|X]\right)'\theta_0\right),
\end{align*}
where $W=\{Y,D,X,Z\}$ and $\eta_0 = \{\mathbb{E}[Y|X], \mathbb{E}[D|X], \mathbb{E}[Z|X]\}$. Here $\theta_0$ is exactly identified, so the proposed debiased GMM estimator collapses to an IV estimator with plug-in nuisance estimates of $\eta_0$.

Under two-way clustering, with $N_1$ and $N_2$ clusters in the two dimensions and hence $N=N_1N_2$ cells (in the notation of Section~\ref{sec:setup}), the data generating process (DGP) for \eqref{eq:PLM} is, for $i=1,\ldots, N_1$ and $j=1,\ldots, N_2$,
\begin{align*}
    Y_{ij} =& \theta_0 D_{ij}+ g(X_{ij}) + U_{ij}, \quad \theta_0 =1,\\
    g(X_{ij}) =& \sin(X_{ij,1} + X_{ij,2}) + 0.5 X_{ij,3} + \cdots + 0.5 X_{ij,d_X},  \\
    D_{ij} = & 0.5 Z_{ij} + 0.2(X_{ij,1}+X_{ij,2})^2 + V_{ij}, \\
    Z_{ij} = & X_{ij,2}X_{ij,3} + \varepsilon_{ij},
\end{align*}
where $d_X\geq 3$. The nuisance conditional expectation functions under this DGP are nonlinear in the observables, but the nonlinearity can be well captured by a small number of polynomial spline basis functions. To induce two-way cluster dependence and endogeneity of $D_{ij}$, the exogenous variables $X$ and the error terms are generated from the following heterogeneous components:
\begin{align*}
     X_{ij,k} =& \alpha_{i,k} + \gamma_{j,k} + \epsilon_{ij,k}, \quad k=1,\ldots,d_X,  \\
    U_{ij} = & \alpha_{i,u} + \gamma_{j,u}+\epsilon_{ij,u}, \\
    V_{ij} = & \alpha_{i,v} + \gamma_{j,v}+\epsilon_{ij,v}, \\
    \varepsilon_{ij} = & \alpha_{i,\varepsilon} + \gamma_{j,\varepsilon}+\epsilon_{ij,\varepsilon}.
\end{align*}
For each cell $(i,j)$, $(\epsilon_{ij,1}, \ldots,\epsilon_{ij,d_X})$ is drawn from a joint normal distribution with mean zero, unit variances, and pairwise covariance $\rho>0$; each of the pairs $(\alpha_{i,u},\alpha_{i,v})$, $(\gamma_{j,u},\gamma_{j,v})$, and $(\epsilon_{ij,u},\epsilon_{ij,v})$ is drawn from a bivariate normal distribution with mean zero, unit variances, and covariance $\rho_{uv}>0$, which governs the degree of endogeneity; and $\alpha_{i,\varepsilon}$, $\gamma_{j,\varepsilon}$, and $\epsilon_{ij,\varepsilon}$ are standard normal and mutually independent. All draws are independent across the row index $i$, the column index $j$, and the cell index $(i,j)$, and the shock families attached to $X$, to $(U,V)$, and to $\varepsilon$ are mutually independent, so that $\mathbb{E}[U|X,Z]=0$, as required by \eqref{eq:PLM}.


The first machine learner for the unknown nuisance functions is the multiplier bootstrap LASSO (mpb-LASSO) of \cite{chiang2023inference}, who provide a theoretical guarantee on its penalty level and a bound on the $\ell^2$ convergence rate of the resulting LASSO estimator under multiway cluster dependence. We further consider cross-validated ridge, elastic net, random forest, and neural network. As the multiway cluster-robust multiplier bootstrap is designed more specifically for this setting, we briefly introduce it below; the tuning choices for each learner are given in the implementation details after the tables.

Random forest and neural network take $X = (X_1,...,X_{d_X})$ directly as inputs, whereas for LASSO, ridge, and elastic net we first approximate the conditional expectations in $\eta_0$ by spline basis functions. For a generic $d_X$-dimensional vector $x = (x_1,...,x_{d_X})$, the spline transformation of degree $l$ with $m$ knots is
\begin{align*}
    \operatorname{SP}(x;l,m) =& \operatorname{sp}(x_1;l,m) \otimes \cdots \otimes \operatorname{sp}(x_{d_X};l,m), \\
    \operatorname{sp}(x_k;l,m) =& \{1,x_k, x_k ^2,\ldots,x_k^l, (x_k - \tau_{k,1} )^l _{+}, \ldots,  (x_k - \tau_{k,m} )^l _{+} \}, \ k=1,\ldots, d_X,
\end{align*}
where $(z)_{+} = z 1\{z>0\}$ and $(\tau_{k,1},...,\tau_{k,m})$ are the ordered (interior) knots for $x_k$. The dimension of $ \operatorname{SP}(x;l,m)$ is $(1+l+m)^{d_X}$. As a common and convenient choice, we use natural cubic splines ($l=3$), which have $(m+2)^{d_X}$ free slope parameters.

\subsection{Multiplier bootstrap LASSO}
\label{sec:mpb_lasso}
The LASSO penalty level $\lambda$ is chosen by the multiplier bootstrap as follows. Let $w_{ij} = \operatorname{SP}(X_{ij};l,m)$ denote the vector of spline basis functions.\footnote{We suppose $w_{ij}$ is standardised after the spline transformation.} Let $\tilde{\varepsilon}$ be the residual from the initial post-LASSO estimation with penalty level $\tilde{\lambda} = \log(N_1N_2) \left(\frac{\log p}{N_1N_2}\right)^{1/2}$, and define the sample average of the empirical score as
\begin{align*}
    \tilde{S} = \frac{1}{N_1N_2} \sum_{i=1}^{N_1}\sum_{j=1}^{N_2} w_{ij}\tilde{\varepsilon}_{ij}.
\end{align*}
For each $b=1,...,B$, let $\{\xi_i^b\}_{i=1}^{N_1}$ and $\{\zeta_j^b\}_{j=1}^{N_2}$ each be a sequence of random draws from $N(0,1)$. Following \cite{chiang2023inference}, the multiplier bootstrap version of $\tilde{S}$ is
\begin{align*}
    \tilde{S}_{H}^b = \frac{1}{N_1} \sum_{i=1}^{N_1} \xi_i^b \left[ \frac{1}{N_2}\sum_{j=1}^{N_2} w_{ij} \tilde{\varepsilon}_{ij} \right] + \frac{1}{N_2} \sum_{j=1}^{N_2} \zeta_j^b \left[ \frac{1}{N_1}\sum_{i=1}^{N_1} w_{ij} \tilde{\varepsilon}_{ij} \right].
\end{align*}
Let $Q\left(\left\Vert  \tilde{S}_{H}^b \right\Vert_{\infty};1-\alpha\right)$ be the $1-\alpha$-th quantile of $\left\Vert  \tilde{S}_{H}^b \right\Vert_{\infty}$. For some $\alpha\to 0$ and $c>1$, we set $\lambda = 2cQ\left(\left\Vert  \tilde{S}_{H}^b \right\Vert_{\infty};1-\alpha\right)$ and obtain the post-LASSO nuisance estimates using this common penalty level.

\subsection{Simulation results} 

The designs combine different sample sizes $(N_1,N_2)$, numbers of cross-fitting folds $S$, dimensions of the covariates $X$, and the five first-step machine learners; performance is measured by Monte Carlo bias, SD, coverage, and computation time.

\begin{table}[htbp]
\small
\centering
\caption{Simulation results: DML with and without cross-fitting, moderate sample size}
\label{tab:sim_results_n50}
\scriptsize
\setlength{\tabcolsep}{2.2pt}
\renewcommand{\arraystretch}{1.10}
\begin{tabular}{@{}ccclrrrrrr@{}}
\toprule
$N_1=N_2$ & $d_X$ &$S$ & Learner & Bias & SD & RMSE & Coverage & Runtime/mpb & CF/No-CF \\
\midrule\midrule
50 & 3 & 0 & mpb-LASSO & 0.019 & 0.154 & 0.155 & 0.926 & 1.00$\times$ & Ref. \\
 &  &  & Ridge & -0.008 & 0.146 & 0.146 & 0.939 & 2.16$\times$ & Ref. \\
 &  &  & Elastic net & -0.007 & 0.147 & 0.147 & 0.941 & 2.41$\times$ & Ref. \\
 &  &  & Random forest & -0.029 & 0.118 & 0.122 & 0.925 & 13.40$\times$ & Ref. \\
 &  &  & Neural network & 0.008 & 0.149 & 0.150 & 0.937 & 73.87$\times$ & Ref. \\
50 & 3 & 2 & mpb-LASSO & 0.059 & 0.555 & 0.558 & 0.955 & 1.00$\times$ & 1.76$\times$ \\
 &  &  & Ridge & -0.016 & 0.236 & 0.236 & 0.945 & 2.65$\times$ & 2.16$\times$ \\
 &  &  & Elastic net & 0.092 & 3.725 & 3.724 & 0.958 & 6.14$\times$ & 4.49$\times$ \\
 &  &  & Random forest & -0.033 & 0.096 & 0.101 & 0.968 & 6.91$\times$ & 0.91$\times$ \\
 &  &  & Neural network & 0.030 & 0.131 & 0.135 & 0.930 & 150.01$\times$ & 3.57$\times$ \\
50 & 3 & 4 & mpb-LASSO & 0.040 & 0.164 & 0.169 & 0.990 & 1.00$\times$ & 7.98$\times$ \\
 &  &  & Ridge & -0.009 & 0.152 & 0.152 & 0.972 & 2.85$\times$ & 10.53$\times$ \\
 &  &  & Elastic net & -0.017 & 0.161 & 0.161 & 0.975 & 4.15$\times$ & 13.77$\times$ \\
 &  &  & Random forest & -0.033 & 0.108 & 0.113 & 0.987 & 12.11$\times$ & 7.22$\times$ \\
 &  &  & Neural network & 0.023 & 0.138 & 0.140 & 0.976 & 274.30$\times$ & 29.64$\times$ \\\hline
50 & 4 & 0 & mpb-LASSO & 0.014 & 0.149 & 0.150 & 0.933 & 1.00$\times$ & Ref. \\
 &  &  & Ridge & 0.000 & 0.144 & 0.143 & 0.930 & 10.41$\times$ & Ref. \\
 &  &  & Elastic net & -0.005 & 0.144 & 0.144 & 0.932 & 29.25$\times$ & Ref. \\
 &  &  & Random forest & -0.023 & 0.125 & 0.127 & 0.933 & 5.20$\times$ & Ref. \\
 &  &  & Neural network & 0.016 & 0.146 & 0.147 & 0.926 & 26.63$\times$ & Ref. \\
50 & 4 & 2 & mpb-LASSO & 0.007 & 1.585 & 1.584 & 0.924 & 1.00$\times$ & 1.12$\times$ \\
 &  &  & Ridge & -0.033 & 0.118 & 0.122 & 0.966 & 4.62$\times$ & 0.50$\times$ \\
 &  &  & Elastic net & -0.059 & 0.421 & 0.425 & 0.958 & 8.78$\times$ & 0.34$\times$ \\
 &  &  & Random forest & -0.027 & 0.105 & 0.109 & 0.973 & 3.77$\times$ & 0.81$\times$ \\
 &  &  & Neural network & 0.023 & 0.126 & 0.128 & 0.937 & 66.83$\times$ & 2.82$\times$ \\
50 & 4 & 4 & mpb-LASSO & -0.003 & 0.167 & 0.167 & 0.990 & 1.00$\times$ & 7.17$\times$ \\
 &  &  & Ridge & -0.023 & 0.150 & 0.152 & 0.974 & 11.36$\times$ & 7.82$\times$ \\
 &  &  & Elastic net & -0.038 & 0.293 & 0.296 & 0.986 & 40.25$\times$ & 9.87$\times$ \\
 &  &  & Random forest & -0.025 & 0.116 & 0.119 & 0.988 & 5.52$\times$ & 7.61$\times$ \\
 &  &  & Neural network & 0.026 & 0.134 & 0.136 & 0.977 & 95.23$\times$ & 25.65$\times$ \\
\bottomrule
\end{tabular}

\vspace{0.35em}
\begin{minipage}{0.96\linewidth}
\scriptsize
\emph{Notes:}   $S=0$ denotes no cross-fitting and $S>0$ cross-fitting with $S$ folds. Runtime/mpb is the total runtime of the learner divided by that of mpb-LASSO in the same $(N_1,N_2,S,d_X)$ setting. CF/No-CF is the total runtime of a learner with cross-fitting divided by that of the same learner without cross-fitting, holding $(N_1,N_2,d_X)$ fixed; ``Ref.'' denotes the reference (no-cross-fitting) case for the same learner.
\end{minipage}
\end{table}

\begin{table}[htbp]
\small
\centering
\caption{Simulation results: DML with and without cross-fitting, small sample size}
\label{tab:sim_results_n25}
\scriptsize
\setlength{\tabcolsep}{2.2pt}
\renewcommand{\arraystretch}{1.10}
\begin{tabular}{@{}ccclrrrrrr@{}}
\toprule
$N_1=N_2$ & $d_X$& $S$ & Learner & Bias & SD & RMSE & Coverage & Runtime/mpb & CF/No-CF \\
\midrule\midrule
25 & 3 & 0 & mpb-LASSO & 0.120 & 0.237 & 0.265 & 0.871 & 1.00$\times$ & Ref. \\
 &  &  & Ridge & -0.014 & 0.229 & 0.229 & 0.921 & 2.04$\times$ & Ref. \\
 &  &  & Elastic net & -0.015 & 0.231 & 0.232 & 0.918 & 4.44$\times$ & Ref. \\
 &  &  & Random forest & -0.029 & 0.155 & 0.158 & 0.946 & 4.79$\times$ & Ref. \\
 &  &  & Neural network & 0.028 & 0.204 & 0.206 & 0.891 & 79.86$\times$ & Ref. \\
25 & 3 & 2 & mpb-LASSO & 0.666 & 24.783 & 24.779 & 0.956 & 1.00$\times$ & 1.77$\times$ \\
 &  &  & Ridge & -0.038 & 0.318 & 0.320 & 0.949 & 2.93$\times$ & 2.55$\times$ \\
 &  &  & Elastic net & -0.125 & 1.107 & 1.114 & 0.928 & 12.34$\times$ & 4.92$\times$ \\
 &  &  & Random forest & -0.032 & 0.116 & 0.120 & 0.985 & 2.88$\times$ & 1.06$\times$ \\
 &  &  & Neural network & 0.046 & 0.223 & 0.228 & 0.927 & 188.26$\times$ & 4.17$\times$ \\
25 & 3 & 4 & mpb-LASSO & -0.013 & 3.559 & 3.557 & 0.984 & 1.00$\times$ & 7.66$\times$ \\
 &  &  & Ridge & -0.036 & 0.311 & 0.313 & 0.969 & 3.25$\times$ & 12.19$\times$ \\
 &  &  & Elastic net & -0.066 & 0.546 & 0.550 & 0.988 & 11.47$\times$ & 19.79$\times$ \\
 &  &  & Random forest & -0.037 & 0.131 & 0.136 & 0.999 & 4.37$\times$ & 6.99$\times$ \\
 &  &  & Neural network & 0.020 & 0.174 & 0.175 & 0.981 & 387.44$\times$ & 37.16$\times$ \\\hline
25 & 4 & 0 & mpb-LASSO & 0.086 & 0.265 & 0.278 & 0.899 & 1.00$\times$ & Ref. \\
 &  &  & Ridge & -0.019 & 0.165 & 0.166 & 0.949 & 3.82$\times$ & Ref. \\
 &  &  & Elastic net & 0.001 & 0.223 & 0.223 & 0.920 & 6.85$\times$ & Ref. \\
 &  &  & Random forest & -0.035 & 0.151 & 0.155 & 0.952 & 3.22$\times$ & Ref. \\
 &  &  & Neural network & 0.033 & 0.196 & 0.198 & 0.897 & 42.23$\times$ & Ref. \\
25 & 4 & 2 & mpb-LASSO & 0.075 & 1.047 & 1.049 & 0.954 & 1.00$\times$ & 1.67$\times$ \\
 &  &  & Ridge & -0.002 & 0.569 & 0.569 & 0.921 & 3.53$\times$ & 1.54$\times$ \\
 &  &  & Elastic net & 0.119 & 13.084 & 13.078 & 0.899 & 2.34$\times$ & 0.57$\times$ \\
 &  &  & Random forest & -0.027 & 0.122 & 0.125 & 0.974 & 1.91$\times$ & 0.99$\times$ \\
 &  &  & Neural network & 0.041 & 0.211 & 0.215 & 0.941 & 108.49$\times$ & 4.28$\times$ \\
25 & 4 & 4 & mpb-LASSO & -0.334 & 8.643 & 8.645 & 0.984 & 1.00$\times$ & 8.87$\times$ \\
 &  &  & Ridge & -0.031 & 0.177 & 0.179 & 0.977 & 4.21$\times$ & 9.78$\times$ \\
 &  &  & Elastic net & -0.071 & 1.680 & 1.680 & 0.975 & 3.03$\times$ & 3.93$\times$ \\
 &  &  & Random forest & -0.031 & 0.137 & 0.140 & 0.996 & 2.63$\times$ & 7.25$\times$ \\
 &  &  & Neural network & 0.030 & 0.176 & 0.178 & 0.990 & 182.14$\times$ & 38.27$\times$ \\
\bottomrule
\end{tabular}

\vspace{0.35em}
\begin{minipage}{0.96\linewidth}
\scriptsize
\emph{Notes:}   $S=0$ denotes no cross-fitting and $S>0$ cross-fitting with $S$ folds. Runtime/mpb is the total runtime of the learner divided by that of mpb-LASSO in the same $(N_1,N_2,S,d_X)$ setting. CF/No-CF is the total runtime of a learner with cross-fitting divided by that of the same learner without cross-fitting, holding $(N_1,N_2,d_X)$ fixed; ``Ref.'' denotes the reference (no-cross-fitting) case for the same learner.
\end{minipage}
\end{table}

Implementation details are as follows. Cubic natural splines with 3 knots, of raw dimension $5^{d_X}$, are used for mpb-LASSO, ridge, and elastic net with even weights; the spline bases are generated as tensor products and then standardised. For mpb-LASSO, the initial penalty $\tilde\lambda$ and the final penalty $\lambda$ are as defined in Section~\ref{sec:mpb_lasso}, with $B=500$, $\alpha=0.10$, and $c=1.10$. For ridge and elastic net, the penalty level is chosen by 5-fold cross-validation with the min-MSE rule. The random forest uses 500 trees, $mtry=\max\{1,\lfloor \sqrt{d_X}\rfloor\}$, and minimum node size 5. For the neural network, raw features and outcomes are standardised; it uses two dense layers with $(64,32)$ hidden units, ReLU activation, Adam with learning rate $10^{-3}$, 30 epochs, batch size 32, and validation split 0.2 with the MSE loss and patience 5 for early stopping. Each learner is applied separately to the nuisance functions of $Y$, $D$, and $Z$. For spline learners, knots, centring, and scaling are computed on the training sample.

We first consider a moderate sample size. Under multiway clustering, the effective sample size is determined by the clustering dimension with the fewest clusters. In Table~\ref{tab:sim_results_n50}, $N_1=N_2=50$ implies an effective sample size of 50 for the second step and for the first step without cross-fitting. With cross-fitting, the auxiliary sample for nuisance estimation is smaller still, with effective size $25$ when $S=2$ and $38$ or $37$ when $S=4$. We therefore expect less accurate nuisance estimation and hence poor finite-sample performance of the second-step estimator. A larger covariate dimension $d_X$ implies more degrees of freedom in the nuisance estimation, which may exacerbate the problem given the limited effective sample size.

The upper panel of Table~\ref{tab:sim_results_n50} reports the finite-sample performance of DML with different numbers of cross-fitting folds $S$, based on 1000 Monte Carlo replications. Here $d_X=3$, which corresponds to $125$ degrees of freedom in the nuisance functions for LASSO, ridge, and elastic net; the effective degrees of freedom of random forest and the neural network also increase with $d_X$, but depend on the particular choice of architecture. DML without cross-fitting has a smaller Monte Carlo bias for all five first-step learners, and a smaller SD and RMSE for the penalised linear learners (mpb-LASSO, ridge, and elastic net) when $S$ is small; for random forest and the neural network, cross-fitting attains a somewhat smaller SD and RMSE at $S=2$. Increasing $S$ from $2$ to $4$ markedly reduces the SD and bias of the penalised linear learners but slightly increases the SD for random forest and the neural network; overall, the results remain comparable to, if not worse than, those without cross-fitting. The empirical coverage rate is in most cases slightly better with cross-fitting, but at the cost of much wider confidence intervals, as reflected in the generally larger SDs; moreover, cross-fitting DML tends to over-cover as $S$ increases. The lower panel of Table~\ref{tab:sim_results_n50}, with $d_X =4 $ and hence more effective degrees of freedom for each first-step learner, shows largely the same pattern.\footnote{That a larger dimension does not worsen estimation and inference may be because the extra covariate in the DGP contributes to the variation of $Y$: with the same error term $U$, the observables can explain a larger proportion of the variation in $Y$, which leads to more accurate nuisance estimation. We expect the pattern to reverse, however, as the dimension of $X$ increases further, owing to the higher complexity induced in the first-step learners.} 

DML without cross-fitting also yields a substantial computational gain when $S=4$, and the gain is largest when the neural network is the first-step learner. The second-last column of Table~\ref{tab:sim_results_n50} shows that the neural network takes significantly more time to run, and the last column that, with $S=4$, the total runtime with cross-fitting is about $7$ to $14$ times that without cross-fitting for the non-neural-network learners, and up to almost $30$ times for the neural network. This computational gain from our approach is significant, given that the machine learners are already computationally costly to varying degrees even without cross-fitting when the sample is large. At $S=2$, by contrast, the runtime comparison is mixed: for some learners the ratio is below one, i.e., cross-fitting with $S=2$ was faster than the full-sample fit in those cells, reflecting learner-specific tuning and fitting costs on subsamples.\footnote{We run our R programs on HPCC using 25 cores on one computing node and compute the $1000$ Monte Carlo replications in parallel, $40$ per core. Within each replication, the first-step learners run sequentially. For each learner type, the total runtime is the sum across replications of the recorded time differences.} 

Table~\ref{tab:sim_results_n25} reports the results for a smaller sample size, $N_1=N_2=25$, which implies an effective sample size of $25$ for the second step and for the first step without cross-fitting, but one as small as $12$ for the first step with $S=2$ and $18$ with $S=4$. With such small effective samples, the task can be somewhat demanding: the DML estimators based on high-dimensional linear methods with $\ell^1$ or $\ell^2$ penalties all perform poorly when cross-fitting is used. In particular, DML with the mpb-LASSO first step suffers from severe bias and/or large SD when $S=2$. The situation improves when $S$ increases from $2$ to $4$, but the trade-offs in potential over-coverage and computation time are also very costly. 

\section{Conclusion}
\label{sec:conslusion}
This paper develops a cross-fitting-free asymptotic theory for two-step debiased GMM estimators under multiway clustered dependence, and shows that valid inference in this setting hinges on new empirical process techniques. Combining orthogonal moment conditions with a localisation-based argument, we show that the impact of high-dimensional or nonparametric nuisance estimation can be controlled without sample splitting, even when the effective sample size is determined by the number of independent cluster units. The resulting estimators are asymptotically linear and normal under separately exchangeable sampling, providing a practical inference framework for empirically relevant clustered environments. In simulations, the proposed approach attains comparable or even better finite-sample bias and standard deviation when the sample size and the number of cross-fitting folds are small, and a substantial computational gain otherwise. A central contribution is the derivation of new global and local maximal inequalities for possibly uncountable, pointwise measurable function classes under multiway dependence, which fill a gap in the existing theory and may be useful beyond DML. Future work may explore stability-type conditions under multiway clustering and extend these tools to other two-step and high-dimensional problems with complex dependence structures.

\appendix

\section{Appendix}
Throughout the appendix, for $0 < \beta < \infty$, let $\psi_{\beta}$ be the function on $[0,\infty)$ given by $\psi_{\beta} (x) = e^{x^{\beta}}-1$, and let $\| \cdot \|_{\psi_{\beta}}$ denote the associated Orlicz norm, i.e.,
$\| \xi \|_{\psi_\beta}=\inf \{ C>0: \mathbb{E}[ \psi_{\beta}( | \xi | /C)] \leq 1\}$ for a real-valued random variable $\xi$.

\subsection{Cluster-size heterogeneity}
\label{app:cluster_heterogeneity}

This section extends the main results to allow cells to contain heterogeneous numbers of unit observations. The results are asymptotic and inherit the uniformity convention of Section~\ref{sec:dml_gmm_main}: statements are uniform over $\{\mathcal{P}_n\}$ in the sequential sense, $P_n$ denotes the data generating measure along an arbitrary sequence $\{P_n\}$ with $P_n\in\mathcal{P}_n$, and all expectations, variances, probabilities, and stochastic orders below are taken under $P_n$. For the array $\{(M_{\bm i}, \{X_{\bm i m}\}_{1\le m\le M_{\bm i}}): \bm i \in \mathbb{N}^K \}$ introduced in Section~\ref{sec:heterogeneous_clusters} and a unit-level score $\psi(X,\theta,\eta)$, define the cell-aggregated score
$$
    \bar\psi(X_{\bm i},\theta,\eta) = \sum_{m=1}^{M_{\bm i}}
    \psi(X_{\bm i m},\theta,\eta).
$$
Let $m_0=\mathbb E[M_{\bm i}]>0$; the population moment condition is
$$
    \mathbb E[\bar\psi(X_{\bm i},\theta_0,\eta_0)] / m_0=0.
$$
The full-sample unit-weighted empirical moment is
$$
    \hat{\bar{\psi}}_N(\theta) = \frac{
    \sum_{\bm i\in[\bm N]}\sum_{m=1}^{M_{\bm i}} \psi(X_{\bm i m},\theta,\widehat{\bar{\eta}}) }{ \sum_{\bm i\in[\bm N]}M_{\bm i} }
    = \frac{\mathbb E_N\bar\psi(X_{\bm i},\theta,\widehat{\bar{\eta}})} {\mathbb E_NM_{\bm i}},
$$
where $\widehat{\bar{\eta}}$ is a first-stage machine learner of $\eta_0$ that treats observations at the unit rather than the cell level. The corresponding cross-fitting-free debiased GMM estimator is
$$
    \widehat{\bar{\theta}} = \arg\min_{\theta\in\Theta} \hat{\bar{\psi}}_N(\theta)' \hat{\bar\Upsilon} \hat{\bar{\psi}}_N(\theta),
$$
where $\hat{\bar\Upsilon}$ is a weighting matrix formed in the same way as $\hat{\Upsilon}$ in the baseline setup, except that the within-cell scores are first aggregated; see, for example, the middle term of the variance estimator below.

Define the population Jacobian for the normalised cell-level score
$$
     \bar{J}_0 = - \partial_\theta \left[
    \frac{\mathbb E[\bar\psi(X_{\bm i},\theta,\eta_0)]}{m_0}
    \right]_{\theta=\theta_0}.
$$
Also define
$$
    \bar{\Psi}_{0,n} := \frac{\mu_1}{m_0^2} {\rm Var}( \mathbb{E}[\bar{\psi}(X_{\bm i},\theta_0,\eta_0)|U_{1,0,...,0}]) +...+ \frac{\mu_K}{m_0^2} {\rm Var}( \mathbb{E}[\bar{\psi}(X_{\bm i},\theta_0,\eta_0)|U_{0,0,...,1}]).
$$
Let
$
    \bar{V}_{0,n}
    =
    \left(\bar{J}_0^{\prime}\bar \Upsilon \bar{J}_0\right)^{-1}
    \bar{J}_0^{\prime} \bar \Upsilon
    \bar{\Psi}_{0,n}
    \bar \Upsilon \bar{J}_0
    \left(\bar{J}_0^{\prime}\bar \Upsilon \bar{J}_0\right)^{-1}
$,
where $\bar \Upsilon$ is the probability limit of $\hat{\bar\Upsilon}$. To state asymptotic normality of this estimator, we impose the following conditions in place of Assumption~\ref{assu:gmm_reg}.
\begin{assumption}
\label{assu:gmm_reg_hetero}
For each $n$ large enough, the following conditions hold uniformly over $P_n\in\mathcal{P}_n$, with all constants below independent of $n$ and $P_n$.

\begin{enumerate}[(i)]
    \item The cell-level array $\{(M_{\bm i}, \{X_{\bm i m}\}_{1\le m\le M_{\bm i}}): \bm i \in \mathbb{N}^K \}$
    satisfies Conditions (SE) and (D) of Section~\ref{sec:setup}.

    \item $c\le m_0\le C$ and $\mathbb E[M_{\bm i}^{1+\delta}]< C$ for some $\delta>0$ and constants $0<c\le C<\infty$.

    \item Assumption~\ref{assu:gmm_reg} holds for the normalised cell score $\psi^{\rm H}(X_{\bm i},\theta,\eta) = m_0^{-1}\bar\psi(X_{\bm i},\theta,\eta)$ and $\widehat{\bar{\eta}}$.

    \item For $j\in[q]$, let $\bar f_j(\eta) = \bar\psi_j(X_{\bm i},\theta_0,\eta)
        - \bar\psi_j(X_{\bm i},\theta_0,\eta_0)$,
    and define the scalar localised cell-level class
    $$
        \bar{ \mathcal F}_n = \left\{ m_0^{-1}\left(\bar f_j(\eta)-\mathbb E_{P_n}[\bar f_j(\eta)]\right): \eta\in\Gamma_n(\eta_0),\ j\in[q] \right\}.
    $$
    Let $\bar{F}_n\ge1$ be a measurable envelope for $\bar{ \mathcal F}_n$. For some $r> 2$, $\|\bar{F}_n\|_{P_n,r}<\infty$, $n^{1/4}\|\bar{F}_n\|_{P_n,2}\le A_n\vee\overline{N}$, and $\bar{ \mathcal F}_n$ is VC-type with characteristics $A_n\ge (e^{2(K-1)}/16)\vee e$ and $v_n\ge1$. The rate condition in \eqref{eq:rate_condition} holds for $\bar{F}_n$.
\end{enumerate}
\end{assumption}

Assumption~\ref{assu:gmm_reg_hetero}(ii) is a mild regularity condition on the heterogeneous cell sizes, and Assumption~\ref{assu:gmm_reg_hetero}(iv) is imposed on the class of cell-aggregated scores. If cell sizes are uniformly bounded and the unit-level score class is VC-type, (iv) follows from standard properties of VC-type classes; if cell sizes are unbounded, it is natural to impose the entropy and moment conditions at the cell level.

\begin{theorem}[Asymptotic linearity and normality with cluster-size heterogeneity]
\label{thm:heterogeneous_cells}
Suppose Assumption~\ref{assu:gmm_reg_hetero} holds, $\widehat{\bar{\theta}}\overset{p}{\to}\theta_0\in{\rm int}(\Theta)$, and $\hat{\bar \Upsilon} \overset{p}{\to}\bar \Upsilon$ for some positive-definite matrix $\bar \Upsilon$ whose eigenvalues are bounded above and away from zero. Then, uniformly over $P_n\in\mathcal{P}_n$ as $n\to\infty$,
$$
    \sqrt n(\widehat{\bar{\theta}}-\theta_0) = \left(\bar{J}_0^{\prime}\bar \Upsilon \bar{J}_0\right)^{-1}   \bar{J}_0^{\prime}\bar \Upsilon
    \sqrt n  \frac{1}{m_0} \mathbb E_N[ \bar\psi(X_{\bm i},\theta_0,\eta_0)]+
    \sum_{k=1}^K\sum_{\bm e\in\mathcal E_k}O_{P_n}(\rho_{n,k}^{\rm H}) + o_{P_n}(1),
$$
where $\rho_{n,k}^{\rm H}$ is $\rho_{n,k}$ of Theorem~\ref{thm:gmm} with $F_n$ replaced by $\bar{F}_n$. If the rate condition in Assumption~\ref{assu:gmm_reg_hetero}(iv) holds and the smallest eigenvalue of $\bar{\Psi}_{0,n}$ is bounded away from zero, then, uniformly over $P_n\in\mathcal{P}_n$ as $n\to\infty$,
$$
    \bar{V}_{0,n}^{-1/2}\sqrt n(\widehat{\bar{\theta}}-\theta_0)\overset{d}{\to} N(0,I_d).
$$
\end{theorem}

\begin{proof}
Fix an arbitrary sequence $\{P_n\}$ with $P_n\in\mathcal{P}_n$ satisfying the conditions of the theorem; in accordance with the uniformity convention of Section~\ref{sec:dml_gmm_main}, all statements below are along this sequence. The proof applies the arguments for Theorem~\ref{thm:gmm} to the normalised cell-level score $\psi^{\rm H}(X_{\bm i},\theta,\eta)$; the only additional step is to account for the random denominator $\mathbb E_N M_{\bm i}$ in $\hat{\bar{\psi}}_N(\theta)$. By Lemma~\ref{lemma:4_3}, $\mathbb E_N M_{\bm i}\overset{p}{\to}m_0$.
Moreover,
$$
    \hat{\bar \psi}_N(\theta_0) = \frac{\mathbb E_N\bar\psi(X_{\bm i},\theta_0,\widehat{\bar{\eta}})}
    {\mathbb E_NM_{\bm i}} = \frac{m_0}{\mathbb E_NM_{\bm i}}
    \mathbb E_N\psi^{\rm H}(X_{\bm i},\theta_0,\widehat{\bar{\eta}}).
$$
Here $\frac{m_0}{\mathbb{E}_N M_{\bm i}} = 1+o_{P_n}(1)$, and, with $\psi$ replaced by $\psi^{\rm H}$, the orthogonality and localisation arguments for $\mathbb E_N\psi^{\rm H}(X_{\bm i},\theta_0,\widehat{\bar{\eta}})$ are unchanged.

The form of $\bar{\Psi}_{0,n}$ follows from the Hájek projection of the cell-level score:
$$
    \sqrt n  \mathbb E_N\psi^{\rm H}(X_{\bm i},\theta_0,\eta_0)= \sum_{k=1}^K
    \frac{\sqrt n}{N_k} \sum_{i_k=1}^{N_k} \mathbb E[
    \psi^{\rm H}(X_{\bm i},\theta_0,\eta_0) \mid  {U_{i_k\bm e_k}}]  + o_{P_n}(1).
$$
Substituting $\psi^{\rm H}=m_0^{-1}\bar\psi$ gives the expression for $\bar{\Psi}_{0,n}$. The conclusion follows from Slutsky's lemma and the same Cramér--Wold and Lyapunov CLT argument as in the proof of Theorem~\ref{thm:gmm}.
\end{proof}

For variance estimation, let $\bar M = \sum_{\bm i\in[\bm N]}M_{\bm i}$ and define
\begin{align*}
    \hat{\bar{\Psi}}_N & = \frac{n}{\bar M^2} \sum_{k=1}^K  
    \sum_{\bm i,\bm j\in[\bm N]}
    \bar\psi(X_{\bm i}, \widehat{\bar{\theta}}, \widehat{\bar{\eta}})
    \bar\psi(X_{\bm j}, \widehat{\bar{\theta}}, \widehat{\bar{\eta}})'
    1\{i_k=j_k\}, \\
    \hat{\bar{J}}_N & = - \frac{1}{\bar M} \sum_{\bm i\in[\bm N]} \sum_{m=1}^{M_{\bm i}} \partial_\theta \psi(X_{\bm i m},\widehat{\bar{\theta}},\widehat{\bar{\eta}}).
\end{align*}
The variance estimator is then
$$
    \hat {\bar{V}} = \left(\hat{\bar{J}}_N^{\prime}\hat{\bar\Upsilon}\hat{\bar{J}}_N\right)^{-1} \hat{\bar{J}}_N^{\prime}\hat{\bar{\Upsilon}} \hat{\bar{\Psi}}_N
    \hat{\bar{\Upsilon}} \hat{\bar{J}}_N  \left(\hat{\bar{J}}_N^{\prime}\hat{\bar{\Upsilon}} \hat{\bar{J}}_N\right)^{-1}.
$$

\begin{theorem}[Variance estimation with heterogeneous cell sizes]
\label{thm:heterogeneous_cells_var}
Suppose that the conditions of Theorem~\ref{thm:heterogeneous_cells} hold, that the moment conditions \eqref{eq:var_moment1}--\eqref{eq:var_moment3} hold instead with $\psi^{\rm H}(X_{\bm i},\theta_0,\eta)$, $\partial_\theta\psi^{\rm H}(X_{\bm i},\theta_0,\eta)$, and $B^{\rm H}(X_{\bm i},\eta)$, respectively, where $B^{\rm H}(X_{\bm i},\eta)$ is the local Hölder modulus of $\partial_\theta\psi^{\rm H}(X_{\bm i},\theta_0,\eta)$, and that Assumption~\ref{assu:gmm_reg}(viii) holds additionally with $a(x,\eta)$ equal to each of $\|\psi^{\rm H}(x,\theta_0,\eta)-\psi^{\rm H}(x,\theta_0,\eta_0)\|^2$, $\|\partial_\theta\psi^{\rm H}(x,\theta_0,\eta)\|^2$, and $B^{\rm H}(x,\eta)^2$. Then, uniformly over $P_n\in\mathcal{P}_n$ as $n\to\infty$,
$$
   \left \|\hat {\bar{V}}-\bar{V}_{0,n}\right\|\overset{p}{\to} 0.
$$

\end{theorem}

\begin{proof}
Fix an arbitrary sequence $\{P_n\}$ with $P_n\in\mathcal{P}_n$ satisfying the conditions of the theorem; all statements below are along this sequence, in accordance with the uniformity convention of Section~\ref{sec:dml_gmm_main}. The proof of Theorem~\ref{thm:var} establishes the consistency of $\widehat \Psi_N$ and $\widehat{J}_N$ for cell-level scores. Applying these results to the normalised cell-level score $\psi^{\rm H}$, together with $\mathbb E_NM_{\bm i}\overset{p}{\to}m_0$ by Lemma~\ref{lemma:4_3}, gives $\|\hat{\bar{\Psi}}_N- \bar{\Psi}_{0,n}\|\overset{p}{\to} 0$ and $  \hat{\bar{J}}_N(\widehat{\bar{\theta}})
    \overset{p}{\to} \bar{J}_0.$
Combining these facts with $\hat{\bar{\Upsilon}}\overset{p}{\to}\bar\Upsilon$, the full rank of $\bar{J}_0$, and the continuous mapping theorem yields $\left \|\hat {\bar{V}}-\bar{V}_{0,n}\right\|\overset{p}{\to} 0$.
\end{proof}

\subsection{Algorithm and implementation}
\label{app:implementation}

\subsubsection{Algorithm}

Consider a generic over-identified moment-restriction model with high-dimensional nuisance parameters and two-way clustering, in which cell $(i,j)$ contains $M_{ij}\in\mathbb{N}_0$ unit observations. Algorithm~\ref{alg:full_sample_dml_gmm} gives the full-sample DML--GMM procedure with two-way cluster-robust inference.

\begin{algorithm}[t]
\caption{Full-sample two-step DML--GMM with two-way clustered inference}
\label{alg:full_sample_dml_gmm}
\begin{algorithmic}[1]

\Require Observations $\{X_{ijm}:i=1,\ldots,N_1,\ j=1,\ldots,N_2,\
m=1,\ldots,M_{ij}\}$, where $M_{ij}\in\mathbb{N}_0$; a $q$-dimensional orthogonal score $\psi(X,\theta,\eta)$. Set $\bar M=\sum_{i=1}^{N_1}\sum_{j=1}^{N_2}M_{ij}$ and $ n=\min\{N_1,N_2\}$. Choose the variance estimator $c\in \{\rm PSD, CGM\}$, where PSD refers to \eqref{eq:var_est} and CGM to the one with the double-counting term removed. 

\State Choose suitable machine-learning methods for the components of
$\eta_0$, tune them, and refit the selected learners on the full sample to obtain $\widehat\eta$.

\State Define the cell score
$\bar s_{ij}(\theta)=\sum_{m=1}^{M_{ij}}\psi(X_{ijm},\theta,\widehat\eta)$ for non-empty cells, and $0$ otherwise.

\State Define the full-sample average moment $ \widehat {\bar{\psi}}(\theta) =\frac{1}{\bar M}\sum_{i=1}^{N_1}\sum_{j=1}^{N_2}\bar s_{ij}(\theta)$,
and, for any positive-definite $q\times q$ matrix $\Upsilon $,
\[
    \widehat\theta(\Upsilon)
      =\arg\min_{\theta\in\Theta}
        \widehat {\bar{\psi}}(\theta)'\Upsilon \widehat {\bar{\psi}}(\theta).
\]

\State Compute the first-step estimate $\widehat\theta^{(1)}=\widehat\theta(I_q)$ using the identity weight.

\State At $\widehat\theta^{(1)}$, form the cluster sums
$ S_{i\cdot}^{(1)}
       =\sum_{j=1}^{N_2}\bar s_{ij}(\widehat\theta^{(1)})$ and $
    S_{\cdot j}^{(1)}
       =\sum_{i=1}^{N_1}\bar s_{ij}(\widehat\theta^{(1)})$, and construct the first-step middle matrix
\[
\widehat\Psi^{(1)}_{\rm PSD}
  =\frac{n}{\bar M^2}
    \left(
      \sum_{i=1}^{N_1}S_{i\cdot}^{(1)}S_{i\cdot}^{(1)\prime}
      +
      \sum_{j=1}^{N_2}S_{\cdot j}^{(1)}S_{\cdot j}^{(1)\prime}
    \right).
\]
For the CGM-type estimator, instead set $\widehat\Psi^{(1)}_{\rm CGM}
  = \widehat\Psi^{(1)}_{\rm PSD}
    - \frac{n}{\bar M^2}\sum_{i,j} {\bar s_{ij}(\widehat\theta^{(1)})\bar s_{ij}(\widehat\theta^{(1)})'}.$

\State Construct the feasible second-step weight
$ \widehat\Upsilon=
  \bigl(\widehat\Psi^{(1)}_{c}\bigr)^{-1}$ and compute $\widehat\theta^{(2)}
      =\widehat\theta(\widehat\Upsilon).$

\State  Re-evaluate the cell scores and cluster sums at
$\widehat\theta^{(2)}$ and compute $\widehat\Psi^{(2)}_c$.

\State Estimate the $q\times  {d}$ Jacobian and its GMM transformation:
\[
    \widehat J =-\frac{1}{\bar M} \sum_{i=1}^{N_1}\sum_{j=1}^{N_2} \sum_{m=1}^{M_{ij}}\partial_\theta\psi(X_{ijm},\widehat\theta^{(2)},\widehat\eta),
    \qquad
    \widehat A
       =(\widehat J'\widehat\Upsilon\widehat J)^{-1}
          \widehat J'\widehat\Upsilon.
\]

\State Output $\widehat \theta^{(2)}$ and $
    \widehat{\operatorname{Var}}(\widehat\theta^{(2)})
      =\frac{\widehat A\widehat\Psi^{(2)}_c\widehat A'}{n}.
$

\end{algorithmic}
\end{algorithm}

\subsubsection{Implementation in statistical software}
Existing DML packages (e.g., \texttt{DoubleML} in \texttt{R} and \texttt{Python}, \texttt{ddml} in \texttt{Stata}) provide important parts of the required
procedure, such as machine-learning wrappers, model-specific orthogonal
scores, cross-fitting or full-sample estimation, or clustered covariance estimators. To our knowledge, however, none provides the particular combination needed here: nuisance learners fitted on the full sample; an arbitrary user-supplied score vector, which may define a generic over-identified GMM model; and a variance estimator admitting an arbitrary number of clustering dimensions. 

We therefore provide R and Python packages, both named \texttt{fullsampleDML}, which integrate full-sample nuisance learning, user-defined moments, two-step GMM estimation, and arbitrary-$K$-way clustered inference. The PSD estimator is the default, and the CGM-type estimator is available as an option. The packages support heterogeneous cell sizes, missing cells, nonrectangular observed arrays, and nonnumeric cluster identifiers. They are openly available on Zenodo, together with a worked synthetic example, tests, and installation instructions.\footnote{Zenodo DOIs: 10.5281/zenodo.21990288 (R package) and 10.5281/zenodo.21990697 (Python package).}

The user supplies the data, clustering variables, nuisance-learning
specifications, score functions, and a starting parameter value (for nonlinear GMM initialisation); an analytic Jacobian is optional, as the package can compute a 
numerical one. 

The core interface is learner-agnostic. If tuning groups are not specified, the supplied clustering dimension with the fewest clusters is used; strict multiway cluster-aware tuning can be requested. After tuning folds select the hyperparameters, the selected learner is refitted on the full sample. The package then computes the identity-weighted
first-step estimate, constructs the selected multiway clustered middle
matrix, computes the feasible second-step GMM estimate, and reports
cluster-robust inference. Further implementation details and complementary examples are given in the package releases.

\subsection{Proofs of the main results}

\subsubsection{Proof of Theorem~\ref{thm:gmm}}\label{sec:Proof of Theorem gmm}
\quad\\
Fix an arbitrary sequence $\{P_n\}$ with $P_n \in \mathcal{P}_n$ satisfying the conditions of the theorem; all statements below are along this sequence, so that, by the uniformity convention of Section~\ref{sec:dml_gmm_main}, the conclusions hold uniformly over $\{\mathcal{P}_n\}$. Since $\hat\theta$ is consistent for $\theta_0$ and $\theta_0$ lies in the interior of $\Theta$, with probability approaching one, $\hat\theta$ is an interior solution of \eqref{gmm_estimator} and satisfies the first-order condition
\begin{align}
    0=\hat J_N(\hat\theta)' \hat\Upsilon \hat \psi_N(\hat\theta) \label{eq:foc} 
\end{align}
where $\hat J_N(\hat\theta) := -\partial_\theta \hat \psi_N(\theta)\big|_{\theta=\hat\theta}$ is the empirical Jacobian. 

\textbf{Step 1}. We show consistency of the Jacobian $\hat J_N(\hat\theta)$ using the decomposition $\Vert\hat J_N(\hat\theta) - J_0 \Vert \leq \Vert \hat J_N(\hat\theta) -  \hat J_N(\theta_0)\Vert + \Vert  \hat J_N(\theta_0) - J_0\Vert$. For $\Vert \hat J_N(\hat\theta) -  \hat J_N(\theta_0)\Vert$,
\begin{align*}
    \Vert \hat J_N(\hat\theta) -  \hat J_N(\theta_0)\Vert =& \left\Vert \frac{1}{N} \sum_{\bm i \in [\bm N]} \partial_\theta\psi(X_{\bm i},\hat\theta,\widehat\eta) - \partial_\theta\psi(X_{\bm i},\theta_0,\widehat\eta) \right\Vert \\
    \leq &  \frac{1}{N}\sum_{\bm i \in [\bm N]} \left\Vert \partial_\theta\psi(X_{\bm i},\hat\theta,\widehat\eta) - \partial_\theta\psi(X_{\bm i},\theta_0,\widehat\eta) \right\Vert \\
    \leq &\frac{1}{N}\sum_{\bm i \in [\bm N]} B(X_{\bm i},\widehat\eta) \Vert \hat\theta-\theta_0\Vert^\alpha =O_{P_n}(1)o_{P_n}(1), 
\end{align*}
where the last line follows from Assumptions~\ref{assu:gmm_reg}(iv) and (viii), $\hat\theta \overset{p}{\to}\theta_0$, and Lemma~\ref{lemma:4_3}. 

For $\Vert  \hat J_N(\theta_0) - J_0\Vert$, note first that $J_0=-\mathbb{E}[\partial_\theta\psi(X,\theta_0,\eta_0)]$: by Assumption~\ref{assu:gmm_reg}(iv), for every $\theta\in\mathcal{N}(\theta_0)$ we have the almost-sure bound $\|\partial_\theta\psi(X,\theta,\eta_0)\|\le\|\partial_\theta\psi(X,\theta_0,\eta_0)\|+B(X,\eta_0)\|\theta-\theta_0\|^{\alpha}$, whose right-hand side is integrable by Assumptions~\ref{assu:gmm_reg}(iii) and (iv), so the derivative and the expectation in the definition of $J_0$ in Assumption~\ref{assu:gmm_reg}(vii) can be interchanged by dominated convergence. We then apply Lemma~\ref{lemma:4_3} again, under the continuity of the score (Assumption~\ref{assu:gmm_reg}(viii)) and ${\mathbb{E} [\sup_{\eta\in\Gamma_n(\eta_0)}\|\partial_\theta\psi(X,\theta_0,\eta)\|^{1+\delta}]<\infty}$ ensured by Assumptions~\ref{assu:gmm_reg}(ii) and (iii). Hence $\Vert\hat J_N(\hat\theta) - J_0 \Vert=o_{P_n}(1)$.

\noindent\textbf{Step 2}. Since the score is continuously differentiable in $\theta$, the fundamental theorem of calculus gives, for the summands of $\hat \psi_N(\hat\theta)$,
\begin{equation}  
\psi(X,\widehat\theta,\widehat\eta)-\psi(X,\theta_0,\widehat\eta)
=\left\{\int_0^1\partial_\theta\psi\bigl(X,\theta_0+t(\widehat\theta-\theta_0),\widehat\eta\bigr)\,dt\right\}(\widehat\theta-\theta_0).
\label{eq:mv_expansion}
\end{equation}
Define the integrated empirical Jacobian
\[
\bar J_N:=-\int_0^1\mathbb E_N\!\left[\partial_\theta\psi\bigl(X,\theta_0+t(\widehat\theta-\theta_0),\widehat\eta\bigr)\right]dt.
\]
Then $\widehat\psi_N(\widehat\theta)=\widehat\psi_N(\theta_0)-\bar J_N(\widehat\theta-\theta_0)$, and \eqref{eq:foc} yields
\begin{align*}
\sqrt n(\widehat\theta-\theta_0)
=&\left\{\widehat J_N(\widehat\theta)'\widehat\Upsilon\bar J_N\right\}^{-1}
\widehat J_N(\widehat\theta)'\widehat\Upsilon
\Bigl\{\sqrt n\,\mathbb E_N[\psi(X,\theta_0,\eta_0)] +\Gn(f(\widehat\eta))+\sqrt n\,\mathbb E_{P_n}[f(\widehat\eta)]\Bigr\},
\end{align*}
where we recall that $\Gn\left(f(\eta)\right): =\sqrt{n}\left( \mathbb{E}_N [  f(\eta)] - \mathbb{E}[  f(\eta)]\right) $ and $f(\eta) =  \psi(X,\theta_0,{\eta}) - \psi(X,\theta_0,{\eta}_0)$. Moreover, by Assumption~\ref{assu:gmm_reg}(iv),
\begin{equation*}
\|\bar J_N-\widehat J_N(\theta_0)\|
\leq\int_0^1\mathbb E_N[B(X,\widehat\eta)]t^\alpha
\|\widehat\theta-\theta_0\|^\alpha\,dt =\frac{\mathbb E_N[B(X,\widehat\eta)]}{1+\alpha}
\|\widehat\theta-\theta_0\|^\alpha=o_{P_n}(1).
\end{equation*}
Hence $\bar J_N\overset p\to J_0$. Given this, together with $\widehat J_N(\widehat\theta)\overset p\to J_0$, $\widehat\Upsilon\overset p\to\Upsilon$ and Assumption~\ref{assu:gmm_reg}(vii), it remains to show a CLT for $\mathbb{E}_N [\psi(X,\theta_0,{\eta}_0)]$ and to bound $\Gn(f(\widehat\eta))+ \sqrt{n} \mathbb{E}[  f(\eta)]_{\eta = \widehat\eta}$.

\noindent\textbf{Step 2-1}. We first bound $ \mathbb{E}[  f(\eta)]_{\eta = \widehat\eta}$ using the orthogonality condition. Recall the nuisance moment map $m_{P_n}(\eta)=\mathbb E_{P_n}[\psi(X,\theta_0,\eta)]$ of \eqref{eq:m_map} and write $h{:}=\widehat\eta-\eta_0$. On the event $\widehat\eta\in\Gamma_n(\eta_0)$, convexity of $\Gamma$ ensures that $\eta_0+th\in\Gamma$. By Assumption~\ref{assu:gmm_reg}(ii), the fundamental theorem of calculus and the Neyman orthogonality condition \eqref{eq:orthogonal},
\begin{align*}
\|\mathbb E_{P_n}[f(\widehat\eta)]\|
&=\left\|\int_0^1 D_\eta m_{P_n}(\eta_0+th)[h]dt\right\|=\left\|\int_0^1\{D_\eta m_{P_n}(\eta_0+th)-D_\eta m_{P_n}(\eta_0)\}[h]dt\right\|\\
&\leq\int_0^1\|D_\eta m_{P_n}(\eta_0+th)-D_\eta m_{P_n}(\eta_0)\|_{\rm op}\,\|h\|_{P_n,2}\,dt\leq\frac{C_4}{2}\|\widehat\eta-\eta_0\|_{P_n,2}^2.
\end{align*}
Assumption~\ref{assu:gmm_reg}(v) therefore implies
\[
\sqrt n\,\|\mathbb E_{P_n}[f(\widehat\eta)]\|
\leq\frac{C_4\sqrt n}{2}\|\widehat\eta-\eta_0\|_{P_n,2}^2=o_{P_n}(1).
\]

\noindent\textbf{Step 2-2}. Next, we bound $\Gn(f(\widehat\eta))$ using the maximal inequality for multiway clustered data. Recall the localisation set $\Gamma_n(\eta_0)= \{\eta\in\mathcal{H}_n: \Vert \eta-\eta_0\Vert_{P_n,2} \le C_1n^{-1/4}\}$ of Assumption~\ref{assu:gmm_reg}; by Assumption~\ref{assu:gmm_reg}(v), $\widehat\eta\in \Gamma_n(\eta_0)$ with probability approaching one. Recall also that $
\mathcal F_n=\{f_j(\eta)-\mathbb E_{P_n}[f_j(\eta)]:\eta\in\Gamma_n(\eta_0),\ j\in[q]\}$ and ${H}_{\bm{N}}^{\bm{e}}(f)
		=\bigl|I_{\bm{N},\bm{e}}\bigr|^{-1} \sum_{\bm{i}\in I_{\bm{N},\bm{e}}}
		\pi_{\bm{e}}f\Bigl(\{U_{\bm{i}\odot \bm{e}'}\}_{\bm{e}'\le \bm{e}}\Bigr)$, where $I_{\bm{N},\bm{e}} = \{\bm{i}\odot \bm{e} : \bm{i}\in [\bm{N}]\}$ and $\bigl|I_{\bm{N},\bm{e}}\bigr| = \prod_{k'\in\supp(\bm{e})} N_{k'}$. By the Hoeffding decomposition of Section~\ref{sec:maximal_inequalities_SE} and the triangle inequality, with probability approaching one,
\begin{align*}
\|\Gn(f(\widehat\eta))\|
&\leq \sqrt q\max_{j\in[q]}\sup_{\eta\in\Gamma_n(\eta_0)}
\left|\sqrt n\,\mathbb E_N\{f_j(\eta)-\mathbb E_{P_n}[f_j(\eta)]\}\right|\leq\sqrt{qn}\sum_{k=1}^K\sum_{\bm e\in\mathcal E_k}
\sup_{f\in \calF_n}\|H_{\bm N}^{\bm e}(f)\|.
\end{align*}

For each $\bm{e}\in\mathcal{E}_k$ and $k=1,...,K$, Theorem~\ref{theorem:local_max_ineq}, which is nonasymptotic and is applied here with $P=P_n$ for each $n$, yields, with $\sigma_{\bm e}$, $\delta_{\bm e}$ and $M_{\bm e}$ as defined there (all computed under $P_n$; $\sigma_{\bm e}$ is chosen in Step 3 below),
$$\sqrt{n}\,\mathbb{E}\left[\sup_{f\in \mathcal{F}_{n}}\left\Vert {H}_{\bm{N}}^{\bm{e}}(f)\right\Vert\right]\lesssim \sqrt{\frac{n}{|I_{\bm{N,e}}|}} \left(J_{\bm{e}}(\delta_{\bm{e}}) \Vert P_{\bm{e}} F_n \Vert_{P_n,2} + \frac{J^2_{\bm{e}}(\delta_{\bm{e}})\Vert M_{\bm{e}}\Vert_{P_n,2}}{\sqrt{n}\delta^2_{\bm{e}}}\right)=:\rho_{n,\bm{e}}.$$ 
Combining this bound with Steps 1, 2 and 2-1, and applying Markov's inequality to each supremum, we obtain
\begin{equation}
        \sqrt{n}(\hat{\theta}-\theta_0) = \left( J_0' \Upsilon  J_0\right)^{-1} J_0' \Upsilon \sqrt{n} \mathbb{E}_N [\psi(X,\theta_0,{\eta}_0)] + \sum_{k=1}^K\sum_{\bm{e}\in \mathcal{E}_k} O_{P_n}(\rho_{n,\bm{e}}) +o_{P_n}(1)  \label{eq:asym_linear}
    \end{equation}
where the $O_{P_n}(\rho_{n,\bm{e}})$ terms arise from the localised empirical process $\Gn(f(\widehat\eta))$, and the $o_{P_n}(1)$ term collects the orthogonality bias of Step 2-1 and the error from replacing the estimated sandwich matrices $\{\widehat J_N(\widehat\theta)'\widehat\Upsilon\bar J_N\}^{-1}\widehat J_N(\widehat\theta)'\widehat\Upsilon$ by their limits; the latter error is $o_{P_n}(1)O_{P_n}(1)$ because $\sqrt n\,\mathbb E_N[\psi(X,\theta_0,\eta_0)]=O_{P_n}(1)$ by Assumption~\ref{assu:gmm_reg}(vi) and the variance computation below.

\noindent\textbf{Step 3}. For the first statement, Assumption~\ref{assu:gmm_reg}(ii) gives
\begin{align*}
   \sup_{\eta\in \Gamma_n}   \mathbb{E}\|f(\eta)\|^2 = \sup_{\eta\in \Gamma_n}   \mathbb{E} \|\psi(X,\theta_0,{\eta}) - \psi(X,\theta_0,{\eta}_0)\|^2 \leq  \sup_{\eta\in \Gamma_n}C_2\Vert \eta - \eta_0\Vert^2_{P_n,2}= O(n^{-1/2})
\end{align*}
It follows that $\sup_{\eta\in \Gamma_n} \mathbb{E}\|f(\eta) -  \mathbb{E}[ f(\eta)]\|^2 =O(n^{-1/2})$. Moreover, ${H}_{\bm{N}}^{\bm{e}}(f)$ is a linear combination of $\{P_{\bm{e}'}f:\bm{e}'\le\bm{e}\}$, and, by Jensen's inequality and the law of iterated expectations,
\begin{align*}
    \sup_{f\in\mathcal{F}_{n}}  \mathbb{E}{|P_{\bm{e}}f|^2} \lesssim \sup_{\eta\in \Gamma_n} \mathbb{E}\|f(\eta) -  \mathbb{E}[ f(\eta)]\|^2 =O(n^{-1/2}).
\end{align*}

Therefore, we can take $\sigma_{\bm e}=\sigma_{{\bm e},n}$ to be a sequence of positive numbers with $\sigma_{\bm e} = O(n^{-1/4})$; specifically, $\sigma_{\bm e}=\sqrt{C_2}C_1 n^{-1/4}$, which satisfies $\sigma_{\bm e}\le 1\le \|P_{\bm e}F_n\|_{P_n,2}$ for all $n$ sufficiently large since $F_n\ge1$. If $\mathcal{F}_{n}$ is a VC-type class with characteristics $A_n\ge {(e^{2(K-1)}/16)\vee e}$ and $v_n\ge1$, then, since $\|P_{\bm e}F_n\|_{P_n,2}/\sigma_{\bm e}\lesssim n^{1/4}\|F_n\|_{P_n,2}\le A_n\vee\overline{N}$ by the envelope-growth condition of the theorem, the second display of Corollary~\ref{cor:local_maximal_ineq_VC} applies (again with $P=P_n$) and gives, for all $\bm{e}\in \calE_k$ and $k=1,...,K$,
\begin{align*}
   & \sqrt{n} \mathbb E_{P_n}\left[\sup_{f\in\mathcal{F}_{n}}\|{H}_{\bm{N}}^{\bm{e}}(f)\|\right]= \sqrt{\frac{n}{|I_{\bm{N},\bm{e}}|}} |I_{\bm{N},\bm{e}}|^{1/2}  \mathbb{E}\left[\sup_{f\in\mathcal{F}_{n}}\|{H}_{\bm{N}}^{\bm{e}}(f)\|\right] \\
   \lesssim& \frac{n^{1/4}}{|I_{\bm{N},\bm{e}}|^{1/2}}  \{v_n\log(A_n\vee \overline{N} )\}^{k/2} + \frac{\|M_{\bm{e}}\|_{P_n,2}}{|I_{\bm{N},\bm{e}}|^{1/2}} \{v_n\log(A_n\vee \overline{N} )\}^{k} ,
\end{align*}
where $M_\e=\max_{t\in [n]}(P_\e F_n)\left(\{U_{(t,...,t)\odot \e'}\}_{\e'\le \e}\right) $, and, for some $r>2$, since $P_\e F_n\ge 0$, bounding the $r$-th moment of the maximum by the sum of the $r$-th moments and applying the conditional Jensen inequality,
\begin{align*}
    \|M_{\bm{e}}\|_{P_n,2} \leq \|M_{\bm{e}}\|_{P_n,r} \leq   \left(   \sum_{t\in [n]}\mathbb{E} \left|  (P_\e F_n)\left(\{U_{(t,...,t)\odot \e'}\}_{\e'\le \e}\right) \right|^r \right)^{1/r}\leq \left(n\,\mathbb{E}[F_n^r]\right)^{1/r} = n^{1/r} \left\|  F_n  \right\|_{P_n,r}.
\end{align*}
Then, by Markov's inequality, for all $\bm{e}\in\calE_k$ and $k=1,\ldots,K$,
\begin{align}
    \sqrt{n}\sup_{f\in \mathcal{F}_{n}}\left\Vert {H}_{\bm{N}}^{\bm{e}}(f)\right\Vert = O_{P_n}\left( \sqrt{ \frac{n^{1/2}\{v_n\log(A_n\vee \overline{N} )\}^{k}}{|I_{\bm{N},\bm{e}}|}} + \frac{n^{1/r}\left\|  F_n  \right\|_{P_n,r} \{v_n\log(A_n\vee \overline{N} )\}^{k}}{|I_{\bm{N},\bm{e}}|^{1/2}}\right) \nonumber \\
= O_{P_n}\left(\left(\frac{v_n\log(A_n\vee \overline{N} )}{ {n^{1-1/(2k)}}}\right)^{k/2} + \left(\frac{\left\|  F_n  \right\|_{P_n,r} \{v_n\log(A_n\vee \overline{N} )\}}{n^{1/2-1/r}}\right)^{k}\right), \label{eq:vc_rate}
\end{align}
where the second line uses $|I_{\bm{N},\bm{e}}|\ge n^k$, $\|F_n\|_{P_n,r}\ge1$ (as $F_n\ge1$), $n^{1/r}\le n^{k/r}$ and $k\ge1$,
which proves the first statement. 

For the second statement, Lemma~\ref{lemma:hajek} gives the independent linear representation (where $I_1=\bigcup_{\bm e\in\calE_1}I_{\bm N,\bm e}$ and $k(\bm i)$ denotes the coordinate in which $\bm i\in I_1$ is non-zero, as defined before Lemma~\ref{lemma:hajek}),
\begin{align}
    &\sqrt{n}\mathbb{E}_N[\psi(X_{\bm i},\theta_0,\eta_0)] = \sum_{\bm i \in I_1} \frac{\sqrt{n}}{N_{k(\bm i)}} \mathbb{E}\left[\psi(X_{\bm i},\theta_0,\eta_0) | U_{\bm i}\right] + O_{P_n}(n^{-1/2})  \label{eq:hajek}  \\
    =& \frac{\sqrt{n}}{N_1}\sum_{i_1=1}^{N_1}  \mathbb{E}[\psi(X_{\bm i},\theta_0,\eta_0)| {U_{(i_1,0,...,0)}}] + ...+  \frac{\sqrt{n}}{N_K}\sum_{i_K=1}^{N_K}  \mathbb{E}[\psi(X_{\bm i},\theta_0,\eta_0)| {U_{(0,0,...,i_K)}}] +O_{P_n}(n^{-1/2}) \nonumber
\end{align}
and the variance of  $\sqrt{n}\mathbb{E}_N[\psi(X_\i,\theta_0,\eta_0)]$,
\begin{align*}
     &  {\rm Var} \left( \sqrt{n}\mathbb{E}_N[\psi(X_{\bm i}, \theta_0,\eta_0)] \right) =\sum_{\bm e\in\mathcal{E}_1}  {\frac{n}{N_{k(\bm e)}}} {\rm Cov}(\psi(X_{\bm 1},\theta_0,\eta_0),\psi(X_{\bm 2-\bm e},\theta_0,\eta_0)) + O(n^{-1}) \nonumber \\
     =& {(\mu_1+o(1))} {\rm Var}( \mathbb{E}[\psi(X,\theta_0,\eta_0)|U_{1,0,...,0}]) + ... +{(\mu_K+o(1))} {\rm Var}( \mathbb{E}[\psi(X,\theta_0,\eta_0)|U_{0,0,...,1}]) \\
     &+ O(n^{-1}) = \Psi_{0,n}+o(1),
\end{align*}
where $k(\bm e)$ denotes the coordinate in which $\bm{e}\in \mathcal{E}_1$ is non-zero, and the second equality uses $n/N_k=\mu_k+o(1)$ (which accommodates $\mu_k=0$) together with the uniform boundedness of the conditional variances implied by Assumption~\ref{assu:gmm_reg}(vi) and Jensen's inequality.

The terms $\mathbb{E}\left[\psi(X_{\bm i},\theta_0,\eta_0) | U_{\bm i}\right]$ are i.i.d.\ within each of the $K$ sums in \eqref{eq:hajek}, and the $K$ sums are mutually independent, being functions of disjoint collections of the i.i.d.\ AHK variables. Given \eqref{eq:asym_linear}, \eqref{eq:vc_rate}, the rate condition \eqref{eq:rate_condition}, the finite $2+
\delta$ moment of $\psi(X_{\bm i},\theta_0,\eta_0)$ in Assumption~\ref{assu:gmm_reg}(vi) and that $\psi(X_{\bm i},\theta_0,\eta_0)$ is mean-zero, we proceed as follows. Fix any $t\in\mathbb R^d$ with $\|t\|=1$ and set $s_n':=t' V_{0,n}^{-1/2}\left(J_0'\Upsilon J_0\right)^{-1}J_0'\Upsilon\in\mathbb R^{1\times q}$, which is well defined for $n$ large by condition (b) and Assumption~\ref{assu:gmm_reg}(vii), with $\|s_n\|$ bounded above and, by the variance display above, 
$${\rm Var}(s_n'\sqrt n\,\mathbb E_N[\psi(X,\theta_0,\eta_0)])=s_n'\Psi_{0,n}s_n+o(1)=t't+o(1)=1+o(1).$$ Applying Lyapunov's CLT to the scalar triangular array $s_n'\sqrt n\,\mathbb E_N[\psi(X,\theta_0,\eta_0)]$, decomposed via \eqref{eq:hajek} into the $K$ mutually independent sums of i.i.d.\ summands with uniformly bounded $2+\delta$ moments, yields $s_n'\sqrt n\,\mathbb E_N[\psi(X,\theta_0,\eta_0)]\overset{d}{\to}N(0,1)$. Since $t$ was arbitrary, the Cram\'er--Wold device gives
\begin{align*}
  V_{0,n}^{-1/2}\left(J_0'\Upsilon J_0\right)^{-1}J_0'\Upsilon\,\sqrt{n}\,\mathbb{E}_N[\psi(X_{\bm i},\theta_0,\eta_0)] \overset{d}{\to} N\left(0,I_d\right),
\end{align*}
and the second statement follows by combining this with \eqref{eq:asym_linear} and Slutsky's theorem.
\qed

\vskip 0.15in

\subsubsection{Proof of Theorem~\ref{eg:vctype}}\label{sec:Proof of Theorem vctype}
\quad\\
   \noindent \textbf {Case (1).} For any index set $I\subset\{1,...,p\}$ with $|I| =s $, let $\mathcal{G}_I = \{x\mapsto  {\Lambda}(x^\top_I\beta_I):\beta_I\in \mathbb{R}^{|I|},\Vert\beta_I \Vert_2{\le B_n} \}$, a subclass of $\mathcal{G}_n$. Since $|x_I^\top\beta_I|\le B_n\max_{|I|=s}\|x_I\|_2<\infty$ for every $x$ and $\Lambda$ is monotone, the class $\mathcal{G}_n$ admits the finite envelope $G(x):=\sup_{|u|\le B_n\max_{|I|=s}\|x_I\|_2}|\Lambda(u)|<\infty$; note that the covering bound below, being a VC-subgraph bound relative to the envelope, does not depend on $B_n$. The functions $x\mapsto x_I^\top\beta_I$ form a subset of an $s$-dimensional vector space of functions and hence a VC-subgraph class of index at most $s+2$ (Lemma 2.6.15 of \citealt{vdVW1996}), and composition with the monotone link $\Lambda$ preserves the VC-subgraph property with the same index bound (Lemma 2.6.18(viii) of \citealt{vdVW1996}). Therefore, by Theorem 2.6.7 of \citealt{vdVW1996}, for any finite discrete measure $Q$ and some universal constant $C<\infty$,
    \begin{align*}
        N(\mathcal{G}_I,\Vert \cdot \Vert_{Q,2},\varepsilon\Vert G\Vert_{Q,2}) \leq \left(\frac{C}{\varepsilon}\right)^{ {2(s+1)}}.
    \end{align*}
    Since $\mathcal{G}_n =  \bigcup_{|I|=s}\mathcal{G}_I$, we can bound the covering number as
    \begin{align*}
        N(\mathcal{G}_n,\Vert \cdot \Vert_{Q,2},\varepsilon\Vert G\Vert_{Q,2}) \leq&  \binom{p}{s} \left(\frac{C}{\varepsilon}\right)^{ {2(s+1)}}
        \leq \left(\frac{ep}{s}\right)^s\left(\frac{C}{\varepsilon}\right)^{ {2(s+1)}}
        \le \left(\frac{C \ ep}{s\varepsilon}\right)^{ {2(s+1)}},
    \end{align*}
     using $s\le 2(s+1)$ and $ep/s\ge1$ in the last step.
    Therefore, we can choose $v_n=  {2(s+1)}$ and $A_n = \frac{C \ ep}{s}\vee {((e^{2(K-1)}/16)\vee e)}$. Given this choice and ${s\log(p/s)} \vee s\log(\overline{N})= o(n^{1/4})$, we have  $ \frac{v_n \log(A_n\vee \overline{N})}{ {n^{1-1/(2k)}}}  = o(1)$ for all $k\geq 1$ and $ \frac{v_n \log(A_n\vee \overline{N})}{n^{1/2-1/r}}  = o(1)$ for $r> 4$, which proves Case (1).
\vskip 0.15in

   \noindent \textbf {Case (2).} Fix any $g\in \mathcal{G}_n$ and $u\in\mathbb{R}$. 
   We can build a decision tree with the splitting rule $t(x) = 1\{g(x)>{u}\}$; the corresponding class is
    \begin{align*}
       \mathcal{H}_{{n}} = \{(x,{u})\mapsto 1\{g(x)>{u}\}: g\in \mathcal{G}_n, \ {u}\in\mathbb{R}\}
    \end{align*} 
    The VC dimension of the subgraph of $\mathcal{G}_n$ is equal to ${\rm VC}(\mathcal{H}_{{n}})$, and the decision tree associated with $\mathcal{H}_{{n}}$ has $2L$ partitions. 
        
    By \cite{leboeuf2022}, the class of classification functions induced by a binary decision tree with partition $\{R_l\}_{l=1}^{2L}$ defined in the statement has VC dimension of order $O(L\log(2Lp))$, where $p$ is the dimension of the real-valued features. Therefore, $\mathcal{G}_n$ is a VC-subgraph class with index $V =CL\log(2Lp)$ for some constant $C\in (0,\infty)$. Since $|\mu_l|\le M_n$ for all $l$, we can take $G\equiv M_n$ as the envelope of $\mathcal{G}_n$. By Theorem 2.6.7 in \cite{vdVW1996}, 
    \begin{align*}
        & N(\mathcal{G}_n,\Vert \cdot \Vert_{Q,2},\varepsilon\Vert G\Vert_{Q,2}) \leq {C_0} V(16e)^V (1/\varepsilon)^{2(V-1)}\\
         \lesssim & C{C_0} L\log(2Lp) (16e/\varepsilon)^{2CL\log(2Lp)} \\
         =& {C_0}  \left([CL\log(2Lp)]^{1/(2CL\log(2Lp))}16e/\varepsilon\right)^{2CL\log(2Lp)} 
    \end{align*}
    for a universal constant {$C_0$
    }. Therefore, we can take $v_n = 2CL\log(2Lp)$ and $A_n$ to be a sufficiently large constant, because $[CL\log(2Lp)]^{1/(2CL\log(2Lp))}$ is bounded for fixed $L$ and converges to a constant as $L$ diverges. The rate condition follows immediately.
\vskip 0.15in

   \noindent \textbf {Case (3).} Let $\mathcal{G}_n^{+}$ denote the class of functions induced by the neural network defined in the statement but with $L-1$ hidden layers and $W$ parameters. Let ${\rm sgn}(\mathcal{G}_{{n}}) : = \{{\rm sgn}(g): g\in \mathcal{G}_{{n}}\}$ and ${\rm sub}(\mathcal{G}_{{n}}) = \{(x,y)\in \mathcal{X}\times\mathbb{R}: g(x)>y, g\in \mathcal{G}_n\}$. By Theorem 14.1 of \cite{anthony2009neural}, $ {\rm VC}({\rm sub}(\mathcal{G}_{{n}}))\leq{\rm VC}({\rm sgn}(\mathcal{G}_{{n}}^{+}){)}$. By Theorem 7 of \cite{Bartlett2019}, ${\rm VC}({\rm sgn}(\mathcal{G}_{{n}}^{+})) = O(LW\log(pU))$. Then, by Theorem 2.6.7 of \cite{vdVW1996} and a calculation similar to that in Case (2), we can choose $A_n= C\vee{((e^{2(K-1)}/16)\vee e)}$ for some sufficiently large constant $C$ and $v_n = 2{C}LW\log(pU)$, and the rate condition follows.  \qed

\vskip 0.15in

\subsubsection{Proof of Theorem~\ref{thm:var}}\label{sec:Proof of Theorem var}

Fix an arbitrary sequence $\{P_n\}$ with $P_n \in \mathcal{P}_n$ satisfying the conditions of the theorem; all statements below are along this sequence, so that, by the uniformity convention of Section~\ref{sec:dml_gmm_main}, the conclusions hold uniformly over $\{\mathcal{P}_n\}$. Given the consistency of $\hat\Upsilon$ for some positive definite $\Upsilon$, the consistency of $\hat{V}$ follows from (1) the consistency of $\hat{J}_N(\hat\theta)$, (2) the invertibility of $J_0$ and (3) the consistency of $\hat{\Psi}_N(\hat\theta)$. As (1) is established in the proof of Theorem~\ref{thm:gmm} and (2) is assumed, it remains to show (3). Since $K$ is finite, it suffices to show, for each $k=1,...,K$, that
$$
\left\|\hat\Psi_{N,k}(\hat\theta) -\mu_k{\rm Cov}(\psi(X_{\bm 1},\theta_0,\eta_0),\psi(X_{\bm 2-\bm{e}_k},\theta_0,\eta_0))  \right\| {= o_{P_n}(1),}
$$
where we recall that $\bm e_k \in \mathcal{E}_1$ is the $K$-dimensional vector whose entries are all zero except the $k$-th. 
    
We use the decomposition
    \begin{align*}
        \left\Vert \hat{\Psi}_{N,k}(\hat\theta) - \mu_k{\rm Cov}(\psi(X_{\bm 1},\theta_0,\eta_0),\psi(X_{\bm 2-\bm{e}_k},\theta_0,\eta_0))\right\Vert \leq \Vert\mathcal{I}_1\Vert + \Vert\mathcal{I}_2\Vert,
    \end{align*}
    where
    \begin{align*}
        \mathcal{I}_1 = & \frac{n}{N^2} \sum_{\bm{i,j\in [N]}} \psi(X_{\bm i},\hat\theta,\widehat\eta)  \psi(X_{\bm j},\hat\theta,\widehat\eta)' \1_{k}\{\bm i, \bm j\} -  \frac{n}{N^2}\sum_{\bm{i,j\in [N]}} \psi(X_{\bm i},\theta_0,\eta_0)  \psi(X_{\bm j},\theta_0,\eta_0)' \1_{k}\{\bm i, \bm j\} ,\\
         \mathcal{I}_2 = & \frac{n}{N^2}  \sum_{\bm{i,j\in [N]}} \psi(X_{\bm i},\theta_0,\eta_0)  \psi(X_{\bm j},\theta_0,\eta_0)' \1_{k}\{\bm i, \bm j\}  - \mu_k \mathbb{E}(\psi(X_{\bm 1},\theta_0,\eta_0)\psi(X_{\bm 2-\bm{e}_k},\theta_0,\eta_0)')
    \end{align*}

Consider $\Vert\mathcal{I}_2\Vert$. Note that $n/N_k = {\mu_k+o(1)}$. Under Conditions (SE) and (D), Lemma~\ref{lemma:second_moment_lln} in the Auxiliary Lemmas, applied with $h_n=\psi(\cdot,\theta_0,\eta_0)$, whose $2+\delta$ moments are bounded uniformly over $P_n\in\mathcal P_n$ by Assumption~\ref{assu:gmm_reg}(vi), gives
   \begin{align*}
       \frac{N_k}{N^2}  \sum_{\bm{i,j\in [N]}} \psi(X_{\bm i},\theta_0,\eta_0)  \psi(X_{\bm j},\theta_0,\eta_0)' \1_{k}\{\bm i, \bm j\} = \mathbb{E}(\psi(X_{\bm 1},\theta_0,\eta_0)\psi(X_{\bm 2-\bm{e}_k},\theta_0,\eta_0)') + o_{P_n}(1).
   \end{align*}
   This is a triangular-array analogue of Proposition 4.1 in \cite{DDG2018}, established there for a fixed data generating process in the parametric case. Combining the two displays with the boundedness of $\mathbb{E}(\psi(X_{\bm 1},\theta_0,\eta_0)\psi(X_{\bm 2-\bm{e}_k},\theta_0,\eta_0)')$, which follows from Assumption~\ref{assu:gmm_reg}(vi) and the Cauchy--Schwarz inequality, we obtain $\Vert\mathcal{I}_2\Vert = o_{P_n}(1)$.

Consider $\Vert\mathcal{I}_1\Vert $. By a product decomposition, the triangle inequality and the Cauchy--Schwarz inequality,
   \begin{align*}
      \Vert\mathcal{I}_1\Vert  & \lesssim R_n \left\{ \left(\mathbb{E}_N\left\Vert\psi(X_{\bm i},\theta_0,\eta_0)\right\Vert^2\right)^{1/2} + R_n \right\} \\
       R_n &= \left(\mathbb{E}_N\left\Vert \psi(X_{\bm i},\hat\theta,\widehat\eta) -\psi(X_{\bm i},\theta_0,\eta_0) \right\Vert^2\right)^{1/2}=\Vert \psi(X_{\bm i},\hat\theta,\widehat\eta) -\psi(X_{\bm i},\theta_0,\eta_0)\Vert_{\mathbb{P}_N,2},
   \end{align*}
   where $\Vert \cdot\Vert_{\mathbb{P}_N,2}$ denotes the empirical $\ell^2$ norm. Lemma~\ref{lemma:4_3}, applied under the finite $2+\delta$ moment of $\psi(X_{\bm i},\theta_0,\eta_0)$, gives $\mathbb{E}_N\left\Vert\psi(X_{\bm i},\theta_0,\eta_0)\right\Vert^2 = O_{P_n}(1)$, so it suffices to show $R_n = o_{P_n}(1)$.
   
For each observation, define $ D_{N,\bm i}:=\int_0^1\partial_\theta\psi\bigl(X_{\bm i},\theta_0+t(\widehat\theta-\theta_0),\widehat\eta\bigr)\,dt$. By the expansion \eqref{eq:mv_expansion} and Minkowski's inequality,
\begin{align*}
R_n\leq{}&\|D_N\|_{\mathbb P_N,2}\|\widehat\theta-\theta_0\|+\|\psi(X,\theta_0,\widehat\eta)-\psi(X,\theta_0,\eta_0)\|_{\mathbb P_N,2}.
\end{align*}
The Hölder condition gives
\begin{align*}
\|D_N\|_{\mathbb P_N,2}
\leq{}&\|\partial_\theta\psi(X,\theta_0,\widehat\eta)\|_{\mathbb P_N,2}+\frac{\|\widehat\theta-\theta_0\|^\alpha}{1+\alpha}
\|B(X,\widehat\eta)\|_{\mathbb P_N,2}.
\end{align*}
By \eqref{eq:var_moment1}, Assumption~\ref{assu:gmm_reg}(vi), the convergence of $\widehat\eta$, and the equicontinuity condition imposed in the theorem (Assumption~\ref{assu:gmm_reg}(viii) applied to $\|\psi(x,\theta_0,\eta)-\psi(x,\theta_0,\eta_0)\|^2$), Lemma~\ref{lemma:4_3} gives
\[
\|\psi(X,\theta_0,\widehat\eta)-\psi(X,\theta_0,\eta_0)\|_{\mathbb P_N,2}^2
\overset p\longrightarrow \mathbb E_{P_n}\|\psi(X,\theta_0,\eta_0)-\psi(X,\theta_0,\eta_0)\|^2=0.
\]
   Similarly, applying Lemma~\ref{lemma:4_3} under the moment condition \eqref{eq:var_moment2}, the convergence of $\hat{\eta}$ and the equicontinuity condition for $\|\partial_\theta\psi(x,\theta_0,\eta)\|^2$ gives
   \begin{align*}
      \left\Vert{\partial_\theta \psi(X,\theta_0,\widehat\eta)}\right\Vert_{\mathbb{P}_N,2} ^2\overset{p}{\to} \mathbb{E} \left\Vert\partial_\theta \psi(X,\theta_0,\eta_0)\right\Vert^2 = O(1).
    \end{align*}

Finally, on the event $\widehat\eta\in\Gamma_n(\eta_0)$, whose probability approaches one, \eqref{eq:var_moment3}, Assumption~\ref{assu:gmm_reg}(iv), the equicontinuity condition for $B(x,\eta)^2$ and Lemma~\ref{lemma:4_3} yield
\[
\|B(X,\widehat\eta)\|_{\mathbb P_N,2}^2
\overset p\longrightarrow\mathbb E_{P_n}[B^2(X,\eta_0)]=O(1).
\]
Since $\widehat\theta\overset p\to\theta_0$, the preceding bounds show that $R_n=o_{P_n}(1)$. \qed

\subsubsection{Proof of Theorem~\ref{theorem:global_maximal_ineq}}\label{sec:Proof of Theorem global_maximal_ineq}

	By the symmetrisation inequality for SE processes (Lemma B.1 in \cite{chiang2023inference}, which is dimension-free), for independent Rademacher variables $(\varepsilon_{1,i_1}),\ldots,(\varepsilon_{k,i_k})$ that are independent of $(X_\i)_{\i\in\N^K}$,
	\begin{align*}
		|I_{\bN,\e}|^{1/2}	\left(\mathbb{E}[\|H_\bN^\e(f)\|_\calF^r]\right)^{1/r}
	=&\left (\mathbb{E} \left [ \left \|\frac{1}{\sqrt{|I_{\bN,\e}|}} \sum_{\i\in I_{\bN,\e}} (\pi_{ \e}f )(\{U_{\i \odot \e'}\}_{\e'\le \e})   \right \|_{\calF}^{r} \right ] \right)^{1/r}\\
		\lesssim& 
		\left (\mathbb{E} \left [ \left \|\frac{1}{\sqrt{|I_{\bN,\e}|}} \sum_{\i\in I_{\bN,\e}}\varepsilon_{1,i_1}...\varepsilon_{k,i_k}\cdot(\pi_{ \e}f )(\{U_{\i \odot \e'}\}_{\e'\le \e}) \right \|_{\calF}^{r} \right ] \right)^{1/r}.
	\end{align*}
	By the convexity of the supremum and of $\cdot \mapsto (\cdot)^r$, Jensen's inequality implies that, up to a constant depending only on $r$, $K$ and $k$, the right-hand side above is bounded above by
	\begin{align*}
		\left (\mathbb{E} \left [ \left \|\frac{1}{\sqrt{|I_{\bN,\e}|}} \sum_{\i\in I_{\bN,\e}}\varepsilon_{1,i_1}...\varepsilon_{k,i_k}\cdot(P_\e f )(\{U_{\i \odot \e'}\}_{\e'\le \e}) \right \|_{\calF}^{r} \right ] \right)^{1/r}.
	\end{align*}
Let $\mathbb P_{I_{\bN,\e}}=|I_{\bN,\e}|^{-1}\sum_{\i \in I_{\bN,\e}}\delta_{\{U_{\i \odot \e'}\}_{\e'\le \e}}$ denote the empirical measure on the support of $\{U_{\i \odot \e'}\}_{\e'\le \e}$.
Conditionally on the AHK variables $\{U_{\i\odot \e'}:\i\in I_{\bN,\e},\ \e'\le \e\}$ (which determine $(X_\i)_{\i\in[\bN]}$), the object
\begin{align*}
\frac{1}{\sqrt{|I_{\bN,\e}|}}\sum_{\i\in I_{\bN,\e}}\varepsilon_{1,i_1}...\varepsilon_{k,i_k}(P_{\e}f )(\{U_{\i\odot \e'}\}_{\e'\le \e})
\end{align*}
is a homogeneous Rademacher chaos process of order $k$ with fixed coefficients.
	By Lemma~\ref{lemma: Orlicz norms} applied conditionally, the conditional $L^r$ norm is bounded above by the conditional $\psi_{2/k}$-norm up to a constant depending only on $(r,k)$, and, by Corollary 5.1.8 in \cite{delaPenaGine1999}, the conditional $\psi_{2/k}$-norm of the supremum is bounded by the entropy integral below, which is measurable with respect to the conditioning variables. Taking unconditional $L^r$ norms of both sides of the resulting conditional bound gives
	\begin{align*}
		&\left (\mathbb{E} \left [ \left \|\frac{1}{\sqrt{|I_{\bN,\e}|}} \sum_{\i\in I_{\bN,\e}}\varepsilon_{1,i_1}...\varepsilon_{k,i_k} \cdot(P_{ \e}f )(\{U_{\i \odot \e'}\}_{\e'\le \e})  \right \|_{\calF}^{r} \right ] \right)^{1/r}\\
		\lesssim&
		\left(\mathbb{E} \left [ \left(\int_0^{\sigma_{I_{\bN,\e}}}\left[1+\log N\left(P_{ \e}\calF,\|\cdot\|_{\mathbb P_{I_{\bN,\e}},2},\tau\right)\right]^{k/2}d \tau\right)^{r} \right ]\right)^{1/r},
	\end{align*}
	where $\sigma_{I_{\bN,\e}}^2:=\sup_{f\in\calF}\| P_\e f \|_{\mathbb P_{I_{\bN,\e}},2}^2
	$. By a change of variable and the definition of $J_\e$, pathwise,
	\begin{align*}
		&\int_0^{\sigma_{I_{\bN,\e}}}\left[1+\log N\left(P_{ \e} \calF,\|\cdot\|_{\mathbb P_{I_{\bN,\e}},2},\tau\right)\right]^{k/2}d \tau \\
		=&\|P_{ \e} F\|_{\mathbb P_{I_{\bN,\e}},2} \int_0^{\sigma_{I_{\bN,\e}}/\|P_{ \e} F\|_{\mathbb P_{I_{\bN,\e}},2}}\left[1+\log N\left(P_{ \e} \calF,\|\cdot\|_{\mathbb P_{I_{\bN,\e}},2},\tau\|P_{ \e} F\|_{\mathbb P_{I_{\bN,\e}},2}\right)\right]^{k/2}d \tau\\
		\le&
		\|P_{ \e} F\|_{\mathbb P_{I_{\bN,\e}},2}\, J_\e\left(\sigma_{I_{\bN,\e}}/\|P_{ \e} F\|_{\mathbb P_{I_{\bN,\e}},2}\right)
		\le J_\e\left(1\right)\|P_{ \e} F\|_{\mathbb P_{I_{\bN,\e}},2},
	\end{align*}
	where the last step uses $\sigma_{I_{\bN,\e}}\le \|P_\e F\|_{\mathbb P_{I_{\bN,\e}},2}$ (since $|P_\e f|\le P_\e F$ pointwise for every $f\in\calF$) and the monotonicity of $J_\e$. Combining the two displays,
	\begin{align*}
		\left (\mathbb{E} \left [ \left \|\frac{1}{\sqrt{|I_{\bN,\e}|}} \sum_{\i\in I_{\bN,\e}}\varepsilon_{1,i_1}...\varepsilon_{k,i_k} \cdot(P_{ \e}f )(\{U_{\i \odot \e'}\}_{\e'\le \e})  \right \|_{\calF}^{r} \right ] \right)^{1/r}\\
		\lesssim J_\e(1)\left(\mathbb{E}\left[\|P_{ \e} F\|_{\mathbb P_{I_{\bN,\e}},2}^{r}\right]\right)^{1/r}\le J_\e\left(1\right) \|F\|_{P,r\vee 2} ,
	\end{align*}
	where the final inequality is justified as follows. For $r\ge 2$, by Jensen's inequality (convexity of $x\mapsto x^{r/2}$) applied to the empirical average $\mathbb P_{I_{\bN,\e}}$, the conditional Jensen inequality $(P_\e F)^r\le P_\e (F^r)$ and the identical distribution of the summands,
	$\mathbb{E}[\|P_\e F\|_{\mathbb P_{I_{\bN,\e}},2}^r]=\mathbb{E}[(\mathbb P_{I_{\bN,\e}} (P_\e F)^2)^{r/2}]\le \mathbb{E}[\mathbb P_{I_{\bN,\e}} (P_\e F)^r]=\mathbb{E}[(P_\e F)^r]\le \mathbb{E}[F^r]$.
	For $1\le r<2$, by the monotonicity of $L^r$ norms, $(\mathbb{E}[\|P_\e F\|_{\mathbb P_{I_{\bN,\e}},2}^r])^{1/r}\le (\mathbb{E}[\|P_\e F\|_{\mathbb P_{I_{\bN,\e}},2}^2])^{1/2}=\|P_\e F\|_{P,2}\le\|F\|_{P,2}$.
\qed

\subsubsection{Proof of Theorem~\ref{theorem:local_max_ineq}}\label{sec:Proof of theorem:local_max_ineq}
	
We first state a key technical lemma, which is used in the proof below and proved at the end of this section.
	\begin{lemma}\label{lemma:partition}
		For any $\e\in \{0,1\}^K\setminus\{\bm 0\}$, let $m_\e=\min_{k'\in\supp(\e)}N_{k'}$, and note $m_\e\ge n$. Then $I_{\bN,\e}$ can be partitioned into subsets $G$ of size $m_\e$ such that each $G$ is transversal, that is, any two distinct tuples
		$
		(i_1,i_2,\dots,i_K),$	$(i_1',i_2',\dots,i_K')\in G
		$
		satisfy
        $
            i_k\ne i_k' \text{ for all } k\in \supp(\bm e).
       $
	\end{lemma}

	We now prove Theorem~\ref{theorem:local_max_ineq}. Throughout, write $m_\e=\min_{k'\in\supp(\e)}N_{k'}\ge n$ as in Lemma~\ref{lemma:partition}, and, for $m\in\N$, let $M_\e^{(m)}=\max_{t\in [m]}(P_\e F)\left(\{U_{(t,...,t)\odot \e'}\}_{\e'\le \e}\right)$, so that $M_\e=M_\e^{(n)}$. We use the following elementary comparison: for i.i.d.\ non-negative random variables $Z_1,Z_2,\dots$ and any $m\ge n$,
	\begin{align}
		\mathbb{E}\left[\max_{t\in[m]}Z_t\right]\le \left\lceil \frac{m}{n}\right\rceil\mathbb{E}\left[\max_{t\in[n]}Z_t\right]\le \frac{2m}{n}\,\mathbb{E}\left[\max_{t\in[n]}Z_t\right],\label{eq:blockmax}
	\end{align}
	which follows by splitting $[m]$ into $\lceil m/n\rceil$ blocks of length at most $n$ and bounding the maximum by the sum of the block maxima. Applying \eqref{eq:blockmax} to $Z_t=(P_\e F)^2\left(\{U_{(t,...,t)\odot \e'}\}_{\e'\le \e}\right)$, for which $\max_{t\in[m]}Z_t = (M_\e^{(m)})^2$, yields, for any $m\ge n$,
	\begin{align}
		\frac{\|M_\e^{(m)}\|_{P,2}}{\sqrt{m}}\le \frac{\sqrt{2}\,\|M_\e\|_{P,2}}{\sqrt{n}}.\label{eq:blockmax2}
	\end{align}

	If $\e\in \calE_1$ with $\supp(\e)=\{k'\}$, the summands are $N_{k'}$ i.i.d.\ random variables, and Lemma~\ref{lemma:local_max_ineq_indep} with sample size $N_{k'}=m_\e$ delivers the stated bound with $\sqrt{m_\e}$ and $\|M_\e^{(m_\e)}\|_{P,2}$ in place of $\sqrt{n}$ and $\|M_\e\|_{P,2}$; the result then follows from \eqref{eq:blockmax2} with $m=m_\e$. We therefore assume $K\ge 2$ and $\e\in \calE_k$ for some $k\in\{2,...,K\} $. As before, we first establish the bound of the theorem with $\left(\sqrt{m_\e},\, M_\e^{(m_\e)}\right)$ in place of $\left(\sqrt{n},\, M_\e\right)$; since $\delta_\e$ and $\|P_\e F\|_{P,2}$ do not involve $n$, and $M_\e$ enters the bound only through $\|M_\e\|_{P,2}/\sqrt{n}$, the stated bound then follows from \eqref{eq:blockmax2} with $m=m_\e$. Assume without loss of generality that the first $k$ entries of $\e$ are $1$ and the others zero. By the symmetrisation inequality of Lemma B.1 in \cite{chiang2023inference}, for independent Rademacher r.v.'s $(\varepsilon_{1,i_1})$,...,$(\varepsilon_{k,i_k})$ independent of $(X_\i)_{\i\in\N^K}$,
	\begin{align*}
			|I_{\bN,\e}|^{1/2}\mathbb{E}[\|H_\bN^\e(f)\|_\calF]
		\lesssim&
		\mathbb{E}\left[\left\|\frac{1}{\sqrt{|I_{\bN,\e}|}}\sum_{\i\in I_{\bN,\e}}\varepsilon_{1,i_1}...\varepsilon_{k,i_k}(\pi_{\e}f )(\{U_{\i\odot \e'}\}_{\e'\le \e})\right\|_\calF\right].
	\end{align*}
	By convexity of the supremum and Jensen's inequality, the RHS is bounded, up to a constant depending only on $K$ and $k$, by
	\begin{align*}
			\mathbb{E}\left[\left\|\frac{1}{\sqrt{|I_{\bN,\e}|}}\sum_{\i\in I_{\bN,\e}}\varepsilon_{1,i_1}...\varepsilon_{k,i_k}(P_{\e}f )(\{U_{\i\odot \e'}\}_{\e'\le \e})\right\|_\calF\right]
	\end{align*}
	
Let $\mathbb P_{I_{\bN,\e}}=|I_{\bN,\e}|^{-1}\sum_{\i \in I_{\bN,\e}}\delta_{\{U_{\i \odot \e'}\}_{\e'\le \e}}$ denote the empirical measure on the support of $\{U_{\i \odot \e'}\}_{\e'\le \e}$.
Conditionally on $\{X_\i\}_{\i\in\bN}$, the object
\begin{align*}
	R_\bN^\e(f)=\frac{1}{\sqrt{|I_{\bN,\e}|}}\sum_{\i\in I_{\bN,\e}}\varepsilon_{1,i_1}...\varepsilon_{k,i_k}(P_{\e}f )(\{U_{\i\odot \e'}\}_{\e'\le \e})
\end{align*}
is a homogeneous Rademacher chaos process of order $k$. Moreover, by Corollary 3.2.6 in \cite{delaPenaGine1999}, for any $f,f'\in \calF$,
\begin{align*}
\left\|R_\bN^\e(f)-R_\bN^\e(f')\right\|_{\psi_{2/k}|\{X_\i\}_{\i\in \bN}}\lesssim \left\|R_\bN^\e(f)-R_\bN^\e(f')\right\|_{\mathbb P_{I_{\bN,\e}},2}.
\end{align*}
Hence, up to a constant, the diameter of $\calF$ in the $\|\cdot\|_{\psi_{2/k}|\{X_\i\}_{\i\in \bN}}$-norm is bounded by  $ {\sigma_{I_{\bN,\e}}}$, where
 $\sigma_{I_{\bN,\e}}^2:=\sup_{f\in\calF}\| P_{\e}f \|_{\mathbb P_{I_{\bN,\e}},2}^2
$. 
By
 Fubini's theorem, Corollary 5.1.8 in \cite{delaPenaGine1999}, and a change of variables,
	\begin{align*}
		&\mathbb{E}\left[\left\|\frac{1}{\sqrt{|I_{\bN,\e}|}}\sum_{\i\in I_{\bN,\e}}\varepsilon_{1,i_1}...\varepsilon_{k,i_k}(P_{\e}f )(\{U_{\i\odot \e'}\}_{\e'\le \e})\right\|_\calF\right]\\
			\lesssim&
			\mathbb{E}\left[\left\|\left\|\frac{1}{\sqrt{|I_{\bN,\e}|}}\sum_{\i\in I_{\bN,\e}}\varepsilon_{1,i_1}...\varepsilon_{k,i_k}(P_{\e}f )(\{U_{\i\odot \e'}\}_{\e'\le \e})\right\|_\calF\right\|_{\psi_{2/k}|\{X_\i\}_{\i\in\bN}}\right]\\
		\lesssim&
		\mathbb{E} \left [ \int_0^{\sigma_{I_{\bN,\e}}}\left[1+\log N\left(P_{ \e}\calF,\|\cdot\|_{\mathbb P_{I_{\bN,\e}},2},\tau\right)\right]^{k/2}d \tau \right ] \\
		=&\mathbb{E} \left [
	\|P_\e F\|_{\mathbb P_{I_{\bN,\e}},2}
			 \int_0^{\sigma_{I_{\bN,\e}}/\|P_\e F\|_{\mathbb P_{I_{\bN,\e}},2}}\left[1+\log N\left(P_{ \e}\calF,\|\cdot\|_{\mathbb P_{I_{\bN,\e}},2},\tau\|P_\e F\|_{\mathbb P_{I_{\bN,\e}},2}\right)\right]^{k/2}d \tau \right ]\\
		\le&
		\mathbb{E} \left [\|P_\e F\|_{\mathbb P_{I_{\bN,\e}},2} J_\e\left(\sigma_{I_{\bN,\e}}/\|P_\e F\|_{\mathbb P_{I_{\bN,\e}},2}\right) \right ].
	\end{align*}
	By Lemma~\ref{lemma:uniform_entropy_integral} and Jensen's inequality,
	\begin{align}
		|I_{\bN,\e}|^{1/2}\mathbb{E}[\|H_{\bN}^\e (f)\|_\calF]\lesssim&
		\|P_\e F\|_{P,2} J_\e\left(z\right), \label{eq:local_maximal_ineq_prelimary_bound}
	\end{align}
	where $z:=\sqrt{\mathbb{E}[\sigma_{I_{\bN,\e}}^2]/\|P_\e F\|_{P,2}^2}$.

	We now bound 
	\begin{align*}
		\mathbb{E}[\sigma_{I_{\bN,\e}}^2]=&\mathbb{E}\left[\left\|\frac{1}{|I_{\bN,\e}|}\sum_{\i \in I_{\bN,\e}}(P_\e f )^2\left(\{U_{\i \odot \e'}\}_{\e'\le \e}\right)\right\|_\calF\right].
	\end{align*}
To handle the squared summands, we wish to apply the Hoffmann--J\o rgensen inequality, but the summands are not independent. We therefore use Lemma~\ref{lemma:partition} to obtain a partition \(\mathcal{G}\) of \(I_{\bN,\e}\) into \(|\calG|=|I_{\bN,\e}|/m_\e\) groups, each containing \(m_\e\) i.i.d.\ observations; the i.i.d.\ property follows from the AHK representation \eqref{eq:AHK_representation}, since within each group any two observations share no common indices \(i_1,\dots,i_K\).

For each group \(G = \{\i_{1}(G), \i_{2}(G), \dots, \i_{m_\e}(G)\} \in \mathcal{G}\), define
\[
D_{f,\e}(G) = \frac{1}{m_\e}\sum_{t=1}^{m_\e} \Bigl(P_{\e}f\Bigr)^2\!\left(\{U_{\i_t(G)\odot \e'}\}_{\e'\le \e}\right),
\]
and let
\[
D_{f,\e} = \frac{1}{m_\e}\sum_{t=1}^{m_\e} \Bigl(P_{\e}f\Bigr)^2\!\left(\{U_{(t,\dots,t)\odot \e'}\}_{\e'\le \e}\right).
\]
		Then
		\begin{align*}
			\frac{1}{|I_{\bN,\e}|}\sum_{\i \in I_{\bN,\e}}(P_{\e}f )^2\left(\{U_{\i \odot \e'}\}_{\e'\le \e}\right)=\frac{1}{|I_{\bN,\e}|/m_\e}\sum_{G\in \calG} D_{f,\e}(G).
		\end{align*}
For each \(G\in\calG\), the AHK representation \eqref{eq:AHK_representation} implies that \(D_{f,\e}\) and \(D_{f,\e}(G)\) are identically distributed, so, by Jensen's inequality,
\begin{align*}
	\mathbb{E}\Bigl[\sigma_{I_{\bN,\e}}^2\Bigr]
	&= \mathbb{E}\!\left[\left\|\frac{1}{|I_{\bN,\e}|/m_\e}\sum_{G\in\calG} D_{f,\e}(G)\right\|_\calF\right]
	\le \mathbb{E}\!\left[\left\|D_{f,\e}\right\|_\calF\right].
\end{align*}
Denote this bound by
\[
B_{m_\e,\e} :=\mathbb{E}\!\left[\left\|D_{f,\e}\right\|_\calF\right]= \mathbb{E}\!\left[\left\|\frac{1}{m_\e}\sum_{t=1}^{m_\e} \Bigl(P_{\e}f\Bigr)^2\!\Bigl(\{U_{(t,\dots,t)\odot \e'}\}_{\e'\le \e}\Bigr)\right\|_\calF\right].
\]
	Thus $z\le\tilde z:=\sqrt{B_{m_\e,\e}}/\|P_\e F\|_{P,2}$.
	By the symmetrisation inequality for independent processes, the contraction principle (Theorem 4.12 in \citealt{ledoux1991probability}), and the Cauchy--Schwarz inequality,
	\begin{align*}
		B_{m_\e,\e}=&\mathbb{E}\left[\left\|\frac{1}{m_\e}\sum_{t=1}^{m_\e} (P_{\e}f )^2\left(\{U_{(t,...,t)\odot \e'}\}_{\e'\le \e}\right)\right\|_{\calF}\right]\\
		\le&
		\sigma_\e^2 +\mathbb{E}\left[\left\|\frac{1}{m_\e}\sum_{t=1}^{m_\e}\left\{ (P_{\e}f )^2\left(\{U_{(t,...,t)\odot \e'}\}_{\e'\le \e}\right)-\mathbb{E}\left[(P_{\e}f )^2\right]\right\}\right\|_{\calF}\right]\\
		\lesssim&
		\sigma_\e^2 +\mathbb{E}\left[\left\|\frac{1}{m_\e}\sum_{t=1}^{m_\e}\varepsilon_t\cdot (P_{\e}f )^2\left(\{U_{(t,...,t)\odot \e'}\}_{\e'\le \e}\right)\right\|_{\calF}\right]\\
		\lesssim&
		\sigma_\e^2 +\mathbb{E}\left[M_\e^{(m_\e)}\left\|\frac{1}{m_\e}\sum_{t=1}^{m_\e}\varepsilon_t\cdot (P_{\e}f )\left(\{U_{(t,...,t)\odot \e'}\}_{\e'\le \e}\right)\right\|_{\calF}\right]\\
		\le&
		\sigma_\e^2 +\|M_\e^{(m_\e)}\|_{P,2}\sqrt{\mathbb{E}\left[\left\|\frac{1}{m_\e}\sum_{t=1}^{m_\e}\varepsilon_t\cdot (P_{\e}f )\left(\{U_{(t,...,t)\odot \e'}\}_{\e'\le \e}\right)\right\|_{\calF}^2\right]}.
	\end{align*}
	The Hoffmann--J\o rgensen inequality (Proposition A.1.6 in \citealt{vdVW1996}) gives
	\begin{align*}
		&\sqrt{\mathbb{E}\left[\left\|\frac{1}{m_\e}\sum_{t=1}^{m_\e}\varepsilon_t\cdot (P_{\e}f )\left(\{U_{(t,...,t)\odot \e'}\}_{\e'\le \e}\right)\right\|_{\calF}^2\right]}\\
		\lesssim&
		\mathbb{E}\left[\left\|\frac{1}{m_\e}\sum_{t=1}^{m_\e}\varepsilon_t\cdot (P_{\e}f )\left(\{U_{(t,...,t)\odot \e'}\}_{\e'\le \e}\right)\right\|_{\calF}\right]+\frac{1}{m_\e}\|M_\e^{(m_\e)}\|_{P,2}.
	\end{align*}
Arguing as in the first part of the proof, we deduce that
	\begin{align*}
		&\mathbb{E}\left[\left\|\frac{1}{\sqrt{m_\e}}\sum_{t=1}^{m_\e}\varepsilon_t\cdot (P_{\e}f )\left(\{U_{(t,...,t)\odot \e'}\}_{\e'\le \e}\right)\right\|_{\calF}\right]\\
		\lesssim& \|P_\e F\|_{P,2}\int_{0}^{\tilde z} \sup_{Q}\sqrt{1+\log N(P_\e\calF,\|\cdot\|_{Q,2},\epsilon \|P_\e F\|_{Q,2})}d\epsilon.
	\end{align*}
	The integral on the RHS is bounded by $J_\e(\tilde z)$, so
	\begin{align*}
		B_{m_\e,\e}\lesssim& \sigma_\e^2 + m_\e^{-1}\|M_\e^{(m_\e)}\|_{P,2}^2 + m_\e^{-1/2} \|M_\e^{(m_\e)}\|_{P,2} \|P_\e F\|_{P,2}J_\e(\tilde z).
	\end{align*}

	Setting $$\Delta=(\sigma_\e\vee m_\e^{-1/2}\|M_\e^{(m_\e)}\|_{P,2})/\|P_\e F\|_{P,2},$$  we obtain
	\begin{align*}
		\tilde z^2\lesssim \Delta^2 + \frac{\|M_\e^{(m_\e)}\|_{P,2}}{\sqrt{m_\e}\|P_\e F\|_{P,2}}J_\e(\tilde z).
	\end{align*}
	Lemma~\ref{lemma:uniform_entropy_integral} and Lemma 2.1 of \cite{vanderVaartWellner2011} with $J=J_\e$, $A=\Delta$, $B=\sqrt{\|M_\e^{(m_\e)}\|_{P,2}/\sqrt{m_\e}\|P_\e F\|_{P,2}}$ and $r=1$ yield
	\begin{align*}
		J_\e(z)\le J_\e(\tilde z)\lesssim J_\e(\Delta)\left\{1+J_\e(\Delta)\frac{\|M_\e^{(m_\e)}\|_{P,2}}{\sqrt{m_\e}\|P_\e F\|_{P,2} \Delta^2}\right\}.
	\end{align*}
	Combining this with (\ref{eq:local_maximal_ineq_prelimary_bound}), we obtain
	\begin{align}
		|I_{\bN,\e}|^{1/2}\mathbb{E}[\|H_{\bN}^\e (f)\|_\calF]\lesssim&
		J_\e(\Delta)\|P_{\e} F\|_{P,2} +\frac{J_\e^2(\Delta)\|M_\e^{(m_\e)}\|_{P,2}}{\sqrt{m_\e}\Delta^2}.\label{eq:local_maximal_ineq_main_bound}
	\end{align}
	By definition, $\delta_\e\le \Delta$, so Lemma~\ref{lemma:uniform_entropy_integral}(iii) gives
	\begin{align*}
		J_\e(\Delta)\le \Delta\frac{J_\e(\delta_\e)}{\delta_\e}=\max\left\{J_\e(\delta_\e), \frac{\|M_\e^{(m_\e)}\|_{P,2} J_\e(\delta_\e)}{\sqrt{m_\e}\|P_\e F\|_{P,2}\delta_\e}\right\}\le\max\left\{J_\e(\delta_\e), \frac{\|M_\e^{(m_\e)}\|_{P,2} J_\e^2(\delta_\e)}{\sqrt{m_\e}\|P_\e F\|_{P,2}\delta_\e^2}\right\},
	\end{align*}
	where the second inequality follows from $J_\e(\delta_\e)/\delta_\e\ge J_\e(1)\ge 1$. Finally, by Lemma~\ref{lemma:uniform_entropy_integral}(iii),
	\begin{align*}
		\frac{J_\e^2(\Delta)\|M_\e^{(m_\e)}\|_{P,2}}{\sqrt{m_\e}\Delta^2}\le \frac{J_\e^2(\delta_\e)\|M_\e^{(m_\e)}\|_{P,2}}{\sqrt{m_\e}\delta_\e^2}
	\end{align*}
	Combining these bounds with (\ref{eq:local_maximal_ineq_main_bound}) yields the bound of the theorem with $\left(\sqrt{m_\e},\, M_\e^{(m_\e)}\right)$ in place of $\left(\sqrt{n},\, M_\e\right)$; as explained at the beginning of the proof, \eqref{eq:blockmax2} with $m=m_\e$ then gives the desired inequality.
	\qed
    
	\vskip 0.15in
	
	\subsubsection*{Proof of Lemma~\ref{lemma:partition}}
	Write $k=\Vert\e\Vert_0$. The tuples in $I_{\bN,\e}=\{\i\odot\e:\i\in[\bN]\}$ vanish outside $\supp(\e)$, and transversality constrains only the coordinates in $\supp(\e)$, so we may identify $I_{\bN,\e}$ with $\prod_{k'\in\supp(\e)}[N_{k'}]$. Order $\supp(\e)$ so that the associated dimensions are non-increasing and denote them by $m_1\ge m_2\ge\cdots\ge m_k$, so that $m_k=\min_{k'\in\supp(\e)}N_{k'}$. It therefore suffices to partition $[m_1]\times\cdots\times[m_k]$ into transversal subsets of size $m_k$.

	If $k=1$, then $[m_1]$ is itself transversal and there is nothing to prove, so assume $k\ge2$. For $j\in[k-1]$ let
	\[
	\phi_j(t,g)=\bigl((t+g-2)\bmod m_j\bigr)+1,\qquad (t,g)\in[m_k]\times[m_j],
	\]
	denote the representative in $[m_j]$ of $t+g-1$ modulo $m_j$, and define
	\[
	\Phi(\bm g,t)=\bigl(\phi_1(t,g_1),\dots,\phi_{k-1}(t,g_{k-1}),\,t\bigr),\qquad
	\bm g=(g_1,\dots,g_{k-1})\in[m_1]\times\cdots\times[m_{k-1}].
	\]
	The candidate groups are the images $G_{\bm g}=\Phi\bigl(\{\bm g\}\times[m_k]\bigr)$ of the fibres of the first argument.

	\emph{The collection $\calG=\{G_{\bm g}\}_{\bm g}$ partitions $[m_1]\times\cdots\times[m_k]$ into sets of size $m_k$.} Since the fibres $\{\bm g\}\times[m_k]$ partition the domain of $\Phi$, it suffices to show that $\Phi$ is a bijection onto $[m_1]\times\cdots\times[m_k]$. Fix $\bm x=(x_1,\dots,x_k)$ in the codomain. Any preimage must have $t=x_k$, and for each $j\in[k-1]$ the map $g\mapsto\phi_j(x_k,g)$ is a translation on the cyclic group $\Z/m_j\Z$ and hence a bijection of $[m_j]$, so the equation $\phi_j(x_k,g_j)=x_j$ determines $g_j$ uniquely. Thus $\bm x$ has exactly one preimage.

	\emph{Every $G_{\bm g}$ is transversal.} Fix $\bm g$ and distinct $t,t'\in[m_k]$. The last coordinates $t$ and $t'$ already differ. For $j\in[k-1]$, the equality $\phi_j(t,g_j)=\phi_j(t',g_j)$ would force $t\equiv t'\ (\mathrm{mod}\ m_j)$, which is impossible because $0<|t-t'|<m_k\le m_j$. Hence the two tuples differ in every coordinate indexed by $\supp(\e)$.

	The construction thus yields $m_1\cdots m_{k-1}=|I_{\bN,\e}|/m_k$ transversal groups of size $m_k=\min_{k'\in\supp(\e)}N_{k'}=m_\e$, as claimed. Moreover, $m_\e$ divides $|I_{\bN,\e}|$ exactly, and $m_\e=n$ whenever $\supp(\e)$ contains a shortest dimension, in particular when the array is balanced.
	\qed

\vskip 0.15in
        
\subsubsection{Proof of Corollary~\ref{cor:local_maximal_ineq_VC}}\label{sec:proof of cor:local_maximal_ineq_VC}

The proof follows that of Corollary 5.3 in \cite{ChenKato2019b}, using Lemma~\ref{lemma:conditional_exp}, and is omitted.  \qed

\vskip 0.15in

\subsection{Auxiliary Lemmas}
The following restates Theorem 5.2 in \cite{CCK2014AoS}, which modifies Theorem 2.1 in \cite{vanderVaartWellner2011} to allow for an unbounded envelope. 
\begin{lemma}[Local maximal inequality for i.i.d.\ data]\label{lemma:local_max_ineq_indep}
Let $X_1,...,X_n$ be $\calS$-valued i.i.d.\ random variables. Suppose $0<\|F\|_{P,2}<\infty$ and let $\sigma^2$ be any positive constant such that $\sup_{f\in \calF} P f^2\le \sigma^2 \le \|F\|_{P,2}^2$. Set $\delta=\sigma/\|F\|_{P,2} $ and $B=\sqrt{\mathbb{E}[\max_{i\in [n]} F^2(X_i)]}$. Then
\begin{align*}
\mathbb{E}[\|\Gn f\|_\calF] 
\lesssim \|F\|_{P,2} J(\delta,\calF,F) + \frac{B J^2(\delta,\calF,F)}{\delta^2 \sqrt{n}}.
\end{align*}
If, in addition, $\calF$ is VC-type with characteristics $(A,v)$,
then
\begin{align*}
\mathbb{E}[\|\Gn f\|_\calF] 
\lesssim \sigma\sqrt{v  \log\left(\frac{A\|F\|_{P,2}}{\sigma}\right)} + \frac{vB }{\sqrt{n}} \log\left(\frac{A\|F\|_{P,2}}{\sigma}\right).
\end{align*}
\end{lemma} 
The following restates Lemma B.3 in \cite{chiang2023inference}.
\begin{lemma}[Bounding $L^r$-norm by Orlicz norm]
	\label{lemma: Orlicz norms}
	Let $0 < \beta < \infty$ and $1 \le q < \infty$, and let $m = m(\beta,q)$ be the smallest positive integer satisfying $m\beta \ge q$. Then, for every real-valued random variable $\xi$, $(\mathbb{E}[|\xi|^{{q}}])^{1/{q}} \le (m!)^{1/( m\beta)} \| \xi \|_{\psi_\beta}$. 
\end{lemma}

The following is analogous to Lemma 5.2 in \cite{ChenKato2019b} and Lemma A.2 in \cite{CCK2014AoS}.
\begin{lemma}[Properties of $J_\e(\delta)$]\label{lemma:uniform_entropy_integral}
Suppose that $J_\e(1)<\infty$ for $\e\in\{0,1\}^K$. Then, for all $\e\in\{0,1\}^K$,
\begin{enumerate}[(i)]
\item $\delta \mapsto J_\e(\delta)$ is non-decreasing and concave.
\item For $c\ge1$, $J_\e(c\delta)\le c J_\e(\delta)$.
\item $\delta\mapsto J_\e(\delta)/\delta$ is non-increasing. 
\item  $(x,y)\mapsto J_\e(\sqrt{x/y})\sqrt{y}$ is jointly concave in $(x,y)\in [0,\infty)\times (0,\infty)$.
\end{enumerate}
\end{lemma}

The following restates Lemma A.2 in \cite{ghosal2000testing}, which is useful for calculating the uniform entropy integrals of the classes after the Hoeffding-type decomposition.
	\begin{lemma}[Uniform covering for conditional expectations]\label{lemma:conditional_exp}
	Let $\mathcal{F}$ be a class of functions $f:\mathcal{X}\times \calY\to \R$ with envelope $F$, and let $R$ be a fixed probability measure on $\calY$. For $f\in \mathcal{F}$, let $\overline f:\mathcal{X}\to \R$ be $\overline f=\int f(x,y)dR(y)$. Set $\overline{\mathcal{F}}=\{\overline f:f\in\mathcal{F}\}$. Note that $\overline{F}$ is an envelope of $\overline{\mathcal{F}}$.
	Then, for any $\varepsilon\in(0,1]$,
	\begin{align*}
\sup_{Q}N(\overline{\mathcal{F}},\|\cdot\|_{Q,2},2\varepsilon\|\overline F\|_{Q,2})\le 	\sup_{Q'}N(\mathcal{F},\|\cdot\|_{Q',2},\varepsilon^2\|F\|_{Q',2}),
	\end{align*}
	where $\sup_Q$ and $\sup_{Q'}$ are taken over all finite discrete distributions on $\mathcal X$ and $\mathcal X\times \mathcal Y$, respectively.

\end{lemma}


The following lemma extends Lemma 4.3 of \cite{NeweyMc1994} to triangular arrays under multiway clustering.

\begin{lemma}[LLN for triangular arrays with plug-in estimates under multiway clustering]
\label{lemma:4_3}
Let \(\mathcal P_n\) be a sequence of collections of data-generating processes for the $K$-array $\{X_{n,\bm i}:\bm i\in[\bm N_n]\}$. Suppose that, uniformly over \(P_n\in \mathcal P_n\), (i) the array satisfies Conditions (SE) and (D); (ii) there is a neighbourhood \(\mathcal N\) of \(\eta_0\) and some \(\delta>0\) such that $\mathbb E_{P_n}[\sup_{\eta\in \mathcal{N}} \Vert a_n(X_{n,\bm i},\eta)\Vert^{1+\delta} ]< C<\infty$; (iii) for every sequence of constants $\delta_n\downarrow0$,
\[
\sup_{P_n\in\mathcal P_n}\mathbb E_{P_n}\left[\Delta_n(X_{n,\bm 1};\delta_n)\right]\to0,\qquad \Delta_n(x;\delta):=\sup_{\eta:\,\|\eta-\eta_{0}\|_{P_n,2}\le\delta}\Vert a_n(x,\eta)-a_n(x,\eta_{0})\Vert,
\]
where the supremum is $P_n$-measurable under the pointwise-measurability convention maintained throughout; and (iv) $\|\widehat\eta-\eta_0\|_{P_n,2} \overset{p}{\to} 0$. Then, uniformly over $P_n\in \mathcal P_n$, $\mathbb E_N [a_n(X_{n,\bm i},\widehat\eta)] \overset{p}{\to} \mathbb{E}_{P_n}[a_n(X_{n,\bm 1},\eta_0)]$.
\end{lemma}

\begin{remark}
Condition (iii) is a uniform $L^1$-equicontinuity requirement at $\eta_0$. When $a_n=a$, $\eta_{0,n}=\eta_0$ and $P_n=P$ for all $n$, it follows from (ii) and the almost-sure continuity of $\eta\mapsto a(x,\eta)$ at $\eta_0$: by sequential continuity in the metric space $(\Gamma,\|\cdot\|_{P,2})$, $\Delta_n(x;\delta_n)\to0$ for $P$-almost every $x$, and $\Delta_n(x;\delta_n)\le 2\sup_{\eta\in\mathcal N}\|a(x,\eta)\|$ for $n$ large enough, so dominated convergence applies. For genuinely drifting triangular arrays, continuity of each $a_n$ alone does not deliver (iii), as a modulus of continuity uniform in $n$ is needed; we therefore impose (iii) directly. In our applications, it is guaranteed by Assumption~\ref{assu:gmm_reg}(viii) and the corresponding conditions in Theorem~\ref{thm:var}.
\hfill $\lozenge$\end{remark}

\begin{proof}[Proof of Lemma~\ref{lemma:4_3}]
The proof follows that of Lemma 4.3 of \cite{NeweyMc1994}, with the continuity step replaced by the uniform $L^1$-equicontinuity condition (iii) and Khintchine's LLN replaced by an LLN for triangular arrays under multiway clustering. First, we show
$$
    \mathbb E_N a_n(X_{n,\bm i},\eta_0)- \mathbb E_{P_n}\left[a_n(X_{n,\bm 1},\eta_0)\right] \overset{p}{\to}0 .
$$
Without loss of generality assume $\delta\le1$, and argue componentwise. By the Hoeffding-type decomposition \eqref{hoeffding} applied to the centred function under Conditions (SE) and (D), it suffices to show $H^{\bm e}_{\bm N}(a_n)\overset{p}{\to}0$ for each $\bm e\in\{0,1\}^K\setminus\{\bm 0\}$, writing $a_n$ for a generic centred component. For $\bm e\in\calE_1$ the summands $\pi_{\bm e}a_n$ are i.i.d.\ with mean zero, so the von Bahr--Esseen inequality gives $\mathbb E|H^{\bm e}_{\bm N}(a_n)|^{1+\delta}\le 2|I_{\bm N,\bm e}|^{-\delta}\,\mathbb E|\pi_{\bm e}a_n|^{1+\delta}$. For $\Vert\bm e\Vert_0=k\ge2$ the summands are dependent but completely degenerate: for any $\ell\in\supp(\bm e)$, $\pi_{\bm e}a_n$ has mean zero conditionally on $\{U_{\bm i\odot\bm e'}\}_{\bm e'\le\bm e-\bm e_\ell}$ (see Section~\ref{sec:maximal_inequalities_SE}). When the summation over $I_{\bm N,\bm e}$ is grouped by its coordinate in one dimension of $\supp(\bm e)$ at a time, the group sums are, conditionally on the AHK variables whose label is zero in that dimension, independent with zero conditional means, so the von Bahr--Esseen inequality applies conditionally; iterating over the $k$ dimensions of $\supp(\bm e)$ yields
$$
\mathbb E\left|H^{\bm e}_{\bm N}(a_n)\right|^{1+\delta}\le 2^k|I_{\bm N,\bm e}|^{-\delta}\,\mathbb E|\pi_{\bm e}a_n|^{1+\delta}\lesssim |I_{\bm N,\bm e}|^{-\delta}\,\mathbb E\Vert a_n(X_{n,\bm 1},\eta_0)\Vert^{1+\delta}\to0,
$$
where the second inequality uses the conditional Jensen inequality on each conditional expectation $P_{\bm e'}a_n$ entering $\pi_{\bm e}a_n$ and the triangle inequality, and the convergence uses $|I_{\bm N,\bm e}|\ge n\to\infty$ and the moment condition $\mathbb E\Vert a_n(X_{n,\bm i},\eta_0)\Vert^{1+\delta} \leq \mathbb E[\sup_{\eta\in \mathcal{N}} \Vert a_n(X_{n,\bm i},\eta)\Vert^{1+\delta}] <C< \infty$. All bounds depend on $P_n$ only through $C$ and $\delta$, so the convergence holds uniformly over $P_n\in\mathcal P_n$.

Next, since $\|\widehat\eta-\eta_0\|_{P_n,2} \overset{p}{\to} 0$ uniformly over $P_n\in\mathcal P_n$, there exists a sequence \(\delta_n\downarrow0\) such that $\sup_{P_n\in\mathcal P_n}P_n(\|\widehat\eta-\eta_0\|_{P_n,2}>\delta_n)\to 0$. On the event $\{\|\widehat\eta-\eta_0\|_{P_n,2}\le\delta_n\}$,
\begin{align*}
    \left\Vert \mathbb E_N a_n(X_{n,\bm i},\widehat\eta) - \mathbb E_N a_n(X_{n,\bm i},\eta_0) \right\Vert \leq \mathbb E_N\Delta_n(X_{n,\bm i};\delta_n),
\end{align*}
and, for every $\varepsilon>0$, Markov's inequality and the separate exchangeability of the array give
$$ \sup_{P_n\in\mathcal P_n} P_n\left( \mathbb E_N\Delta_n(X_{n,\bm i};\delta_n)>\varepsilon\right) \leq
    \frac{\sup_{P_n\in\mathcal P_n}\mathbb E_{P_n}\left[\Delta_n(X_{n,\bm 1};\delta_n)\right]}{\varepsilon}\to 0$$
by condition (iii). The conclusion follows from the triangle inequality.
\end{proof}

The following second-moment law of large numbers for triangular arrays under multiway clustering extends the argument of Proposition 4.1 in \cite{DDG2018}, established there for a fixed data generating process, to drifting sequences.

\begin{lemma}[Second-moment LLN under multiway clustering]\label{lemma:second_moment_lln}
Let $\{X_{n,\bm i}:\bm i\in[\bm N_n]\}$ satisfy Conditions (SE) and (D) for each $n$, and let $h_n:\calS\to\R^{q}$ be measurable functions with $\sup_n\sup_{P_n\in\mathcal P_n}\mathbb E_{P_n}\Vert h_n(X_{n,\bm 1})\Vert^{2+\delta}\le C<\infty$ for some $\delta>0$. Then, for each $k\in[K]$, uniformly over $P_n\in\mathcal P_n$,
$$
\frac{N_k}{N^2}\sum_{\bm i,\bm j\in[\bm N]}h_n(X_{\bm i})h_n(X_{\bm j})'\1\{i_k=j_k\} - \mathbb E\left[h_n(X_{\bm 1})h_n(X_{\bm 2-\bm e_k})'\right] \overset{p}{\to} \bm 0.
$$
\end{lemma}

\begin{proof}
Fix $k$; without loss of generality $k=1$. Write $h=h_n$, $m=N/N_1$, and, for $c\in[N_1]$, let $T_c = m^{-1}\sum_{\bm i\in[\bm N]:\,i_1=c}h(X_{\bm i})$, so that the first term of the display equals $N_1^{-1}\sum_{c=1}^{N_1}T_cT_c'$. By the moment condition and the Cauchy--Schwarz inequality, $\mathbb E\Vert h(X_{\bm 1})\Vert^2\le C^{2/(2+\delta)}=:C'$ and every $\Vert\mathbb E[h(X_{\bm i})h(X_{\bm j})']\Vert\le C'$.

\emph{Step 1 (mean).} By (SE), $\mathbb E[T_cT_c'] = m^{-2}\sum_{\bm i,\bm j:\,i_1=j_1=c}\mathbb E[h(X_{\bm i})h(X_{\bm j})']$ does not depend on $c$. Among the $m^2$ index pairs, those with $\bm i=\bm j$ number $m$, and those sharing a coordinate in some dimension $\ell\ge2$ number at most $m^2\sum_{\ell\ge2}N_\ell^{-1}$; for every remaining pair, $\mathbb E[h(X_{\bm i})h(X_{\bm j})'] = \mathbb E[h(X_{\bm 1})h(X_{\bm 2-\bm e_1})']$ by (SE). Hence
$$\left\Vert \mathbb E[T_cT_c'] - \mathbb E[h(X_{\bm 1})h(X_{\bm 2-\bm e_1})']\right\Vert \le 2C'\left(m^{-1}+\sum_{\ell\ge2}N_\ell^{-1}\right)\le \frac{2C'K}{n}\to0.$$

\emph{Step 2 (truncation).} For $\tau>0$, let $h^\tau = h\,\1\{\Vert h\Vert\le\tau\}$ and define $T^\tau_c$ accordingly. By Minkowski's inequality and (SE), $(\mathbb E\Vert T_c-T^\tau_c\Vert^2)^{1/2}\le (\mathbb E[\Vert h(X_{\bm 1})\Vert^2\1\{\Vert h(X_{\bm 1})\Vert>\tau\}])^{1/2}\le C^{1/2}\tau^{-\delta/2}$ and $(\mathbb E\Vert T_c\Vert^2)^{1/2}\vee(\mathbb E\Vert T^\tau_c\Vert^2)^{1/2}\le (C')^{1/2}$, uniformly over $n$, $P_n$, and $c$. Hence, writing $T_cT_c' - T^\tau_c T^{\tau\prime}_c = (T_c-T^\tau_c)T_c' + T^\tau_c(T_c-T^\tau_c)'$ and using the Cauchy--Schwarz inequality,
$$
\mathbb E\left\Vert \frac{1}{N_1}\sum_{c=1}^{N_1} \left(T_cT_c' - T^\tau_c T^{\tau\prime}_c\right)\right\Vert \le 2(C'C)^{1/2}\,\tau^{-\delta/2},
$$
and the same bound applies to $\Vert\mathbb E[T_1T_1']-\mathbb E[T^\tau_1T^{\tau\prime}_1]\Vert$.

\emph{Step 3 (variance of the truncated average).} The entries of $T^\tau_c T^{\tau\prime}_c$ are bounded by $\tau^2$. For $c\ne c'$ and entries $(a,b)$, expand
$${\rm Cov}\left((T^\tau_cT^{\tau\prime}_c)_{ab},\,(T^\tau_{c'}T^{\tau\prime}_{c'})_{ab}\right) = m^{-4}\sum_{\substack{\bm i,\bm i':\,i_1=i_1'=c\\ \bm j,\bm j':\,j_1=j_1'=c'}}{\rm Cov}\left(h^\tau_a(X_{\bm i})h^\tau_b(X_{\bm i'}),\,h^\tau_a(X_{\bm j})h^\tau_b(X_{\bm j'})\right).$$
Since $c\ne c'$, each covariance on the right vanishes by Condition (D) unless $\{\bm i,\bm i'\}$ and $\{\bm j,\bm j'\}$ share a coordinate in some dimension $\ell\ge2$; the number of such quadruples is at most $4m^4\sum_{\ell\ge2}N_\ell^{-1}$, and each covariance is bounded by $2\tau^4$. The diagonal terms $c=c'$ contribute at most $N_1^{-1}\tau^4$ to the variance of the average. Therefore
$${\rm Var}\left(\frac{1}{N_1}\sum_{c=1}^{N_1} (T^\tau_cT^{\tau\prime}_c)_{ab}\right)\le \tau^4\left(\frac{1}{N_1}+8\sum_{\ell\ge2}\frac{1}{N_\ell}\right)\le \frac{9\tau^4 K}{n}\to0,$$
and Chebyshev's inequality yields $N_1^{-1}\sum_c T^\tau_cT^{\tau\prime}_c - \mathbb E[T^\tau_1T^{\tau\prime}_1]\overset{p}{\to}0$ entrywise, for each fixed $\tau$.

\emph{Conclusion.} Given $\varepsilon>0$, choose $\tau$ so that the Step-2 bound satisfies $2(C'C)^{1/2}\tau^{-\delta/2}<\varepsilon$, and combine Steps 1--3 via the triangle and Markov inequalities. All bounds depend on $P_n$ only through $C$, $\delta$, $q$, and $K$, so the convergence holds uniformly over $P_n\in\mathcal P_n$.
\end{proof}

The following lemma restates Lemma 3 in \cite{chiang2022multiway}; to avoid extra notation, we consider the special case in which each cell $\bm i \in [\bm N]$ contains exactly one observation. Let $I_1 = \bigcup_{\bm e \in \mathcal{E}_1} I_{\bm{N,e}}$, and let $k(\bm i)$ denote the coordinate in which $\bm{i} \in I_1$ is non-zero.
\begin{lemma}
\label{lemma:hajek}
    Suppose, for each $n\in \mathbb{N}$, $(X_{\bm i})_{\bm i\in [\bm N]}$ satisfies Conditions (SE) and (D). Let $\mathcal{F}_n$, $|\mathcal{F}_n| = {m}$, be a family of functions $f: {\calS}\to \mathbb{R}^{{q}}$ such that $\mathbb{E}\left\Vert f(X_{\bm i}) \right\Vert^2<C<\infty$ for some $C$ independent of $n$. Let $\mu_k$ be such that $\frac{n}{N_k}\to \mu_k\geq 0$. Then there exists a family of mutually independent standard uniform r.v.s $(U_{\bm i })_{\bm i\in I_1}$ such that the H\'ajek projection of $\Gn(f)$ on the set of statistics of the form $\sum_{\bm i \in I_1} g_{\bm i}(U_{\bm i})$, with $ g_{\bm i}(U_{\bm i})$  square-integrable, satisfies
    \begin{align*}
        H_nf = \sum_{\bm i \in I_1} \frac{\sqrt{n}}{N_{k(\bm i)}} \left( \mathbb{E}\left[f(X_{\bm i}) | U_{\bm i}\right] - \mathbb{E}\left[f(X_{\bm i}) \right] \right),
    \end{align*}
     where, for $\bm i\in I_1$, $\mathbb{E}[f(X_{\bm i})|U_{\bm i}]$ is shorthand for $\mathbb{E}[f(X_{\bm j})|U_{\bm i}]$ with $\bm j\in[\bm N]$ any index whose $k(\bm i)$-th coordinate agrees with that of $\bm i$; by (SE) this does not depend on the choice of $\bm j$.
    Moreover, uniformly over $\mathcal{F}_n$,
    \begin{align*}
        \Gn(f) & =  H_nf + O_{P_n}(n^{-1/2}), \\
{\rm Var}\left(\Gn(f)\right) & =  {\rm Var}\left({H}_n(f)\right) + O(n^{-1}) = \sum_{\bm e\in\mathcal{E}_1} {\frac{n}{N_{k(\bm e)}}} {\rm Cov} (f(X_{\bm 1}),f(X_{\bm{2-e}})) + O(n^{-1}) \\
&= \mu_1 {\rm Var}( \mathbb{E}[f(X)|U_{1,0,...,0}]) + ... +\mu_K {\rm Var}( \mathbb{E}[f(X)|U_{0,0,...,1}])+  {o(1)},
    \end{align*}
     where the last equality uses $n/N_k=\mu_k+o(1)$ and the boundedness of the covariances.
\end{lemma}

\begin{proof}[Proof of Lemma~\ref{lemma:hajek}]
    All statements in Lemma~\ref{lemma:hajek} except the last equality are established in the proof of Lemma 3 in \cite{chiang2022multiway}, so we prove only the last equality. Fix any $k\in[K]$. Under Conditions (SE) and (D), we have the AHK representation \eqref{eq:AHK_representation}: $f\left(X_{\bm{i}}\right) = f\left(\tau\left( \{ U_{\bm{i} \odot \bm{e}} \}_{\bm{e} \in \{0,1\}^K \setminus \{\bm{0}\}} \right)\right)$. Letting $g\Bigl( \{ U_{\bm{i} \odot \bm{e}} \}_{\bm{e} \in \{0,1\}^K \setminus \{\bm{0}\}} \Bigr) =f\left(\tau\left( \{ U_{\bm{i} \odot \bm{e}} \}_{\bm{e} \in \{0,1\}^K \setminus \{\bm{0}\}} \right) \right) - \mathbb E\left[f\left(\tau\left( \{ U_{\bm{i} \odot \bm{e}} \}_{\bm{e} \in \{0,1\}^K \setminus \{\bm{0}\}} \right) \right)\right]$, we have
    \begin{align*}
        &{\rm Cov}(f(X_{\bm 1}),f(X_{\bm{2}-\bm{e}_k})) =\mathbb{E}\left[g\left( \{ U_{\bm{1} \odot \bm{e}} \}_{\bm{e} \in \{0,1\}^K \setminus \{\bm{0}\}} \right)g\left( \{ U_{{(\bm{2}-\bm{e}_k}) \odot \bm{e}} \}_{\bm{e} \in \{0,1\}^K \setminus \{\bm{0}\}} \right)'\right] \\
         =& \mathbb{E}\left[\mathbb{E}\left[g\left( \{ U_{\bm{1} \odot \bm{e}} \}_{\bm{e} \in \{0,1\}^K \setminus \{\bm{0}\}} \right)g\left( \{ U_{{(\bm{2}-\bm{e}_k}) \odot \bm{e}} \}_{\bm{e} \in \{0,1\}^K \setminus \{\bm{0}\}} \right)'|U_{\bm{e}_k} \right]\right] \\
         =&  \mathbb{E}\left[\mathbb{E}\left[g\left( \{ U_{\bm{1} \odot \bm{e}} \}_{\bm{e} \in \{0,1\}^K \setminus \{\bm{0}\}} \right)|U_{\bm{e}_k}\right] \mathbb E\left[g\left( \{ U_{{(\bm{2}-\bm{e}_k}) \odot \bm{e}} \}_{\bm{e} \in \{0,1\}^K \setminus \{\bm{0}\}} \right) |U_{\bm{e}_k}\right]'\right] \\
         =&  {\rm Var} \left( \mathbb{E}\left[g\left( \{ U_{\bm{1} \odot \bm{e}} \}_{\bm{e} \in \{0,1\}^K \setminus \{\bm{0}\}} \right)|U_{\bm{e}_k}\right] \right) = {\rm Var} \left( \mathbb{E}\left(f(X)|U_{\bm{e}_k}\right) \right)
    \end{align*}
    where the third equality follows from Condition (D), and the fourth because, conditional on $U_{\bm{e}_k}$, the distribution of $g\left( \{ U_{\bm{i} \odot \bm{e}} \}_{\bm{e} \in \{0,1\}^K \setminus \{\bm{0}\}} \right)$ does not depend on $i$, which is also why we suppress the generic index $\bm i$ in the last equality. Since $k$ is arbitrary,  ${\rm Cov}(f(X_{\bm 1}),f(X_{\bm 2-\bm e_k}))={\rm Var}(\mathbb E[f(X)|U_{\bm e_k}])$ for every $k=1,\ldots,K$, and the last equality of the lemma follows from $n/N_k=\mu_k+o(1)$ and the boundedness of these covariances, which is implied by the Cauchy--Schwarz inequality and the moment condition.
\end{proof}

\vskip 0.15in

\bibliographystyle{ecta}
\bibliography{bib}

\end{document}